\documentclass[pra,twocolumn,superscriptaddress,letterpaper]{revtex4-2}

\usepackage{nicefrac}
\usepackage{braket}
\usepackage{xcolor}
\usepackage{graphicx}
\usepackage[colorlinks=true,urlcolor=blue,citecolor=blue,linkcolor=blue]{hyperref}

\usepackage{amssymb,amsmath,amsfonts}
\usepackage{amsthm}
\newtheorem{lemma}{Lemma}
\newtheorem{theorem}{Theorem}

\usepackage{enumerate}

\usepackage{booktabs}
\usepackage{multirow}

\DeclareMathAlphabet\mathbfcal{OMS}{cmsy}{b}{n}

\renewcommand{\vec}[1] {\boldsymbol{#1} }
\newcommand{\veccal}[1] {\boldsymbol{\mathcal{#1}} }

\newcommand{\trace}[1] {\mathrm{tr}\!\left( #1  \right) }
\newcommand{\rank}[1] {\mathrm{rank}\!\left( #1  \right) }

\begin{document}

\title{
A Robust Framework for the Comprehensive Characterisation of Parametric Down-Conversion Sources
}

\author{Denis A. Kopylov}
\email{denis.kopylov@uni-paderborn.de}
\author{Michael Stefszky}
\author{Christine Silberhorn}

\affiliation{Paderborn University, Integrated Quantum Optics Group, Institute for Photonic Quantum Systems (PhoQS), Warburger Stra\ss{}e 100, 33098 Paderborn, Germany}

\date{\today}
\begin{abstract}
\noindent
Reliable verification of squeezed-light sources is essential for photonic quantum technologies.
High-dimensional quantum optical networks demand single-mode, pure squeezed states which, in principle, can be generated through careful engineering of parametric down-conversion-based (PDC) sources.
However, due to the imperfections of optical components, the generated state is never pure and commonly contains internal mode structure.
Existing characterisation techniques, based on either joint spectral intensity or normalized second-order correlation functions, work under the assumption that the probed state is pure, and therefore can result in incomplete or misleading assessment of the source quality.
Here, we develop a device-agnostic characterisation framework for PDC-based sources with distinguishable signal and idler fields. 
Treating the source as a black box producing a mixed multimode Gaussian state, we find a lower bound for the state purity, which can be experimentally determined via mode-unresolved measurements.
Our approach enables a holistic assessment of realistic squeezed-light source quality without requiring a detailed model of the nonlinear device and its surrounding optical network, providing practical benchmarks for optimizing sources for state-of-the-art quantum applications.
\end{abstract}

\maketitle


\section{ Introduction }

Pulsed squeezed-light sources are a core component of modern photonic quantum technologies.
As a key resource for numerous quantum protocols~\cite{Weedbrook_2012,Andersen2016,Hamilton2017}, squeezed light provides a versatile platform for the experimental preparation of non-Gaussian states, including single photons~\cite{Mosley_2008}, cat states~\cite{Sychev_2017,Kan2021}, and GKP states~\cite{Larsen_2025,Konno2024,Erkilik2026}.
State-of-the-art squeezed-light sources are based on parametric down-conversion (PDC) or four-wave mixing, with the highest performance typically achieved using PDC~\cite{luo2026,Mehmet2016,BarconsPlanas2026}.
PDC-based squeezed-light sources have been successfully employed in proof-of-principle demonstrations of quantum computation~\cite{Larsen2019,Yu_2026,yokoyama2026,stefszky2026}, communication~\cite{Madsen2012}, and metrology~\cite{Labont__2024,Aasi2013}.

Any quantum optical experiment consists of quantum light sources, an optical circuit and a detection system (also often coined generation, manipulation and detection). 
As optical circuits are scaled up, they become increasingly sensitive to input state quality.
While some imperfections of light sources can be neglected in few-channel systems operating with a low number of photons, they ultimately limit performance in highly demanding, large-scale networks.
Understanding how these imperfections affect the generated quantum state is therefore essential for accurately assessing source performance and informing improved designs.
Verification schemes for quantum state preparation must provide an accurate and reliable characterisation that necessarily accounts for all relevant source characteristics---a topic that has, to date, remained comparatively underexplored.
As systems advance towards real-world applications, the development of such comprehensive characterisation schemes becomes critical.

A reliable characterisation technique should be grounded in the underlying physical processes. 
For pulsed PDC-based sources, this requires accounting for several key properties of the generated light. 
Its quantum state is squeezed (i.e., has an undefined number of photons), intrinsically multimode---even under near-single-mode operation---and generally mixed.
We discuss the implications of each of these properties in turn:

First, the PDC produces squeezed light which, by definition, consists of an undefined number of photons. The weak squeezed states are commonly analysed within the two-photon approximation, which is invalid for bright states and which becomes an increasingly critical issue as higher degrees of squeezing are demanded. Thus, an accurate characterisation technique has to be based on the squeezed states and must be valid both in the low-gain, two-photon regime, and as the sources are driven toward the high-gain regime. 

Second, the majority of existing quantum protocols are developed for single-mode light~\cite{Braunstein_2005_RMP,Weedbrook_2012,Walschaers_2021,Hahn_2022}.
However, under general conditions, PDC light is intrinsically multimode~\cite{Fabre_Treps_2020_review,Raymer_2020}: The fabrication of single-mode sources requires careful engineering and even then, is only ever approximated by a perfect single-mode source~\cite{luo2026,Pickston21}. Therefore, experimental quantification of the mode structure is necessary for a reliable source characterisation and is required to understand and compensate for the impact of these imperfections.

Third, successful realization of quantum optical protocols requires complete knowledge of the quantum state entering the optical circuits, which will be influenced by optical elements that are present between the PDC crystal and the optical system.
These additional elements, such as mirrors, lenses, waveguide couplers, pinholes, and bandpass filters, can dramatically alter the state's properties, affecting both modal structure and purity \cite{MeyerScott2017,Blay_2017,Christ2014}. 
Importantly, for integrated schemes with PDC-based sources, these linear optical systems can be an inherent component of the integrated optical system~\cite{Schuhmann_2024,Williams_2025,Wang_2026} and the squeezed state exiting the nonlinear stage may not be accessible at all. Thus, these additional linear optical elements should be considered as part of the quantum source.

Fourth, experimental optical components introduce losses and imperfections that cannot be characterized with perfect accuracy.
The effective refractive indices and poling structure of nonlinear waveguides can vary from waveguide to waveguide due to fabrication imperfections, which will impact the spectral properties of the generated light~\cite{Pelc2011,Santandrea2019}.
The modal structure can be altered by waveguide loss~\cite{Kopylov_2025,Kopylov_2025_HOM}.
Finally, undesired optical processes such as cascaded up-conversion~\cite{Rasputnyi_2022}, cross- and self- and phase-modulation~\cite{Triginer_2020, Thekkadath_2024}, or photo-refraction~\cite{Lange2026}, can be significant under strong pumping.
These factors influence the resulting state and the inclusion of all of the experimental imperfections in theoretical models is unfeasible.

In summary, experimentally generated states are always multimode, mixed squeezed states and any reliable treatment of these states necessarily needs to treat them as such.
They are impacted by experimental imperfections and therefore, we need a device-agnostic characterisation technique that operates with a PDC-based source as a black box producing a multimode, mixed squeezed state.

In the following, we focus on PDC-based sources that generate states with distinguishable signal and idler fields.
These signal and idler fields can have different polarizations (produced via type-II PDC) or non-overlapping frequency spectra (produced via non-degenerate type-I or type-O PDC).
Currently, two experimental methods are typically used for characterisation of their mode structure:
The first one is based on the Schmidt decomposition of the pure two-photon state fully described by the two-photon amplitude (TPA)~\cite{Law_2000,Mikhailova_2008}. 
The second method requires normalized second-order correlation measurements via slow detectors~\cite{Christ_2011}.
Both methods are widely used in experiments~\cite{Kim_2005,Wasilewski_2006_OL,Eckstein_2011,Sh_Iskhakov_2012,Harder_2013,P_rez_2014,Zielnicki_2018,Xin_2022,Wang_2026,Faleo_2026,Tischler_2015,Ansari_2018_PRL,MacLean_2019,Davis_2020,Thekkadath_2022}; however, these methods are based on assumptions that are not valid in the presence of non-negligible loss or as these systems are driven at high-gain.
The Schmidt decomposition based method requires measurement of the TPA, the measurement of which require sophisticated interferometric schemes~\cite{Tischler_2015,Ansari_2018_PRL,MacLean_2019,Davis_2020,Thekkadath_2022}.
To bypass this issue, the joint spectral intensity (JSI) is typically measured instead. However, the JSI is phase insensitive and as such the TPA phase remains unknown.
In order to infer the JSA for mode characterisation, a flat-phase assumption is often used~\cite{Harder_2013,Zielnicki_2018,Graffitti_2018}.
However, the flat-phase assumption is typically unjustified and will lead to an unreliable characterisation~\cite{Davis_2020}.
Moreover, the two-photon-approximation is no longer valid as the system is driven to produce appreciable amounts of squeezing; therefore, the validity of the first method for applications requiring squeezed states is dubious.
The second method is valid for both weak and bright squeezed states, but it assumes the state under study is pure~\cite{Christ_2011}.
Based on the Bloch-Messiah reduction~\cite{Braunstein_2005}, which provides a solid framework for mode analysis of pure multimode squeezed states even in the presence of strong pump fields~\cite{Wasilewski_2006,Christ_2011,Gouzien_2020,Horoshko_2021_PRA,Quesada_2022,Kopylov_2025,Scharwald_2026}, this method predicts equal signal and idler normalized second-order correlation functions.
However, the Bloch-Messiah reduction is not valid for mixed states~\cite{Kopylov_2025}; external spectral filtering~\cite{Sh_Iskhakov_2012} and optical losses~\cite{Kopylov_2025_HOM} modify the normalized second-order correlation functions.
In addition, recent experiments show that the normalized second-order correlation functions for signal and idler fields can differ~\cite{Wang_2026}.
These factors indicate that the pure-state assumption is very strong and does not hold for experimental mixed states.
Both the aforementioned methods, in their current forms, rely on pure states, and therefore their validity for mixed states should be reexamined.

In this work, we study the mode structure and correlation properties of multimode mixed states generated from PDC-based sources that produce distinguishable signal and idler fields.
Using the framework of Gaussian states, we reveal the limitations of existing characterization methods based on both second-order correlation functions and joint-spectral intensity.
Finally, we develop new, experimentally measurable quality metrics for mixed squeezed states with signal-idler distinguishability. Our framework overcomes the limitations of previous schemes and provides a new, robust, device-agnostic characterization technique for PDC-based sources, for both weak and highly squeezed states.

The paper is organized as follows: 
In Sec.~\ref{sec_system_under_study}, we specify the properties of the PDC-based sources under study and briefly present the main results.
The core of the developed formalism is a framework of multimode Gaussian states, and, in Sec.~\ref{sec_gtts_states}, we apply it to mixed and pure GSID-sources and analyse their mode structure and correlation properties.
These results build the theoretical background for the state characterisation technique, presented in Sec.~\ref{sec_characterisation}.
In Sec.~\ref{sec_examples}, we demonstrate its usage on a numerically simulated realistic PDC-based source.
Section~\ref{sec_discussion} discusses the differences with existing methods and presents possible extensions.

\section{ System under study and main results }
\label{sec_system_under_study}

    \begin{figure}
    \centering
    \includegraphics[width=1.\linewidth]{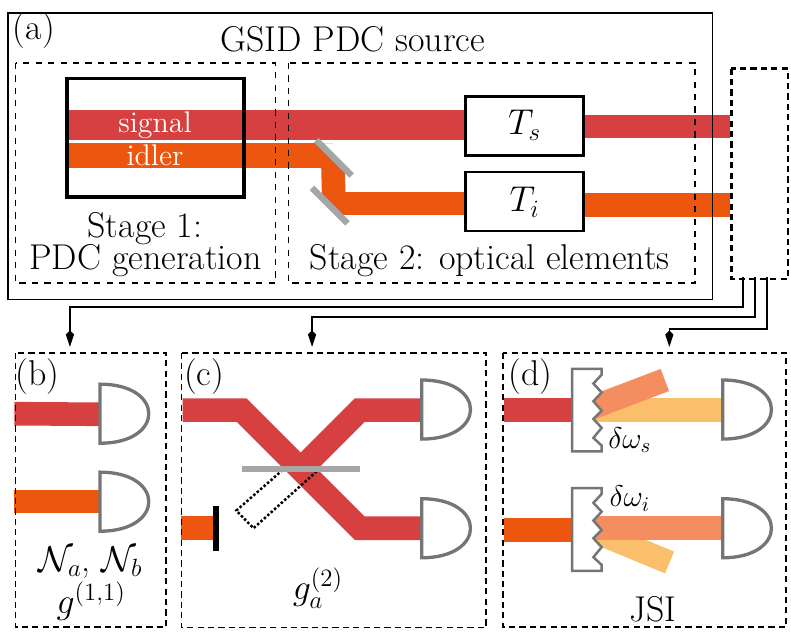}
    \caption{ (a) The scheme of a GSID-source. (b) Detection scheme for mean photon number and $g^{(11)}$ measurements. (c) The detection scheme for $g_a^{(2)}$ measurement. (d) Frequency resolved measurement: spectrum and JSI.  }
    \label{fig_scheme_1}
    \end{figure}

\subsection{Generalized Signal-Idler Distinguishable Source}
\label{sub_gtts_source}

        We define a PDC-based source as an optical system that consists of two principal subsystems: generation and preparation stages, shown in Fig.~\ref{fig_scheme_1}(a).
        The generation stage consist of a nonlinear medium where the signal and idler fields are generated via PDC. 
        The second stage---the preparation one---consists of a linear optical system that modifies and prepares the generated PDC light for its use in the experiment. 
        For example this stage might include spectral filtering to shape the generated state.
        
        In this work, we focus on PDC-based sources that satisfies the following requirements:
                 
        \textbf{R1}.
            We have a pulsed PDC process, which can be described via the interaction Hamiltonian 
            \begin{multline}
                \hat{G}_{int}(z) \propto \int d \vec{k}_\perp d \vec{k}^\prime_\perp  d \omega d \omega^\prime f(\omega,\omega^\prime, \vec{k}_\perp,\vec{k}^\prime_\perp,z) \\ \times \hat{a}^\dagger(\omega, \vec{k}_\perp)\hat{b}^\dagger(\omega^\prime, \vec{k}^\prime_\perp) + h.c.,
                \label{eq_hamiltonian}
            \end{multline}
            where the creation operators $\hat{a}^\dagger(\omega, \vec{k}_\perp)$ and $\hat{b}^\dagger(\omega, \vec{k}_\perp)$ correspond to distinguishable optical modes, i.e.,  the commutation relation $[\hat{a}(\omega, \vec{k}_\perp), \hat{b}^\dagger(\omega, \vec{k}_\perp)]=0$.
            This interaction Hamiltonian does not necessarily require type-II phase-matching, but also covers non-degenerate or non-collinear type-I phase-matching. 
            Further, the subsystems $a$ and $b$ (described by the operators $\hat{a}$ and $\hat{b}$) are labeled signal and idler, respectively.
            The coupling function $ f(\omega,\omega^\prime, \vec{k}_\perp,\vec{k}^\prime_\perp,z) $ depends on the pump properties (pump power and temporal pulse profile), the  properties of the nonlinear medium (linear dispersion and effective nonlinear susceptibility). 
            Note that the classical undepleted pump approximation is used.

        \textbf{R2}.
            We assume that the propagation loss of the nonlinear medium during the PDC generation can be described as a linear loss.
            The signal and idler subsystems are considered to have their own independent vacuum or thermal environment.
            This assumption guarantees the absence of linear interference and linear coupling between the signal and idler fields during PDC generation~\cite{Kopylov_2025_HOM}.

        \textbf{R3}. 
            The initial state of signal and idler fields is vacuum or thermal.

        \textbf{R4}. 
            The preparation stage of a source consists only of linear optical elements.
            These linear optical elements act on each of the subsystems independently.

        \textbf{R5}.
            Both output channels of the source produce fields in a single spatial mode.
            PDC in single mode waveguides guarantees this requirement (the interaction Hamiltonian does not depend on spatial degrees of freedom).
            For bulk crystals, this condition requires that the preparation stage should contain additional spatial filtering, e.g., via coupling to single-mode fibers.
            In Sec.~\ref{ssec_spatial_degrees}, we discuss PDC sources beyond this requirement.

        Requirements R1, R2, and R4 mean that our source acts as a multimode Gaussian channel without displacement.
        Together with requirements R2 and R3, they guarantee that the output state is an undisplaced Gaussian state.
        Requirements R1 and R2 provide distinguishability of the signal and idler waves during PDC generation (see Refs.~\cite{Kopylov_2025_HOM,Kopylov_2025_APL}), and, together with R4, ensure that the source produces separate signal and idler waves. 
        In other words, the signal and idler fields are never mixed.
        Requirement R5 allows us to omit spatial modes from further consideration. 
        
        Assuming the aforementioned requirements are met, our PDC-based source produces a frequency-multimode, undisplaced, Gaussian quantum state $\hat{\rho}$ with distinguishable signal and idler fields.
        The generated light is fully determined by the nonzero correlation functions 
        \begin{equation}
                \braket{\hat{a}^\dagger(\omega)\hat{a}(\omega^\prime)}, \braket{\hat{b}^\dagger(\omega)\hat{b}(\omega^\prime)} ~ \text{and} ~ \braket{\hat{a}(\omega)\hat{b}(\omega^\prime)},
                \label{eq_corr_funcitons}
        \end{equation}
        while both $\braket{\hat{a}(\omega)\hat{a}(\omega^\prime)}=\braket{\hat{b}(\omega)\hat{b}(\omega^\prime)}=\braket{\hat{a}^\dagger(\omega)\hat{b}(\omega^\prime)}\equiv 0$ and $\braket{\hat{a}(\omega)}=\braket{\hat{b}(\omega)}\equiv0$. 
        In what follows, we call this state \textit{a generalized signal-idler distinguishable PDC state} or \textit{GSID-state} for short.
        The PDC-based source producing the GSID-state is further called a \textit{GSID-source}.

\subsection{Intensity measurements}
\label{sub_int_measurem}

        In this work, we consider exclusively intensity measurements performed with slow detectors.
        Here, slow means that the detection time is much longer than the pulse duration which implies that internal field fluctuations cannot be resolved. 
        However, the detectors are fast enough to measure the energy of each pulse and resolve pulse-to-pulse fluctuations.
        In other words, the internal mode structure remains hidden; therefore, we call such measurements \textit{mode-unresolved measurements}.
        Fig~\ref{fig_scheme_1}(b,c) presents the mode-unresolved detection schemes considered in this work.
        These measurements give us the average number of photons and the mode-unresolved second-order correlation functions.

        Placing monochromators before the detectors allows us to perform spectral measurements (see Fig~\ref{fig_scheme_1}(d)).
        In this case, we can measure the spectra of signal and idler modes as well as the joint spectral intensity. 
        Such measurements are further called \textit{frequency-resolved measurements}.

        In this paper, we characterize the GSID-source with both mode-unresolved and frequency-resolved intensity measurements.
        Here we briefly present the main results; the details for the characterisation procedure are given in Sec.~\ref{sec_characterisation}.

\subsection{Main results}
\label{sub_main_results}

    \textbf{Statement 1}.
        The mixed state of a pulsed GSID-source provides three different bases: The Mercer-Wolf basis of the signal subsystem, the Mercer-Wolf basis of the idler subsystem, and the basis of pairwise signal-idler correlated modes (called further the C-basis).
        They are defined independently from correlation functions $\braket{\hat{a}^\dagger(\omega)\hat{a}(\omega^\prime)}$, $\braket{\hat{b}^\dagger(\omega)\hat{b}(\omega^\prime)}$ and  $\braket{\hat{a}(\omega)\hat{b}(\omega^\prime)}$, respectively.
        These bases are characterized by their own effective number of modes: the number of modes in the signal subsystem $\mu_a$, the number of modes in the idler subsystem $\mu_b$, and the number of correlated modes  between the signal and idler subsystems $K$, respectively.
        The parameters $\mu_a$, $\mu_b$, and $K$ are defined independently and therefore can take different values.
        For details, see Sec.~\ref{sub_properties_of_mixed_gsid_state_}.

        \textbf{Statement 2}.
            The number of modes in the signal $\mu_a$ and idler $\mu_b$ can be measured via the normalized second-order correlation functions $g^{(2)}_{a,b}$ and the parameters $\mu_a$ and $\mu_b$ are not necessarily equal.
            For details, see Sec.~\ref{ssub_meas_corr_modes}.
            
        \textbf{Statement 3}.
            The number of correlated modes $K$ cannot be reliably determined from the joint spectral intensity  $ J(\omega,\omega^\prime) = |\braket{\hat{a}(\omega)\hat{b}(\omega^\prime)}|^2$ because
            the JSI disregards the phase information within $\braket{\hat{a}(\omega)\hat{b}(\omega^\prime)}$.
            We show that the decomposition of $\sqrt{J(\omega,\omega^\prime)}$ provides the value $K_{\sqrt{J}}$, which can only be used as a lower bound for $K$:
            \begin{equation}
               K_{\sqrt{J}} \leq K.
            \end{equation}
            For details, see Sec.~\ref{ssec_JSI}.

        \textbf{Statement 4}.
            The degree of correlation between signal and idler can be quantified using $\xi$.
            It obeys                        
            \begin{equation}
                0 \leq \xi \leq 1,
            \end{equation}
            and $\xi=0$ indicates the absence of the signal-idler, while $ \xi = 1$ for maximally correlated states.
            The degree of correlation $\xi$ can be experimentally determined from the $g^{(1,1)}$ measurements.
            For details, see Sec.~\ref{sub_properties_of_mixed_gsid_state_}.

        \textbf{Statement 5}.
            For pure GSID-states, the number of photons in signal $\mathcal{N}_{a}$ and idler $\mathcal{N}_{b}$ subsystems are equal $\mathcal{N}_a = \mathcal{N}_b$ and the number of modes coincide $\mu_a = \mu_b$.
            The degree of correlations $\xi=1$. 
            However, $K \neq \mu$, and $K \approx \mu$ only for states with low photon number. 
            For details, see Sec.~\ref{sec_pure_gtts}.

        \textbf{Statement 6}.
            The values $\mu_{a}$, $\mu_{b}$, $K$ and $\xi$ together with the number of photons in signal and idler subsystems $\mathcal{N}_{a,b}$  define the lower bounds for the state purity.
            This purity bound can be used as an experimental measure of the GSID-source quality.
            For details, see Sec.~\ref{sub_purity_bounds}.

\vspace{0.5cm}

In order to obtain these results, let us start first with the analysis of GSID-states in the formalism of Gaussian states.

\section{ Properties of GSID-states }
\label{sec_gtts_states}

\subsection{ Notation and preliminaries }
\label{sub_notation}
    In this section, we briefly present the notation and the main mathematical results used for the description of  Gaussian states.

    Let us consider a complex $N\times M$ matrix $\vec{A}$ (for short, $\vec{A} \in \mathbb{C}^{N\times M}$).
    The element of $\vec{A}$ at the $i$-th row and $j$-the column is denoted as $[\vec{A}]_{ij}$.
    The matrix $\vec{A}^* \in \mathbb{C}^{N\times M}$ is the conjugated matrix of $\vec{A}$ with $[\vec{A}^*]_{ij} = ([\vec{A}]_{ij})^*$.
    The matrix $\vec{A}^T  \in \mathbb{C}^{M\times N}$ is the transposed matrix of $\vec{A}$ with $[\vec{A}^T]_{ij} = [\vec{A}]_{ji}$.

    The matrix $\vec{A}^H \in \mathbb{C}^{M\times N}$ is the conjugate-transposed matrix of $\vec{A}$ with $[\vec{A}^H]_{ij} = ([\vec{A}]_{ji})^*$.
    Note that throughout the text we use two notations for Hermitian-conjugation: the notation ${[.]}^\dagger$ refers to the quantum mechanical operators while ${[.]}^H$ refers to the matrices.  

    The matrix $\vec{B} \in \mathbb{C}^{N\times N}$ is Hermitian if $\vec{B}^H = \vec{B}$.
    The notation $\vec{B} \succ 0$ ($\vec{B} \succeq 0$) denotes that the Hermitian matrix $\vec{B}$ is positive definite (positive semi-definite).

    For the matrix $\vec{C} \in \mathbb{C}^{N\times N}$, the trace is defined as $\trace{\vec{C}} = \sum_{i=1}^N [\vec{C}]_{ii}$.
    For the matrix $\vec{A} \in \mathbb{C}^{N\times M}$, the following identity holds
    \begin{equation}
        \sum_{i=1}^N \sum_{j=1}^M |[\vec{A}]_{ij}|^2 = \trace{ \vec{A}^H \vec{A} } = \trace{\vec{A}  \vec{A}^H }.
        \label{eq_trace_indent}
    \end{equation}
    Appendix~\ref{app_math} contains the complete mathematical results and derivations used in this paper.

\subsection{ Uncertainty relation for GSID-state }
    \label{subsec_uncert_rel}
    We begin our analysis with the GSID-state defined on a discrete monochromatic basis. 
    The transition to continuous frequencies is discussed in Sec.~\ref{sec_characterisation}.

    Commonly, multimode Gaussian states are described in terms of the Wigner covariance matrix, defined by symmetric second-order quadrature moments. 
    This description is useful for the studies of Gaussian states~\cite{Simon_1994,Vogel_Welsch_book,Weedbrook_2012}, when homodyne measurements are considered.
    However, in the case of intensity-correlation measurements, the description in terms of symmetric second-order moments is implicit: the measurement results are determined by the normal-ordered field correlators.
    Thus, in this manuscript, we formulate the Gaussian formalism in terms of normal-ordered second-order correlation matrices (for details, see Appendix~\ref{app_Gaussian_states}).

    Let us consider a $(N+M)$-mode bosonic system  consisting of two subsystems: signal ($a$) and idler ($b$) on two discrete uniform frequency spaces $(\omega_1, \dots, \omega_N)$ and $(\omega_1, \dots, \omega_M)$, respectively.
    The subsystems $a$ and $b$ are described then by $N$ and $M$ monochromatic operators $\hat{a}_i\equiv\hat{a}(\omega_i)$ and $\hat{b}_i\equiv\hat{b}(\omega_i)$, respectively.
    These operators obey the bosonic commutation relations $[\hat{a}_i, \hat{a}^\dagger_j] = [\hat{b}_i, \hat{b}^\dagger_j] = \delta_{ij}$ and  $[\hat{a}_i, \hat{b}^\dagger_j] = 0$.
    In compact form, the operators can be written as a vector $\hat{\vec{a}} \equiv  (\hat{a}_1, \dots, \hat{a}_N)^T$ and $\hat{\vec{b}} \equiv  (\hat{b}_1, \dots, \hat{b}_M)^T$.

    In the discrete monochromatic basis, the GSID-state $\hat{\rho}$ is defined by two block-matrices
    \begin{equation}
        \veccal{D} =     
        \begin{pmatrix}
            \veccal{D}_a  & \mathbf{0}  \\
            \mathbf{0}  & \veccal{D}_b 
        \end{pmatrix}   ~ ~ \text{and} ~ ~
        \veccal{C}  =     
        \begin{pmatrix}
            \mathbf{0}  & \veccal{C}_{ab}  \\
            \veccal{C}^T_{ab}  & \mathbf{0} 
        \end{pmatrix},
        \label{eq_solution_block}
    \end{equation}
    with $\veccal{D}_a \in \mathbb{C}^{N\times N}$, $\veccal{D}_b \in \mathbb{C}^{M\times M}$ and $\veccal{C}_{ab} \in \mathbb{C}^{N\times M}$.
    Their matrix elements are $[\veccal{D}_a]_{ij} = \braket{\hat{a}^\dagger_i \hat{a}_j}$, $[\veccal{D}_b]_{ij} = \braket{\hat{b}^\dagger_i \hat{b}_j}$ and $[\veccal{C}_{ab}]_{ij} = \braket{\hat{a}_i \hat{b}_j}$.
    These matrices are the discretized versions of correlations Eq.~\eqref{eq_corr_funcitons}.
    Note that the $[\veccal{D}^T]_{mn} = [\veccal{D}]_{nm} = \braket{\hat{a}^\dagger_n \hat{a}_m} \neq \braket{\hat{a}_n \hat{a}^\dagger_m}$.

    In order to define a valid quantum state $\hat{\rho}$, the matrices $\veccal{D}_a$, $\veccal{D}_b$ and $\veccal{C}_{ab}$ should satisfy the following conditions
    \begin{enumerate}[(I)]
        \item Matrices $\veccal{D}_{a,b}$ are Hermitian: $\veccal{D}_{a,b} = \veccal{D}_{a,b}^H $ and $\veccal{D}_{a,b}^* = \veccal{D}_{a,b}^T $.  
        \item $\veccal{D}_{a,b}$ are positive semi-definite matrices ($\veccal{D}_{a,b} \succeq 0$).
        \item The uncertainty relation is (see Appendix~\ref{app_Gaussian_states})
        \begin{equation}
            \begin{cases}
             \veccal{D}^T_a - \veccal{C}_{ab}(\veccal{D}_b  + 1)^{-1}\veccal{C}^H_{ab} & \succeq 0, \\
                     \veccal{D}_b - \veccal{C}_{ab}^H (\veccal{D}_a^T  + 1)^{-1}\veccal{C}_{ab} & \succeq 0.
            \end{cases}
            \label{eq_uncertainty_main}
        \end{equation}
     
    \end{enumerate}
    The conditions (I-III) are necessary and sufficient for matrices $\veccal{D}_a$, $\veccal{D}_b$ and $\veccal{C}_{ab}$ to define a valid quantum GSID-state.
    
    The correlation matrices~\eqref{eq_solution_block} imply the following structure for GSID-states: each subsystem quantum state is a general multimode thermal state. This follows from tracing out the other subsystem, e.g., the idler. The reduced quantum state of the signal field is then described only by the matrix $\veccal{D}_{a}$.
    Similarly, the state of the idler system is described by $\veccal{D}_{b}$.
    Therefore, any GSID-state consists of two multimode thermal states. The matrix $\veccal{C}_{ab}$ stores all correlation properties between the signal and idler subsystems, and its analysis is essential for signal-idler correlation studies. 
    In the absence of correlations, when $\veccal{C}_{ab}=0$, the GSID-state is just a composition of two thermal states and such a GSID-state will henceforth be referred to as a \textit{thermal} or \textit{uncorrelated GSID-state}.

\subsection{Optical modes}
\label{sub_mode_trans}

    Thus far, we described the state $\hat{\rho}$ using matrices $\veccal{D}_a$, $\veccal{D}_b$ and $\veccal{C}_{ab}$ given in the monochromatic basis.
    However, the same state $\hat{\rho}$ can be written in any arbitrary basis.
    Given a basis transformation of the form 
    \begin{equation}
        \vec{\hat{A}} = U_A \vec{\hat{a}} ~ ~ \text{and} ~ ~ \vec{\hat{B}} = U_B \vec{\hat{b}}, 
        \label{eq_mode_transform}
    \end{equation}
    where $U_A$ and $U_B$ are $N\times N$ and $M\times M$ unitary matrices, respectively, we obtain new annihilation operators $\vec{\hat{A}} \equiv  (\hat{A}_1, \dots, \hat{A}_N)^T$ and $\vec{\hat{B}} \equiv  (\hat{B}_1, \dots, \hat{B}_M)^T$.
    These new operators also satisfy the bosonic commutation relations $[\hat{A}_n, \hat{A}_m^\dagger]=\delta_{nm} $, $[\hat{B}_n, \hat{B}_m^\dagger]=\delta_{nm} $ and $[\hat{A}_n, \hat{B}_m^\dagger]=0 $.

    Each row of the matrices $U_A$ and $U_B$ is interpreted as a new optical mode for the subsystems $a$ and $b$, respectively.  
    As long as the modes in this new basis are built as a superposition of monochromatic ones, they are commonly called broadband modes, or specto-temporal modes~\cite{Brecht_2015,Raymer_2020,Fabre_Treps_2020_review}.  

    Under the mode-transformation~\eqref{eq_mode_transform}, the correlation matrices in this new basis are
     \begin{multline}
        \veccal{D}_A = U_A^* \veccal{D}_a U_A^T, ~ ~ \veccal{D}_B = U_B^* \veccal{D}_b U_B^T  \\  ~ \text{and} ~ ~ 
        \veccal{C}_{AB} = U_A \veccal{C}_{ab} U_B^T.
        \label{eq_correlation_transformation}
    \end{multline}
    Note that the transformation defined by~\eqref{eq_correlation_transformation} does not change the state $\hat{\rho}$: both the sets $\veccal{D}_a$, $\veccal{D}_b$, $\veccal{C}_{ab}$ and $\veccal{D}_A$, $\veccal{D}_B$, $\veccal{C}_{AB}$ describe the same quantum state $\hat{\rho}$.

    In the context of mode-unresolved detectors, the measurement outcome is insensitive to the internal structure of the light.
    Therefore, we can work with any basis, and the measurement outcome should be determined by the \textit{basis-independent quantities} of the correlation matrices, i.e., such quantities that are unchanged under the transformation~\eqref{eq_correlation_transformation}.
    On the one hand, this means that we do not need a basis transformation at all---the measurement results can be expressed using the correlation function in the initial monochromatic basis.
    On the other hand, as we will see further, some bases lend themselves to a simple physical interpretation and therefore are more convenient than others.

    In this work, in addition to the monochromatic basis, we consider three other bases: the Mercer-Wolf basis, the C-Basis, and the S-Basis.

\subsubsection{Mercer-Wolf basis }
\label{ssub_mw_basis} 
        Let us consider the basis that brings matrices $\veccal{D}_a$ and $\veccal{D}_b$ in diagonal form.
        In the case that $\veccal{D}_{a,b}$ are Hermitian and positive semidefinite, their eigendecomposition takes the form 
        $ 
             \veccal{D}_a = W_A \vec{\Lambda}_a W_A^H, ~ ~      \veccal{D}_b = W_B \vec{\Lambda}_b W_B^H,
        $ 
        where $W_A$ and $W_B$ are unitary matrices. 
        The diagonal matrices $\vec{\Lambda}_{a,b}$ are built from the eigenvalues of matrices $\veccal{D}_{a,b}$, respectively.
        In what follows, we always take the eigenvalues in decreasing order.         
        Introducing the new broadband modes as $\vec{\hat{A}} = W_A^T\vec{\hat{a}}$ and $\vec{\hat{B}} = W_B^T\vec{\hat{b}}$, we get that the correlation matrices are
        \begin{equation}
            \veccal{D}_A=\vec{\Lambda}_a, ~ \veccal{D}_B=\vec{\Lambda}_b,
             ~  \text{and} ~ ~  \veccal{C}_{AB} = W_A^T \veccal{C}_{ab} W_B . 
             \label{eq_MW_basis}
        \end{equation}
        This basis is known as a \textit{Mercer-Wolf (MW) basis}~\cite{Kopylov_2025,Kopylov_2025_HOM,Kopylov_2025_APL} or as the first-order coherence basis~\cite{Fabre_Treps_2020_review}.

        In the MW-basis, the matrices $\veccal{D}_A$ and $\veccal{D}_B$ are well defined and unique. 
        However, as soon as the MW-basis is defined independently for each of the subsystems $a$ and $b$, the matrix $\veccal{C}_{AB}$ is not necessarily diagonal.
        It means that correlations can be present between MW-modes with different indices.
        Moreover, the matrix $\veccal{C}_{AB}$ is not unique because of additional uncertainty caused by a degenerate eigenspectrum (see Appendix~\ref{app_uncertainty_of_mw_modes}). 
        This uncertainty does not play a further role; therefore, we do not focus on it. 

        The eigenvalues of the matrices $\veccal{D}_{a,b}$ are not changed under the transformation~\eqref{eq_correlation_transformation}; therefore, building an MW-basis from any arbitrary basis, we always get $\veccal{D}_{A,B}=\vec{\Lambda}_{a,b}$.
        As a result, the matrices $\vec{\Lambda}_{a,b}$ can be used for the definition of basis-independent quantities.

\subsubsection{ C-Basis   }
\label{ssub_c_basis}

        An alternative broadband basis is the basis that diagonalizes the correlation matrix $\veccal{C}_{ab}$.
        Indeed, the singular value decomposition (SVD) of the matrix $\veccal{C}_{ab} = S^H_A \vec{\Sigma}_C S_B^*$, produces two unitary matrices $S_A\in \mathbb{C}^{N\times N}$ and $S_B\in \mathbb{C}^{M\times M}$.
        The matrix $\vec{\Sigma}_C \in \mathbb{C}^{N\times M}$ is a diagonal rectangular matrix, built from singular values of $\veccal{C}_{ab}$: $\sigma_n = [\vec{\Sigma}_C]_{nn}$.

        Defining a new basis as $\vec{\hat{A}} = S_A \vec{\hat{a}}$ and $\vec{\hat{B}} = S_B \vec{\hat{b}}$, together with Eq.~\eqref{eq_correlation_transformation}, we get
        \begin{equation}
            \veccal{D}_A = S_A^* \veccal{D}_a S_A^T, ~ ~ \veccal{D}_B = S_B^* \veccal{D}_b S_B^T, ~ ~
            \veccal{C}_{AB} = \vec{\Sigma}_C.    
        \end{equation}
        In what follows, we shall refer to this basis as the \textit{C-basis}.
        
        The matrix $\veccal{C}_{AB}$ is diagonal, meaning that, in the C-basis, the signal and idler modes are pairwise correlated.
        However, the matrices $\veccal{D}_{A,B}$ are not necessarily diagonal, indicating additional internal correlations in each subsystem.

        In contrast to the MW-basis, the C-basis is not well defined for any quantum state: for the matrix $\veccal{C}_{ab}$, the modes $S_A$ and $S_B$ are well-defined only for nonzero singular values, i.e., if there exists correlation between them.
        For the thermal GSID-state (when $\veccal{C}_{ab}=0$), the C-basis is undefined.

        The matrix $\vec{\Sigma}_C$ is independent under the transformation given by~\eqref{eq_correlation_transformation}.
        Therefore, there are basis-independent quantities that can be built from the singular values of $\veccal{C}_{ab}$ and interpreted in terms of C-modes.

\subsubsection{ S-Basis   }
\label{ssub_s_basis}
        As we will see in Sec.~\ref{ssub_meas_corr_modes}, there is an additional basis of practical interest.
        Namely, the basis which diagonalizes the matrix 
        $\vec{G} = |\veccal{C}_{ab}|_\circ $, where $|.|_\circ$ is the Hadamard-modulus that denotes $[\vec{G}]_{ij} = | [\veccal{C}_{ab}]_{ij} |$.

        The SVD-decomposition of the real-valued matrix $\vec{G} = O^T_A \vec{\Sigma}_G O_B $ allows us to introduce two valid sets of modes $\vec{\hat{A}} = O_A \vec{\hat{a}}$ and $\vec{\hat{B}} = O_B \vec{\hat{b}}$.
        However, neither the correlation matrices $\veccal{D}_A = O_A \veccal{D}_a O_A^T$ and $\veccal{D}_B = O_B \veccal{D}_b O_B^T$, nor the matrix $\veccal{C}_{AB} = O_A \veccal{C}_{ab} O_B^T $ are necessarily diagonal in this basis.
        This comes from the fact that the SVD of a Hadamard product does not coincide with the SVD of its factors (for details, see e.g. Ref.~\cite{Ciaperoni_2024}). 

        Moreover, the singular values of matrices $\vec{G}$ and $\vec{\tilde{G}} = |U_A \veccal{C}_{ab} U_B^T|_\circ $ are not necessarily equal.
        In other words, these singular values are changed under the transformation~\eqref{eq_correlation_transformation} and therefore can not be used for the introduction of basis-independent properties.
        Because of this, we call the basis built from the matrix $\vec{G}$ a \textit{spurious basis} or \textit{S-basis}.

        The S-basis commonly arises when joint spectral intensity measurements are used for the experimental determination of the mode structure.
        However, as we show in Sec.~\ref{ssub_meas_corr_modes}, the S-basis does not provide accurate nor sufficient information for GSID-state characterisation.

\subsection{ Properties of mixed GSID-state }
\label{sub_properties_of_mixed_gsid_state_}
        Let us consider a GSID-state $\hat{\rho}$ defined by the matrices $\veccal{D}_a$, $\veccal{D}_b$ and $\veccal{C}_{ab}$.

\subsubsection{ Total number of photons of the subsystem }

    \textit{The total number of photons} is defined for each of the signal and idler subsystems independently as
    \begin{equation}
        \mathcal{N}_a = \trace{\veccal{D}_a}  ~ ~ \text{and} ~ ~ \mathcal{N}_b = \trace{\veccal{D}_b}.
        \label{eq_numphot_1}
    \end{equation}
    In other words, the total number of photons is the sum of the mean number of photons over all the modes, e.g., for the signal system $\mathcal{N}_a = \sum_i \braket{\hat{a}_i^\dagger\hat{a}_i^\dagger}$.
    Note that the total number of photons does not depend on the working basis, i.e., it is independent of the transformation~\eqref{eq_correlation_transformation}.    
    Indeed, $\trace{\veccal{D}_a} = \trace{U^H U\veccal{D}_a } = \trace{ U\veccal{D}_a U^H}$ for any unitary matrix $U$.

    In terms of MW-basis, 
    \begin{equation}
        \mathcal{N}_a = \trace{\vec{\Lambda}_a}  ~ ~ \text{and} ~ ~ \mathcal{N}_b = \trace{\vec{\Lambda}_b},
        \label{eq_numphot_2}
    \end{equation}
    i.e., the total number of photons are the sum of all the eigenvalues of matrices $\veccal{D}_a$ and $\veccal{D}_b$, respectively.    
    For conciseness, in what follows we refer to the total number of photons as simply, the number of photons.

\subsubsection{Number of modes in the subsystem}
\label{ssub_num_modes_subsyst}

Let us first consider the signal subsystem. 
    In the MW-basis, $\veccal{D}_A=\vec{\Lambda}_a$,
    and $\lambda^a_i = [\vec{\Lambda}_a]_{ii}$ is then interpreted as the number of photons in the $i$-th mode of the MW-basis.
    We can ask ourselves: how many modes are occupied in the MW-basis?
    Inspired by Ref.~\cite{Grobe_1994}, 
    let us first consider the relative number of photons in the $i$-th mode defined as  $p^{a}_i = \lambda^{a}_i/\trace{\vec{\Lambda}_{a}}$. 
    Then the values $p^{a}_i$ can be interpreted as a discrete distribution.
    The number of modes $\mu_a$ in the signal system can then be defined as the inverse of the averaged probability, namely as
    $[\sum_i (p^{a}_i)^2]^{-1}$.
    Similarly, we build the number of modes in the idler subsystem.

    As a result, \textit{the number of modes} for signal $\mu_a$ and idler $\mu_b$ subsystems are defined as
    \begin{equation}
        \mu_a  =  \dfrac{[\trace{\vec{\Lambda}_a}]^2}{\trace{\vec{\Lambda}_a^2}},  ~ ~ 
        \mu_b  =  \dfrac{[\trace{\vec{\Lambda}_b}]^2}{\trace{\vec{\Lambda}_b^2}}.
        \label{eq_nummod_1} 
    \end{equation}
    Using the trace-property Eq.~\eqref{eq_trace_indent}, we have
    \begin{equation}
        \sum_{ij}\big|[\veccal{D}_{a}]_{ij}\big|^2= \trace{\veccal{D}_{a}^H \veccal{D}_{a} } = \trace{\vec{\Lambda}_a^2} = \trace{\veccal{D}_{a}^2} ,
        \label{eq_deriv_bas_indep_N}
    \end{equation}
    and therefore the number of modes is a basis-independent characteristic of each of the subsystems and has the basis-independent form
    \begin{equation} 
        \mu_a  =  \dfrac{[\trace{\veccal{D}_a}]^2}{\trace{\veccal{D}_{a}^2}},  ~ ~ 
        \mu_b  =  \dfrac{[\trace{\veccal{D}_b}]^2}{\trace{\veccal{D}_{b}^2}}.
        \label{eq_nummod_2} 
    \end{equation}
    The number of modes is defined for each of the subsystems independently; therefore, the values $\mu_a$ and $\mu_b$ are not necessarily equal.

    Note that the definition of the number of modes differs from Ref.~\cite{Kopylov_2025}, where \textit{the number of occupied modes in an arbitrary basis} was defined.
    In that terminology, the number of modes $\mu_{a,b}$ is then the number of occupied modes in the MW-basis, and is minimal over all the possible bases~\cite{Kopylov_2025}.

\subsubsection{Total amount of correlations}

    Both the number of photons  $\mathcal{N}_{a,b}$ and the number of modes  $\mu_{a,b}$ are defined via the matrices $\veccal{D}_{a,b}$ and therefore are characteristics of the subsystems but not the joint system.
    As was mentioned in the Sec.~\ref{subsec_uncert_rel}, the correlations are determined by the matrix $\veccal{C}_{ab}$.
    In order to characterise the signal-idler correlations, we define \textit{the total amount of signal-idler correlations} as 
    \begin{equation}
        \mathcal{J} = \sum^N_{i}\sum^M_{j} |[\veccal{C}_{ab}]_{ij}|^2.
        \label{eq_J_def_1}
    \end{equation}
    With the use of Eq.~\eqref{eq_trace_indent}, we can write it as     
    \begin{equation}
        \mathcal{J} = \trace{\veccal{C}^H_{ab} \veccal{C}_{ab} },
        \label{eq_J_def_2}
    \end{equation}
    and therefore the total amount of correlations is a basis-independent characteristic of the GSID-state (see Eq.~\eqref{eq_deriv_bas_indep_N}).

\subsubsection{The number of correlated modes}
        The total amount of correlations does not describe how the correlations are distributed among the signal and idler modes. 
        Nevertheless, in the C-basis, the total amount of correlations takes the form
        \begin{equation}
            \mathcal{J} = \trace{\vec{\Sigma}_C^2} = \sum_i \sigma^2_i,
        \end{equation}
        where $\sigma_i = [\vec{\Sigma}_C]_{ii}$.
        The values $\sigma^2_i$ are then interpreted as the amount of correlations between the $i$-th signal mode and the $i$-th idler mode of the C-basis.
        In a similar way to the number of modes, we could ask ourselves, how many correlated modes exist in the C-basis?
        To answer it, we construct the averaged probabilities similarly to those in Sec.~\ref{ssub_num_modes_subsyst}.
        As a result, we get \textit{the number of correlated modes}         
        \begin{equation}
            K =  \dfrac{[\trace{\vec{\Sigma}_C^2}]^2}{\trace{\vec{\Sigma}_C^4}}
        \end{equation} 
        or, in the basis-independent form,
        \begin{equation}
            K =  \dfrac{[\trace{\veccal{C}^H_{ab}\veccal{C}_{ab} }]^2}{\trace{ (\veccal{C}^H_{ab}\veccal{C}_{ab})^2 }}.
        \end{equation} 

        Note that the number of correlated modes $K$ is defined via the correlation matrix $\veccal{C}_{ab} $, while the number of modes in each subsystem $\mu_{a,b}$ is defined from the matrices $\veccal{D}_{a,b}$. 
        The interpretation of the values $K$ and $\mu_{a,b}$ are given in different bases: namely, in the C-basis and the MW-basis, respectively.
        Therefore, these values have a different physical meaning. 
        Note that the uncertainty relation~\eqref{eq_uncertainty_main} allows both larger and smaller values of $K$ in comparison with the values of $\mu_a$ and/or $\mu_b$.
        We stress that \textit{the parameters $\mu_a$, $\mu_a$ and $K$ are independent and must be determined independently}.
        In Sec.~\ref{sec_examples_mw_basis}, we present examples of GSID-states that highlight that mode structure of signal and idler subsystem does not define correlation properties of GSID-staes. 

        An important remark should be made regarding the value
        \begin{equation}
            K_{\sqrt{J}} =  \dfrac{[\trace{\vec{G}^T_{ab}\vec{G}_{ab} }]^2}{\trace{ (\vec{G}^T_{ab}\vec{G}_{ab})^2 }},
        \end{equation} 
        obtained from the matrix $\vec{G}$ with elements $[\vec{G}]_{nm} = |[\veccal{C}_{ab}]_{nm}|$ (see Sec.~\ref{ssub_s_basis}).
        For arbitrary states, 
        \begin{equation}
            K_{\sqrt{J}} \leq K,
        \end{equation}
        (for proof, see Appendix~\ref{app_JSI_inequality} )
        which provides the lower bound of the number of correlated modes. 
        This inequality additionally emphasizes the difference between the C-basis and the S-basis.
        Therefore, the value $K_{\sqrt{J}}$ does not reliably represent the number correlated modes, but nevertheless can still be helpful for the characterization of the GSID-sources (see Sec.~\ref{sub_comments_bounds}).

\subsubsection{ Degree of signal-idler correlations  }

        One of the corollaries of the uncertainty condition is the following inequality (see Appendix~\ref{app_bipartitie_Gaussian} for details)
        \begin{equation}
            \left|\trace{\veccal{C}_{ab}^H \veccal{C}_{ab}}\right|^2 \leq \trace{\veccal{D}_a + \veccal{D}^2_a}\trace{\veccal{D}_b + \veccal{D}^2_b}.
            \label{eq_uncet_inequality}
        \end{equation}
        This inequality restricts the correlations between the fixed modes
        \begin{equation}
            \left|[\veccal{C}_{ab}]_{ij}\right|^4 \leq ({[\veccal{D}_a]_{ii} + [\veccal{D}_a]^2_{ii}})  ({[\veccal{D}_b]_{jj} + [\veccal{D}_b]_{jj}^2}).
        \label{eq_uncet_inequality_elements}
        \end{equation}
        Note that these inequalities are necessary but not sufficient to define a valid quantum GSID-state with the matrices $\veccal{D}_a$, $\veccal{D}_b$ and $\veccal{C}_{ab}$.

        Using inequality~\eqref{eq_uncet_inequality} together with Eqs.~\eqref{eq_numphot_1}, \eqref{eq_nummod_2}, \eqref{eq_J_def_2}, we get that the total amount of correlations is bounded as
        \begin{equation}
            \mathcal{J}^2 \leq 
            \left(\mathcal{N}_a + \dfrac{\mathcal{N}_a^2}{\mu_a}\right)
            \left(\mathcal{N}_b + \dfrac{\mathcal{N}_b^2}{\mu_b}\right).
        \end{equation}
        This condition implies that the maximal amount of correlations is limited by the subsystem's properties: the number of photons and the number of modes.
        Thus, we are able to normalize the amount of correlations $\mathcal{J}$, and we introduce the \textit{degree of signal-idler correlations} 
 
        \begin{equation}
           \xi = \dfrac{\mathcal{J}}{  \sqrt{\mathcal{N}_a + \dfrac{\mathcal{N}_a^2}{\mu_a} }   \sqrt{\mathcal{N}_b + \dfrac{\mathcal{N}_b^2}{\mu_b} } }.
           \label{eq_degree_corr_def}
        \end{equation}
        By definition, $ 0 \leq \xi \leq 1$, where $\xi=0$ indicates the absence of signal-idler correlations and corresponds to a thermal GSID-state.
        For maximally correlated states, $ \xi = 1$.

        Note that the value $\xi$ is a measure of correlations for the complete GSID-state: Both the values $\xi$ and $\mathcal{J}$ are the characteristics of the GSID-state as a whole, and they do not give any information about the distribution of the correlations over the modes.

    \subsubsection{ Purity and renormalized purity  }
    \label{sec_purity}
        The purity of a quantum state $\hat{\rho}$ is defined as $P(\hat{\rho})=\trace{\hat{\rho}^2}$.
        For a GSID-state $\hat{\rho}$, the purity reads as (see Appendix~\ref{app_purity})
        \begin{equation}
                P(\hat{\rho}) =  \Big [\det  \begin{pmatrix}
                    2\veccal{D}_a^T + \vec{1}_{N}   &  2\veccal{C}_{ab}              \\
                    2\veccal{C}^H_{ab}              &  2\veccal{D}_b + \vec{1}_{M}
                \end{pmatrix}  \Big ]^{-1}.
                \label{eq_purity}
        \end{equation}

        The purity itself is not a good quality measure for GSID-states with a low number of photons. 
        Indeed, let us consider a thermal GSID-state $\hat{\rho}_{th}$ (a state with $\xi=0$) that has the same matrices $\veccal{D}_{a,b}$ as the considered state $\hat{\rho}$.
        The purity $P(\hat{\rho}_{th})\leq P(\hat{\rho})$.
        In the case of $\mathcal{N}_{a,b}\ll 1$, this inequality still holds; however, the values $P(\hat{\rho}_{th})$ and $P(\hat{\rho})$ are both close to $1$. 
        This comes from the fact that the vacuum component of both states is higher than all other Fock components. 
        As a result, it is very hard to distinguish the states $\hat{\rho}_{th}$ and $\hat{\rho}$ from the vacuum and the purity is not an obvious measure for this.

        In order to adapt the purity for weak GSID-states, we introduce an alternative purity measure, namely \textit{renormalized purity} as
        \begin{equation}
            R(\hat{\rho}) = \dfrac{P(\hat{\rho}) - P(\hat{\rho}_{th})}{1-P(\hat{\rho}_{th})}.
            \label{eq_renormalized_purity}
        \end{equation}
        The renormalized purity $R(\hat{\rho})$ is a valid purity measure: It satisfies $0\leq R(\hat{\rho})\leq 1$.
        The value $R(\hat{\rho})=1$ if the state is pure and $R(\hat{\rho})=0$ if the GSID-state is thermal.
        
        By definition, purity and renormalized purity are also basis-independent characteristics of GSID-states.
        Therefore, in Eq.~\eqref{eq_purity}, we can use correlation matrices given in any basis.

\subsection{Properties of pure GSID-states}
\label{sec_pure_gtts}
    Pure states are defined as having a purity $P=\trace{\hat{\rho}^2} = 1$.
    Pure GSID-states satisfy the following theorem
    \begin{theorem}
    \label{theorem_pure_state}
        The GSID-state $\hat{\rho}$, described by the matrices $\veccal{D}_a$, $\veccal{D}_b$ and $\veccal{C}_{ab}$ is pure if there exist two unitary matrices $W_A$ and $W_B$ and a real diagonal matrix $\vec{\Lambda}$ such that the following transformation holds
        \begin{equation}
            W_A \veccal{D}_a W_A^H = \vec{\Lambda}, ~ ~W_B \veccal{D}_b W_B^H = \vec{\Lambda}
        \end{equation}
        and
        \begin{equation}
            W_A \veccal{C}_{ab} W_B^H  = (\vec{\Lambda}+\vec{\Lambda}^2)^{\frac12}.
        \end{equation}
    \end{theorem}
    Its proof is presented in Appendix~\ref{app_theorem_proof} and provides the following corollaries for pure GSID-states:

    \textbf{1}. 
        There exists a basis of broadband modes such that $\vec{\hat{A}} = W_A^T\vec{\hat{a}}$ and $\vec{\hat{B}} = W_B^T\vec{\hat{b}}$.
        In this basis, the correlation matrices are diagonal
        \begin{equation}
            \veccal{D}_A=\veccal{D}_B=\vec{\Lambda} ~ \text{and}  ~ \veccal{C}_{AB} = (\vec{\Lambda}+\vec{\Lambda}^2)^{\frac12}.
            \label{eq_schmidt_correlators}
        \end{equation}
        We will call this basis the basis of \textit{bipartite Schmidt modes}.
        As long as all the correlation matrices are diagonalized in the same basis, the basis of bipartite Schmidt modes is simultaneously the MW-basis as well as the C-basis.

    \textbf{2}.
        The correlation matrices satisfy the minimal uncertainty condition in the form
    \begin{equation}
            \begin{cases}
             \veccal{D}^T_a - \veccal{C}_{ab}(\veccal{D}_b  + 1)^{-1}\veccal{C}^H_{ab} & = 0, \\
                     \veccal{D}_b - \veccal{C}_{ab}^H (\veccal{D}_a^T  + 1)^{-1}\veccal{C}_{ab} & = 0.
            \end{cases}
        \end{equation}

    \textbf{3}. The photon number distribution for signal ($a$) and idler ($b$) subsystems coincide: $\veccal{D}_A=\veccal{D}_B = \vec{\Lambda}$.
        Therefore, the number of photons
        \begin{equation}
            \mathcal{N}_a = \mathcal{N}_b \equiv \mathcal{N} 
        \end{equation}
        as well as the number of modes
        \begin{equation}
            \mu_a=\mu_b =  \dfrac{[\trace{\vec{\Lambda}}]^2}{\trace{\vec{\Lambda}^2}} \equiv \mu  .
            \label{eq_numberMu_pure_state}
        \end{equation} 
        for both subsystems are equal.

    \textbf{4}. 
        In the basis of bipartite Schmidt modes, the signal and idler subsystems are pairwise correlated: the matrix $\veccal{C}_{AB}$ is diagonal.
        The total amount of correlations then takes the form 
        \begin{equation}
            \mathcal{J} = \mathcal{N}+\dfrac{\mathcal{N}^2}{\mu} 
        \end{equation}
        and the degree of correlation $\xi=1$.

    \textbf{5}. 
        The number of correlated modes takes the form
            \begin{equation}
                  K = \dfrac{\left[ \trace{ \vec{\Lambda}+\vec{\Lambda}^2 } \right]^2}{\trace{[\vec{\Lambda}+\vec{\Lambda}^2]^2}}.
                  \label{eq_numberK_pure_state}
             \end{equation} 
        The equality $K=\mu$ holds only for pure states with a uniform distribution for the number of photons, i.e., when only $n$ entries of the matrix $\vec{\Lambda}$ are equal and the rest are zeros.

        For an arbitrary distribution of the number of photons, 
        \begin{equation}
            K\neq\mu.   
        \end{equation}
        Only when $\vec{\Lambda}^2 \ll \vec{\Lambda}$, does $K \approx \mu$, which corresponds to optical fields with low-photon number. 
        Therefore, the number of modes $\mu$ and the number of correlated modes can take different values even for pure GSID-states.

\subsection{Summary}

    In this section, we presented the main parameters used in this paper.
    As we see, the mode structure of the arbitrary GSID-states can be quite complicated, providing not only the different number of modes in the signal and idler subsystems, but also different number of correlated modes.  
    Even for pure states, the number of correlated modes differs from the number of modes, although the number of modes in both subsystems is equal.

    Real experimental states are never pure, meaning that characterisation of only one of signal (or idler) subsystems does not fully characterize the quantum state.
    In the next section, we present a reliable, holistic characterization procedure based on the mode-unresolved detection, that can be applied to the GSID states.

\section{GSID-source characterisation}
\label{sec_characterisation}

Let us return to the GSID-source depicted in Fig.~\ref{fig_scheme_1}(a), which satisfies the restrictions R1-R5 from Sec.~\ref{sec_system_under_study}.
Let us assume that the correlation matrices
$$
    \braket{\hat{a}^\dagger(\omega)\hat{a}(\omega^\prime)}, \braket{\hat{b}^\dagger(\omega)\hat{b}(\omega^\prime)} ~ \text{and} ~ \braket{\hat{a}(\omega)\hat{b}(\omega^\prime)}
$$
are unknown.
The main question we answer in this section is: How can we characterize the source with intensity detectors, and what can we say about the source quality?

The first type of intensity measurements we consider are frequency-dependent measurements.
These measurements are provided by an optical spectrum analyzer (see Fig.~\ref{fig_scheme_1}(d)).
They are described by the photon-number density operators $ \hat{n}_{a}(\omega) = \hat{a}^\dagger(\omega)\hat{a}(\omega)$ and $ \hat{n}_{b}(\omega^\prime) =  \hat{b}^\dagger(\omega^\prime)\hat{b}(\omega^\prime)$, that correspond to the signal monochromatic mode $\omega$ and the idler monochromatic mode $\omega^\prime$, respectively.

The second type of intensity measurements we consider are mode-unresolved detection (see Fig.~\ref{fig_scheme_1}(b)).
In this case, the measurement outcomes are defined by the total photon-number operators $ \hat{N}_{a} = \int d\omega  \hat{n}_{a}(\omega) $ and $ \hat{N}_{b} =  \int d\omega^\prime \hat{n}_{b}(\omega^\prime) $.
Mode-unresolved detection averages over all modes in the detected light, and therefore only provides basis-independent characteristics. 
Note that we assume that the actual number of photons can be reconstructed from the measurements (i.e., we assume detectors with 100\% efficiency). 
In the case of detectors with imperfect efficiency, we assume that they are calibrated and the absolute photon number produced by the GSID-source can be determined. 

We also assume that the correlation functions describing the state are `good' in the sense that they can be approximated on a discrete frequency space and this approximation is convergent in the limit of $\Delta \omega \rightarrow 0$.
In this case, all the results from the previous section are directly applicable for continuous frequencies: The continuous operators are defined as $\hat{a}(\omega) \equiv \lim_{\Delta \omega \rightarrow0} \frac{ \hat{a}(\omega_i) }{\sqrt{\Delta \omega}}$ and the sums are replaced by integrals $\sum_{\omega_i} \rightarrow \int d\omega$.

We also assume that the state is described by `good' correlation functions that are convergent in the limit of $\Delta \omega \rightarrow 0$. 
In this case, all the results from the previous section are directly applicable for continuous frequencies: The continuous operators are defined as $\hat{a}(\omega) \equiv \lim_{\Delta \omega \rightarrow0} \frac{ \hat{a}(\omega_i) }{\sqrt{\Delta \omega}}$ and the sums are replaced by integrals $\sum_{\omega_i} \rightarrow \int d\omega$.

\subsection{ Mode-unresolved measurements }

    The first characteristic we can measure is the total number of photons, which is just the average number of photons measured by each of the detectors 
    \begin{equation}
        \mathcal{N}_a = \braket{\hat{N}_a}, ~ ~ \mathcal{N}_b = \braket{\hat{N}_b}.
    \end{equation}
    The mode-unresolved second-order correlation functions for signal and idler subsystems are defined as 
    \begin{equation}
       g_a^{(2)} = \dfrac{\braket{:\hat{N}_a \hat{N}_a:}}{\braket{\hat{N}_a}^2}, ~ ~ g_b^{(2)} = \dfrac{\braket{:\hat{N}_b \hat{N}_b:}}{\braket{\hat{N}_b}^2}, 
          \label{eq_g2_1}
    \end{equation}
    while the  mode-unresolved signal-idler second-order correlation function is 
    \begin{equation}
       g^{(1,1)} =  \dfrac{\braket{:\hat{N}_a \hat{N}_b:}}{\braket{\hat{N}_a}\braket{\hat{N}_b}},
       \label{eq_g11_1}
    \end{equation}
    where $\braket{:...:}$ indicates normal-ordering.
    These correlation functions can be measured with the optical schemes depicted in Fig.~\ref{fig_scheme_1} (see Appendix~\ref{app_sec_ord_corr_func} for details).
    Using the results from Sec.~\eqref{ssub_num_modes_subsyst}, we find that the second-order functions are connected with the number of modes through 
    \begin{equation}
       g_a^{(2)} = 1 + \dfrac{1}{\mu_a},  ~ ~    g_b^{(2)} = 1 + \dfrac{1}{\mu_b},
       \label{eq_g2_2}
    \end{equation}
    and 
     \begin{equation}
       g^{(1,1)} = 1 + \dfrac{ \mathcal{J}  }{\mathcal{N}_a \mathcal{N}_b} .
          \label{eq_g11_2}
    \end{equation}
    Equation~\eqref{eq_g2_2} provides an experimental way to determine the number of modes $\mu_a$ and $\mu_b$ from the measurements and Eq.~\eqref{eq_g11_2} -- the total amount of correlations.

    With the use of Eq.~\eqref{eq_degree_corr_def}, the degree of correlation can be formulated in terms of the correlation functions as 
    \begin{equation}
       \xi = \dfrac{g^{(1,1)}_{ab}-1}{\sqrt{ \dfrac{1}{\braket{\hat{N}_a}} + (g_a^{(2)}-1) } \sqrt{ \dfrac{1}{\braket{\hat{N}_b}} + (g_b^{(2)}-1) }   } .
    \end{equation}
    This equation provides an experimental way to obtain the degree of correlation $\xi$.

  As a result, mode-unresolved intensity measurements of the GSID-source allow us to determine five parameters: $\mathcal{N}_a$, $\mathcal{N}_b$, $\mu_a$, $\mu_b$ and $\xi$. 
  The parameters $\mathcal{N}_{a,b}$ and $\mu_{a,b}$ are the individual properties of subsystems, and only the parameter $\xi$ indicates the presence of signal-idler correlations.
  As we show in Sec.~\ref{sub_purity_bounds}, all these parameters together are necessary to determine the purity bounds (see Sec.~\ref{sub_purity_bounds}).
  
\subsection{ Frequency-resolved measurements }
    
    Introducing monochromators in the signal and idler channels enables us to determine the spectral properties of the studied state $\hat{\rho}$ (see Fig.~\ref{fig_scheme_1}(d)).
    Here, we consider two of them: the photon number spectrum and the joint spectral intensity.

\subsubsection{ Photon number spectrum }
        The photon number spectral densities for signal and idler fields are defined as 
        \begin{align}
            n_{a}(\omega) & = \braket{\hat{n}_{a}(\omega)} = \braket{\hat{a}^\dagger(\omega)\hat{a}(\omega)},\\
            n_{b}(\omega) & =  \braket{\hat{n}_{b}(\omega)} =  \braket{\hat{b}^\dagger(\omega)\hat{b}(\omega)}
        \end{align}
        and are measured independently by scanning over all the frequencies.
        Therefore, the spectra $n_{a}(\omega)$ and $n_{b}(\omega)$ are properties of the subsystems, and in order to characterize the spectral correlation properties, we need to define the joint spectral intensity.

\subsubsection{ Joint spectral intensity  }
\label{ssec_JSI}

        Before we introduce the joint spectral intensity (JSI), let us consider a second-order photon-number correlator $\braket{\hat{n}_a(\omega) \hat{n}_b(\omega^\prime)}$.
        For the GSID-state, it takes the form (see Appendix~\ref{app_norm_ordered_undisplcaed})
        \begin{equation}
            \braket{\hat{n}_a(\omega) \hat{n}_b(\omega^\prime)} =  |\braket{\hat{a}(\omega) \hat{b}(\omega^\prime)}|^2 + n_{a}(\omega)n_{b}(\omega^\prime).
        \end{equation}
        The second term depends only on the product of the signal and idler photon number spectra; therefore, it does not represent the correlations stored in the GSID-state but the accidental coincidences.
        In order to exclude them,  we can define the JSI as a second-order photon-number covariance  
        \begin{equation}
            J(\omega,\omega^\prime) = \braket{\hat{n}_a(\omega) \hat{n}_b(\omega^\prime)} - \braket{\hat{n}_a(\omega)}\braket{\hat{n}_b(\omega^\prime)}.
        \end{equation}
        Then the JSI for the GSID-state is 
        \begin{equation}
            J(\omega,\omega^\prime) = |\braket{\hat{a}(\omega)\hat{b}(\omega^\prime)}|^2
        \end{equation}
        and represents the spectral correlations between monochromatic modes of signal and idler subsystems.
        The function $\braket{\hat{a}(\omega)\hat{b}(\omega^\prime)}$ can be interpreted as the complex amplitude of the JSI, and therefore we call it the joint spectral amplitude (JSA).

        Typically, the JSI is defined via the two-photon approximation, and only coincidence counts without subtracting accidental ones are taken into account while measuring it~\cite{Harder_2013,Zielnicki_2018}. 
        Indeed, for pure states, in the low-photon-number limit, the impact of accidental counts is negligible (as seen in Eqs.~\eqref{eq_uncet_inequality_elements} and \eqref{eq_degree_corr_def}), and our definition of JSI coincides with that obtained from the two-photon approximation.
        However, accidental counts become significant as the photon number grows (the quadratic photon-number term in Eq.~\eqref{eq_degree_corr_def}) and must be accounted for to measure signal-idler correlations accurately.

        Note that the JSI does not contain the full information about the spectral correlations: the function $\braket{\hat{a}(\omega)\hat{b}(\omega^\prime)}$ is a complex function, while the JSI is not.
        This means that JSI measurements omit the phase information. As we see in the next section, the phase is critical for the C-basis determination and the number of correlated modes.

\subsubsection{ The number of correlated modes and flat phase assumption of JSA  }
\label{ssub_meas_corr_modes}

        As we have found in Sec.~\ref{ssub_c_basis}, the number of correlated modes $K$ is defined in terms of the C-basis and requires not only the knowledge of the intensity but also the phase of the JSA.
        Its measurement is possible using interferometer-based detection schemes~\cite{Tischler_2015,MacLean_2019,Davis_2020,Thekkadath_2022}, and as a result, the desired value $K$ can be obtained.

        However, these interferometer-based schemes are more sophisticated than JSI measurements with frequency-resolving intensity detectors. 
        Therefore, the typical experimental assumption is that the JSA has a flat- or linear- phase profile. 
        Then the value $\sqrt{\text{JSI}}$ is postulated as the number of modes~\cite{Harder_2013,Zielnicki_2018}.
        However, this assumption can lead to significant underestimation of the number of modes and below we show why.

        Let us assume that we have measured the JSI $J(\omega,\omega^\prime)$ and therefore we can have determined the modulus of the JSA $|\braket{\hat{a}(\omega)\hat{b}(\omega^\prime)}| = \sqrt{J(\omega,\omega^\prime)}$ (but not its phase).
        In terms of basis reconstruction, the modulus of the JSA corresponds to the S-basis, which is not a `good' basis (see Sec.~\ref{ssub_s_basis}).
        The Schmidt decomposition of the $\sqrt{J(\omega,\omega^\prime)}$ defines the number $K_{\sqrt{J}}$ which satisfies the following inequality (see Appendix~\ref{app_JSI_inequality})
        \begin{equation}
            K_{\sqrt{J}} \leq K.
            \label{eq_K_lbound}
        \end{equation}
        As a result, the value $K_{\sqrt{J}}$ is a lower bound for the number of correlated modes and therefore cannot be considered as the actual number of correlated modes $K$.
        
        This result shows that the linear-phase profile assumption of JSA is very strong and can lead to underestimation of the number of correlated modes. 
        For an accurate source characterisation, we should expect the worst scenario and assume an unknown phase-profile for the JSA.
        Although the value $K_{\sqrt{J}}$ provides only the lower bound for the number of correlated modes, it can be still useful for purity bound determination (see Sec.~\ref{sub_comments_bounds}).

\subsection{ Degree of correlation bound }
    \label{sub_deg_corr_bound}
        The parameters $\mathcal{N}_{a,b}$, $\mu_{a,b}$, $K$ and $\xi$ for a GSID-state $\hat{\rho}$ are not completely independent.
        There exists a bound on $\xi$ (see Appendix~\ref{app_ineq_degree_corr})
        \begin{equation}
            \xi \leq \mathrm{min} \big[  1, H(\vec{x}_a), H(\vec{x}_b) \big],  
            \label{eq_inequality_xi}
        \end{equation}
        where the vectors of parameters are $\vec{x}_a=(\mathcal{N}_a,\mathcal{N}_b,\mu_a,\mu_b,K)$ and $\vec{x}_b=(\mathcal{N}_b,\mathcal{N}_a,\mu_b,\mu_a,K)$.
        The function
        \begin{equation}
             H(\vec{y}) = \frac{ x \big(\frac{ x^\prime }  { \sqrt{km^\prime} }  + 1 \big)}{ \sqrt{x + \frac{x^2}{m} }   \sqrt{x^\prime + \frac{(x^\prime)^2}{m^\prime} }},
         \end{equation}
         where the vector of variables reads as $\vec{y} = (x,x^\prime,m,m^\prime,k)$.

        From a practical point of view, inequality~\eqref{eq_inequality_xi} can be used as an indicator of a correctness and consistency of the measured values.

\subsection{ Purity bounds }
    \label{sub_purity_bounds}
    Let us assume that, for the studied GSID-state $\hat{\rho}$, we have measured the values  $\mathcal{N}_{a,b}$, $\mu_{a,b}$, $K$ and $\xi$.
    What can we say about the purity of the state?
    Here we present two lower bounds for state purity $P(\hat{\rho})$ and an upper bound for the renormalized purity $R(\hat{\rho})$. 
    To write the expression in compact forms, in this section, we use two vectors of parameters 
    $\vec{x}_a = (\mathcal{N}_a, \mathcal{N}_b, \mu_a, \mu_b, K, \xi) $ and 
    $\vec{x}_b = (\mathcal{N}_b, \mathcal{N}_a, \mu_b, \mu_a, K, \xi) $.

\subsubsection{ Strict lower bound for purity }
    \label{sub_strict_purity_bounds}

    The first bound is 
    \begin{equation}
        P(\hat{\rho}) \geq \tilde{P}^s(\hat{\rho}) = \mathrm{max} \Big[ p_s(\vec{x}_a), p_s(\vec{x}_b) \Big],
        \label{eq_purity_bound_strict}
    \end{equation}
    where
    \begin{equation*}
        p_s(\vec{y}) = e^{ -2(x + x^\prime) + 4 F(\vec{y})},
    \end{equation*}
    and 
    \begin{equation*}
        F(\vec{y}) = \dfrac {\zeta \sqrt{k m (x+\frac{x^2}{m}) (x^\prime+\frac{(x^\prime)^2}{ m^\prime})} }{  (2 x   +  \sqrt{k m})},
    \end{equation*}
    with the vectors of variables $\vec{y} = (x, x^\prime, m, m^\prime, k, \zeta)$.
    For the derivation, see Appendix~\ref{app_purity_bound_strict}.
    
    This bound is a strict one: it is valid for any state; however, it provides a very conservative estimate.
    Note that all the signal-idler correlations are determined by the functions $F(\vec{y})$ where the parameter $\xi$ plays a critical role.
    For thermal GSID-state $\hat{\rho}_{th}$ with  $\xi=0$, the bound is given by
    \begin{equation*}
       P(\hat{\rho}_{th}) \geq e^{ -2(\mathcal{N}_a + \mathcal{N}_b) }.
    \end{equation*}

\subsubsection{ Approximate lower bound for purity }
    \label{sub_approximate_purity_bounds}

    The second bound reads as
    \begin{equation}
        P(\hat{\rho}) \geq \tilde{P}^a_{r_a, r_b}(\hat{\rho}) =  \mathrm{max} \Big[ p_a(\vec{x}_a, r_a, r_b),    p_a(\vec{x}_b, r_a, r_b) \Big].
        \label{eq_purity_bound_2__}
    \end{equation}
    where the function
    \begin{equation*}
        p_a(\vec{y}, r, r^\prime) =    \bigg[ 1 + \frac{2x}{r}   \bigg]^{-r} \Bigg[ 1 + \frac{2x^\prime}{r^\prime} -\dfrac{4F(\vec{y})}{ r^\prime   }   \Bigg]^{-r^\prime}.
    \end{equation*} 
    For the derivation, see Appendix~\ref{app_purity_bound_approx}.

    To use Eq.~\eqref{eq_purity_bound_2__}, we need to determine two additional parameters $r_{a,b}$.
    According to the derivation, these parameters correspond to the number of nonzero occupied modes of the studied state $\hat{\rho}$.
    For single-mode states (when $\mu_a=\mu_b=1$), $K = 1$ and the equality holds.
    For multimode states, the equality holds only for pure states with a uniform photon-number distribution in MW-basis.
    In this case the $r_{a}=r_{b}$, $K=\mu_{a}=\mu_{b}$ and $\xi=1$.

    However, the photon-number distribution remains unknown in the studied framework.
    Therefore, for practical usage, we suggest to replace the values $r_{a,b}$ with the measured values $\mu_{a,b}$ for any state.
    Indeed, the values $\mu_{a,b}$ indicate the number of `significant' modes and the function~\eqref{eq_purity_bound_2__} supports the fractional values of $r_{a,b}$.
    Replacement of the unknown non-zero matrix elements $r_{a,b}$ by the values $\mu_{a,b}$ is not a correct mathematical operation; therefore, 
    \begin{equation}
         P(\hat{\rho})  \gtrsim  \tilde{P}^a_{\mu_a, \mu_b}(\hat{\rho}),
        \label{eq_purity_bound_approx}
    \end{equation}
    where we use the symbol $\gtrsim$ to emphasize that this inequality can be violated.
    Despite of the absence of mathematical rigor, the value $\tilde{P}^a_{\mu_a, \mu_b}(\hat{\rho})$ has a large practical potential for almost single-mode sources, which is demonstrated in Sec.~\ref{sec_examples_realistic}.

    In addition to $\tilde{P}^a_{\mu_a, \mu_b}(\hat{\rho})$, we can use, e.g., $\tilde{P}^a_{2\mu_a, 2\mu_b}(\hat{\rho})$, which gives a more accurate bound; however, it is also not necessarily strict.
    An optimal choice of $\tilde{P}^a_{r_a, r_b}(\hat{\rho})$ for a given state requires the knowledge of the photon-number distribution, which remains unknown in the studied framework, and therefore beyond the scope of this paper.

    \subsubsection{ Bound for the renormalized purity }
    \label{sub_renorm_purity}

        In Sec.~\ref{sec_purity}, we introduced the purity and the renormalized purity $R(\hat{\rho})$.
        Unfortunately, the lower purity bounds~\eqref{eq_purity_bound_strict} and \eqref{eq_purity_bound_approx} do not provide a bound on $R(\hat{\rho})$: the chain of inequalities breaks during the derivations.
        Therefore, to determine bounds for $R(\hat{\rho})$ a different approach must be used.

        In this work, we present an upper bound for the  renormalized purity for weak GSID-states, i.e., the states with $\mathcal{N}_{a,b} \ll 1$. 
        In this case, 
        \begin{equation}
             R(\hat{\rho}) \leq \tilde{R}=\mathrm{min}[g(\vec{x}_a), g(\vec{x}_b)],
             \label{eq_renormalized_purity_bound}
        \end{equation}
        where 
        \begin{equation}
            g(\vec{y}) = \dfrac{ \dfrac{1 + 2x^\prime}{ 1 + 2x^\prime - 4 F(\vec{y}) } - 1} {(1 + 2x)(1 + 2x^\prime)  - 1} 
        \end{equation}
        with the vectors of variables $\vec{y} = (x, x^\prime, m, m^\prime, k, \zeta)$.
        For derivation, see Appendix~\ref{app_renormalized_purity_bound}.

        In the limit of $\mathcal{N}_{a,b} \rightarrow 0$,  $\tilde{R} \rightarrow \xi$; therefore, for weak GSID-states, 
        the degree of correlation is an upper bound for the renormalized purity
        \begin{equation}
            R(\hat{\rho}) \leq \xi.
            \label{eq_r_purity_xi}
        \end{equation}

    \subsubsection{ Comment on the bounds }
    \label{sub_comments_bounds}

        All of the derived bounds~\eqref{eq_purity_bound_strict}, \eqref{eq_purity_bound_approx} and \eqref{eq_renormalized_purity_bound} require determination of the parameter $K$, which (as was discussed in Sec.~\ref{ssec_JSI} and \ref{ssub_meas_corr_modes}) is difficult to measure.
        One can instead use the value $K_{\sqrt{J}}$ (see Appendix~\ref{app_purity_estimation}), however, the purity bounds are worse in this case.
        In the case when $K_{\sqrt{J}}$ is unknown, we can use $K=1$, providing the lowest bound.

\section{Numerical Examples}
\label{sec_examples}

In this section, we consider several examples of GSID-states.
Firstly, we present states with identical photon number distributions, but which exhibit unique signal-idler correlation properties, highlighting the independence between the properties of the subsystems and the full state.
In the second part, we consider a realistic waveguided GSID-source with 5\% optical loss after state generation, and study the effect of spectral filtering. This investigation shows the importance of treating both source generation and preparation holistically, and demonstrates the applicability of our framework to current state-of-the-art sources.

\subsection{Examples in MW-basis}
\label{sec_examples_mw_basis}

\begin{figure}
\centering
\includegraphics[width=\linewidth]{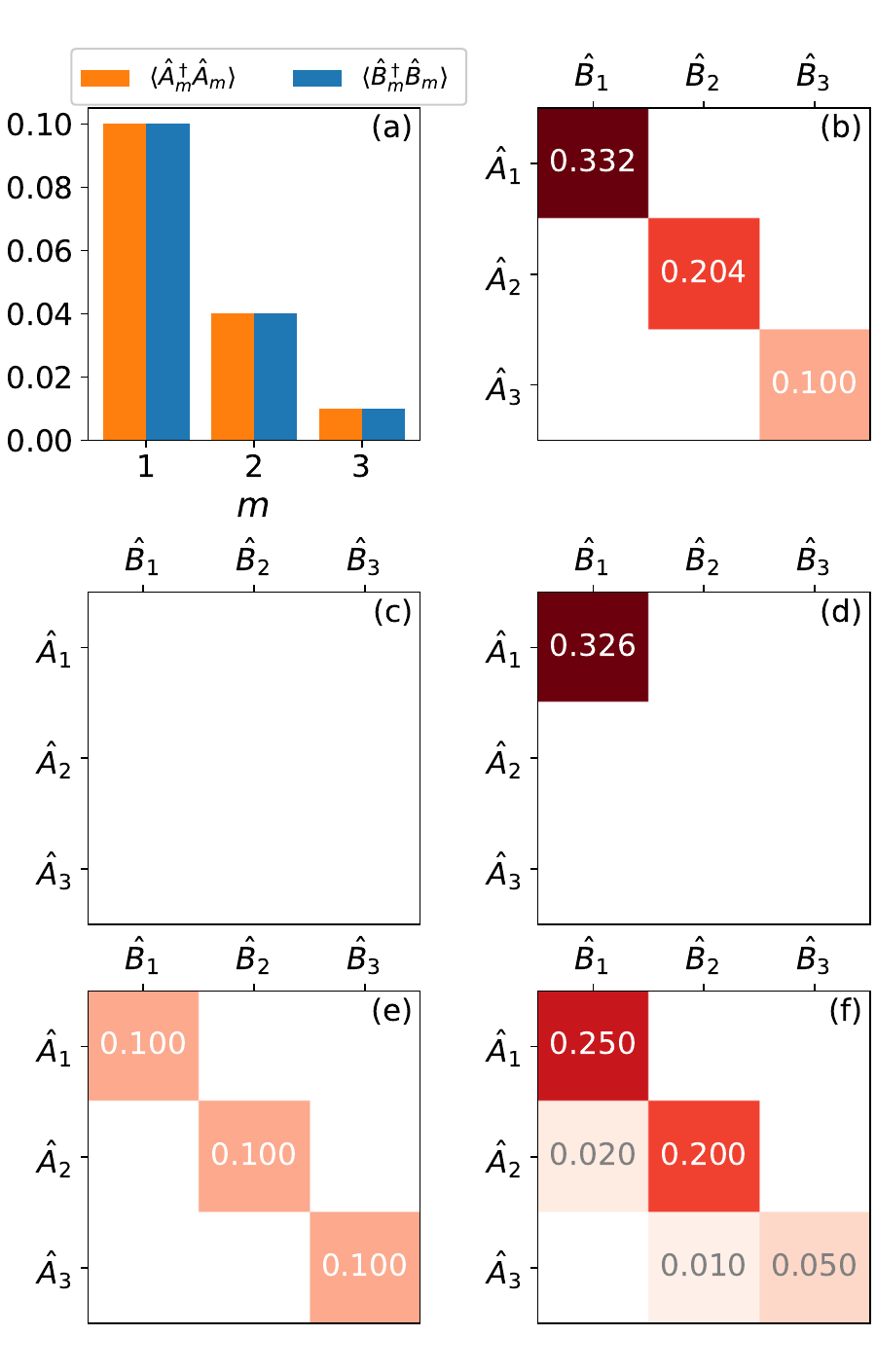}

\caption{Correlation matrices in MW-basis for the considered GSID-states in Sec.~\ref{sec_examples_mw_basis}.
(a) The diagonal elements of matrices $\veccal{D}_{A, B}$ (The photon number distribution in MW-basis). The matrices $\veccal{D}_{A, B}$ are the same for all the considered states.
The matrices $\veccal{C}_{AB}$ for the considered states: (b) pure, (c) thermal, (d) mixed I, (e) mixed II and (f) mixed III.
The resulting parameters are given in Table~\ref{tab_examples} .
}
\label{fig_2}
\end{figure}

GSID-states can always be written in the MW-basis; therefore, we can always define a valid quantum state with the diagonal matrices $\veccal{D}_{A, B}$ and some physical matrices $\veccal{C}_{AB}$ that satisfy the uncertainty relation~\eqref{eq_uncertainty_main}.
Here, for simplicity, we consider a ($3+3$)-mode system and define five states with the same photon-number distributions for both signal and idler fields.
This means that the matrices $\veccal{D}_{A}=\veccal{D}_{B}\equiv\vec{\Lambda}$ are the same for all of them. The matrices $\veccal{D}_{A,B} $ are shown in Fig.~\ref{fig_2}(a), providing the number of photons $\mathcal{N}_a=\mathcal{N}_b=0.15$ and the number of modes $\mu_a=\mu_b\equiv\mu= 1.92$.
To obtain GSID states with different correlation properties, we consider five different $\veccal{C}_{AB}$-matrices shown in Fig.~\ref{fig_2}(b-f).
The presented GSID states are physical, i.e., the uncertainty relation holds for all states~\eqref{eq_uncertainty_main}. 
Table~\ref {tab_examples} lists the parameters of the considered states.

\begin{table}[h!]
\begin{tabular}{l|c|c|c|c|c|c|c|c}
State & $\mathcal{N}_{a,b}$ & $\mu_{a,b}$ & $g_{a,b}^{(2)}$  & $\xi$ & $K$    & $g^{(11)}$ & $P(\hat{\rho})$ & $R(\hat{\rho})$\\
\hline
Pure      &\multirow{5}{*}{0.15} & \multirow{5}{*}{1.92} & \multirow{5}{*}{1.52} & 1.000 & 1.88 & 8.19 & 1.000 & 1.000 \\
Thermal   & & &  & 0.000 & --     & 1.00 & 0.572 & 0.000  \\
Mixed I   & & &  & 0.658 & 1.00   & 5.73 & 0.812 & 0.561 \\
Mixed II  & & &  & 0.187 & 3.00   & 2.35 & 0.635 & 0.146 \\
Mixed III & & &  & 0.652 & 1.99   & 5.69 & 0.812 & 0.561 \\
 \end{tabular} \\
 \caption{Parameters of the considered GSID-states in Sec.~\ref{sec_examples_mw_basis}.}
 \label{tab_examples}
\end{table}

The first state we investigate is a pure one ($P=1$). 
The MW-basis corresponds to the bipartite Schmidt basis, and the matrix $\veccal{C}_{AB}$ is diagonal with elements defined by Eq.~\eqref{eq_schmidt_correlators} (see Fig.~\ref{fig_2}(b)).
The pure state is also maximally correlated with $\xi=1$; however, the number of correlated modes is less than the total number of modes, since $K=1.88<\mu$. 
To explain this behavior, let us have a look at the matrices $\veccal{D}_{A}$ and $\veccal{C}_{AB}$ in Eq.~\eqref{eq_schmidt_correlators}.
The matrix $\veccal{C}_{AB}$ also includes the term $\vec{\Lambda} ^2$, which changes the ratio of signal-idler correlations between different modes. 
As a result, the ratios $[\veccal{C}_{AB}]_{11}/[\veccal{C}_{AB}]_{33} > [\veccal{D}_{A}]_{11}/[\veccal{D}_{A}]_{33}$, and from Eqs.~\eqref{eq_numberMu_pure_state} and \eqref{eq_numberK_pure_state}, we get that $\mu\neq K$.
These results emphasize that -- even for pure states -- the number of correlated modes and the number of modes can be different.

The second considered state is a thermal GSID-state, which shows the absence of signal-idler correlations. 
The matrix $\veccal{C}_{AB}=0$; thus, $\xi=0$ and the number of correlated modes $K$ is undefined.
Note that the purity of this state is $P=0.57$ (see Sec.~\ref{sec_purity}) while the renormalized purity is $R=0$, which reveals the relevance of this renormalized metric.

The final three states which we consider, three mixed states with nonzero matrices $\veccal{C}_{AB}$, are shown in Fig.~\ref{fig_2}(d-f). 
Once again, the chosen state states are physical; satisfying the uncertainty relation~\eqref{eq_uncertainty_main}. However, they contain different signal-idler correlation distributions.
Indeed, the presented states demonstrate that the number of correlated modes $K$ can be both smaller and larger than the number of modes $\mu$ (see mixed states I and II, respectively).  
Note that the purities of the mixed states I and III coincide despite having different $\xi$ values and different distributions of signal-idler correlations that, in the case of mixed state III, also contain off diagonal elements.

In summary, the GSID-states can have a sophisticated structure of signal-idler correlations -- even when the properties of signal and idler subsystems are the same.
The number of photons $\mathcal{N}_{a,b}$ and the number of modes $\mu_{a,b}$ are properties of the signal and idler subsystems and no conclusion can be drawn about the signal-idler correlations from them.

\subsection{Realistic example: nearly single-mode source}
\label{sec_examples_realistic}

\begin{figure}[h]
\centering
\includegraphics[width=0.42\linewidth]{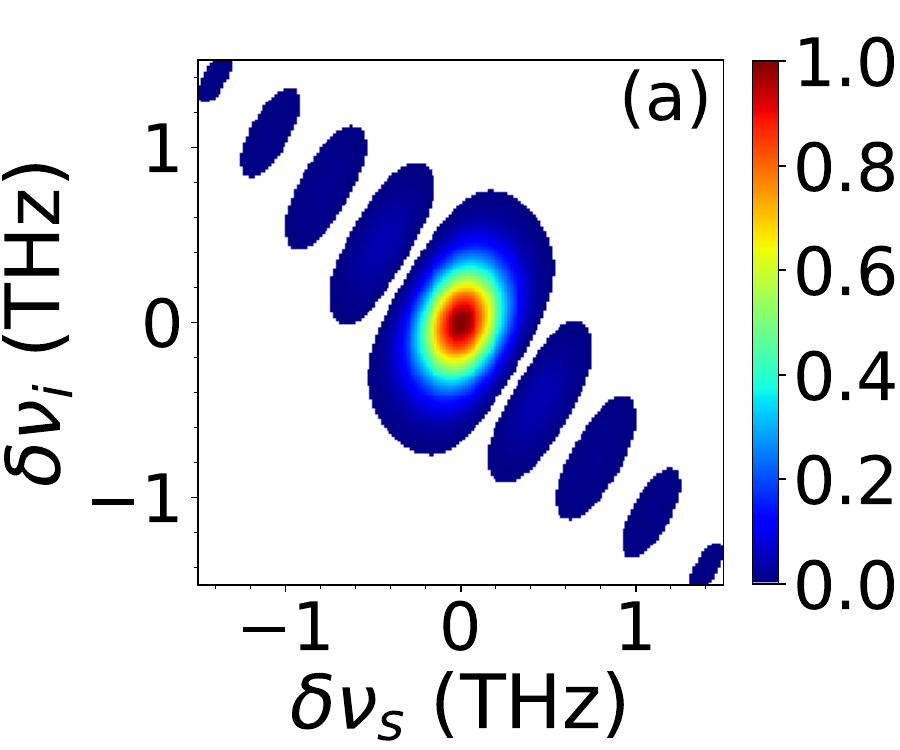}
\includegraphics[width=0.56\linewidth]{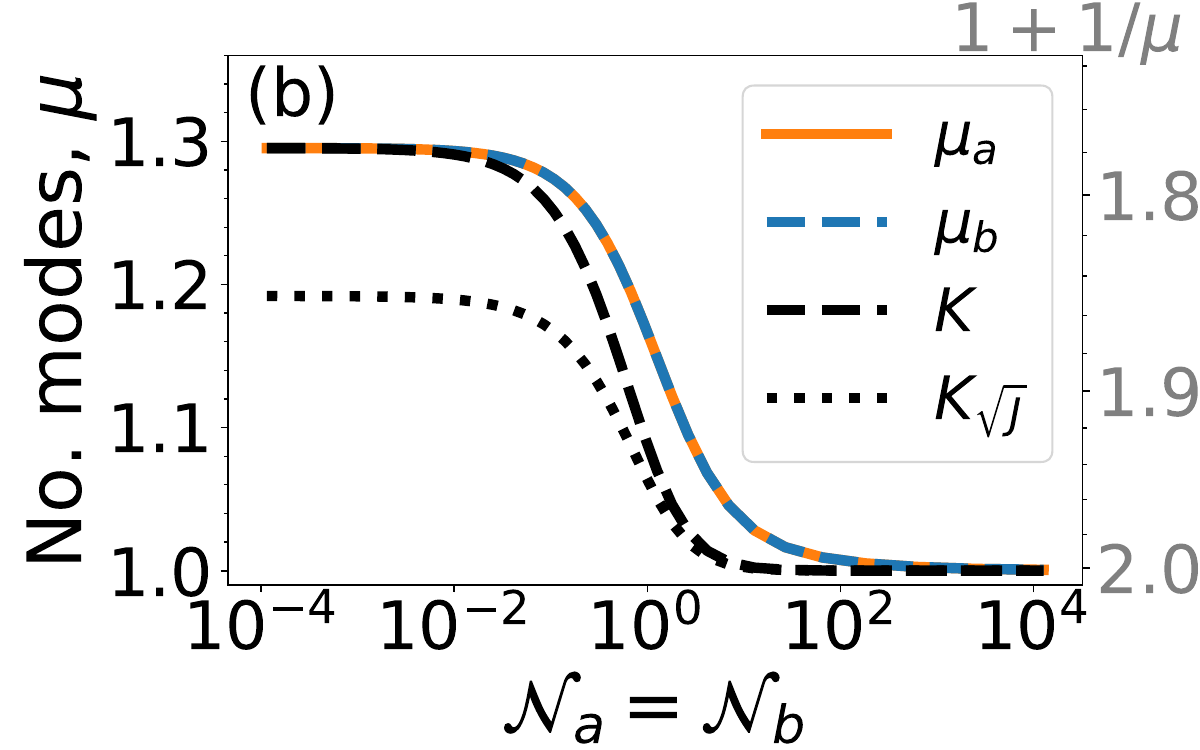} 
\caption{ (a) The JSI and (b) the dependence of the number of modes for the state seen directly after the generation stage of the realistic GSID-source).
}
\label{fig_3}
\end{figure}
 
In this section, we consider numerical simulations of three realistic GSID-sources, which consist of both generation and preparation stages (see Fig.~\ref{fig_scheme_1}(a)). 
The sources under study (denoted as `S1', `S2', and `S3') have identical generation stages but varying preparation stages, which represent different spectral shaping of the generated squeezed field.

For the generation stage, we consider a type-II PDC in a lossless single-mode waveguide and, for its numerical simulation, we use the formalism developed in Refs.~\cite{Kopylov_2025,Kopylov_2025_HOM}. 
This approach supports spatial-ordering effects (also often referred to as time-ordering), allowing us to study the PDC properties both in low- and high-gain regimes. 
The matrices $\veccal{D}_{a,b}$ and $\veccal{C}_{ab}$ are obtained from the numerical solution of the corresponding differential equations~\cite{kopylov_2025_zenodo}.
The generated mean photon number is varied by tuning the intensity of the pump field driving the process.
As a pump, we use a Gaussian pulse with a central frequency $\omega_0$ and a pulse duration of $0.6$~ps.
For the waveguide, we consider a linear dispersion with the phase-matching function    
$
    \Delta k(\omega, \omega^\prime) = a (\omega-\frac{\omega_0}{2}) + b (\omega^\prime-\frac{\omega_0}{2}),
$ 
with parameters $a=0.12$~mm$\cdot$ps and $b=-0.2$~mm$\cdot$ps.
The parameters are taken to produce a PDC with a low number of modes that is a fair representation of what can be achieved in realistic waveguides \cite{luo2026}. To investigate the impact of the different preparation stages, we first present the properties of the state directly after generation.
The JSI of the degenerate PDC field in the low-gain regime (with mean photon number of $\approx10^{-4}$) is shown in Fig.~\ref{fig_3}(a), and Fig.~\ref{fig_3}(b) presents the number of modes as a function of the generated mean photon number.

Let us look at the mode structure of the GSID-state exiting the generation stage of the source.
The waveguide is taken to be lossless, resulting in the generation of a pure state. 
Therefore, for its analysis we can use the results from Sec.~\ref{sec_gtts_states}.
Indeed, Fig.~\ref{fig_3}(b) shows that the number of modes $\mu_{a}=\mu_{b}$ for any number of generated photons.
The number of modes reduces as the generated photon number increases, a behavior which is expected~\cite{luo2026,Christ2013}.
However, only in the low photon number regime is the number of correlated modes approximately equal to the number of modes, $K\approx\mu_{a,b}$, in agreement with results in Sec.~\ref{sec_pure_gtts}.
In addition, we see that the number of modes approximated from the JSI is less than the actual number of correlated modes determined from the JSA 
$K_{\sqrt{J}} < K$. 
The results presented in Fig.~\ref{fig_3}(b) demonstrate that the numbers $\mu$, $K$ and $K_{\sqrt{J}}$ do not coincide.
From an experimental point of view, this result shows that the measurement of second-order correlation functions and JSI are not interchangeable even for pure states
and instead provide complementary information.
This means that the values $K$ and $K_{\sqrt{J}}$, obtained from the JSA and the JSI, respectively, do not necessarily coincide with the value $\mu$, obtained from $g^{(2)}$.
Therefore, measurement of the $g^{(2)}$ alone, as suggested in Ref.~\cite{Christ_2011}, is not sufficient for full characterization of the mode structure of GSID-states, even when they are pure.

\begin{figure*}
    \centering
    \includegraphics[width=0.2\linewidth]{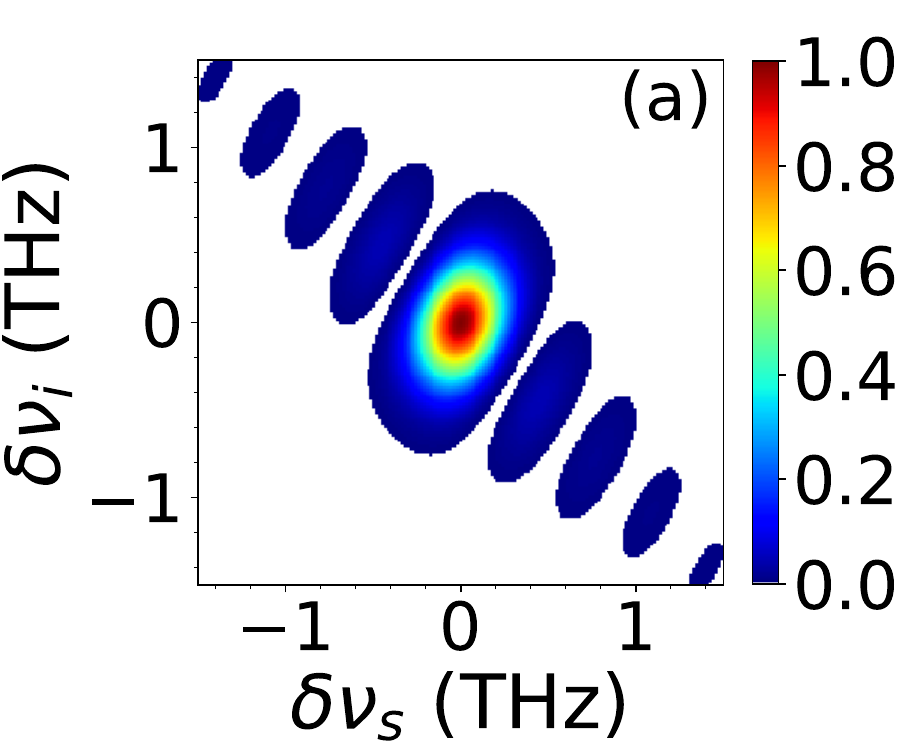}
    \includegraphics[width=0.26\linewidth]{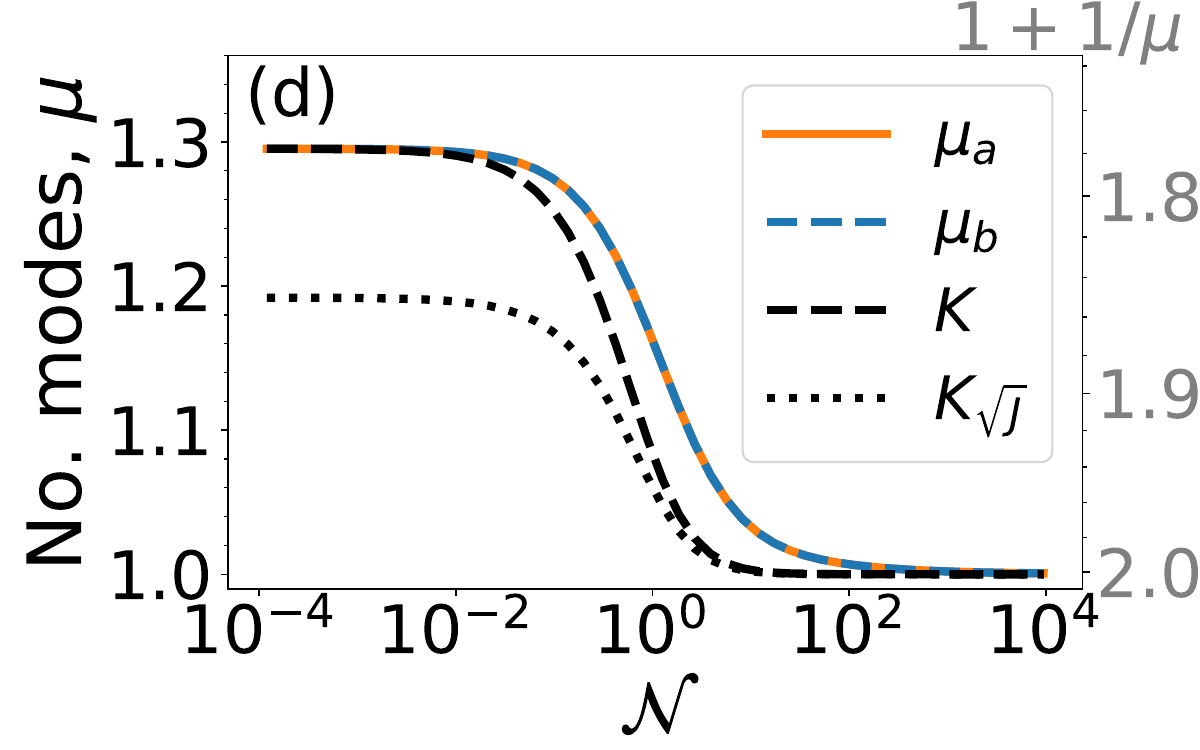}
    \includegraphics[width=0.26\linewidth]{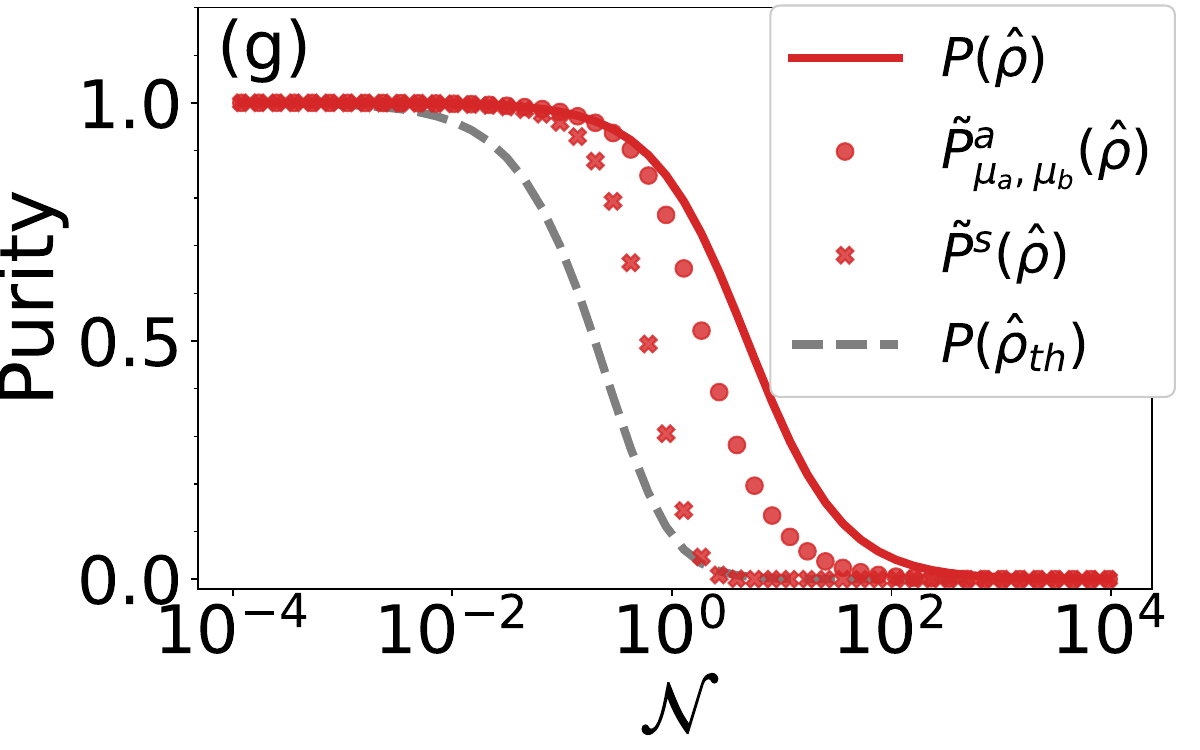}
    \includegraphics[width=0.26\linewidth]{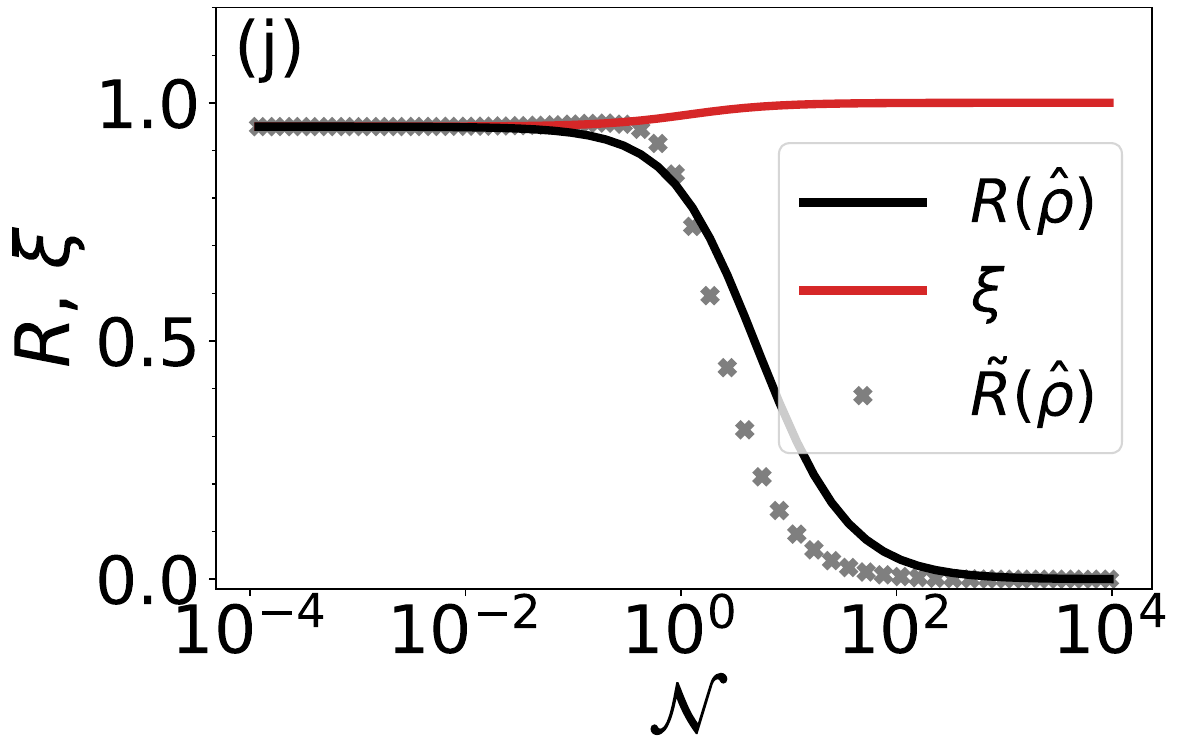}\\
    \includegraphics[width=0.2\linewidth]{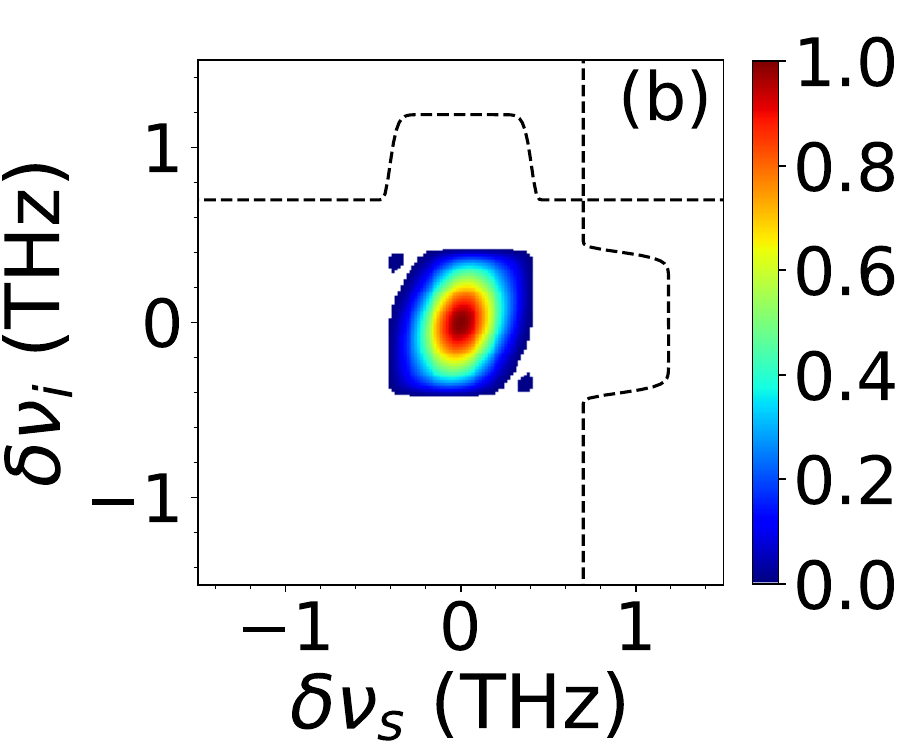}
    \includegraphics[width=0.26\linewidth]{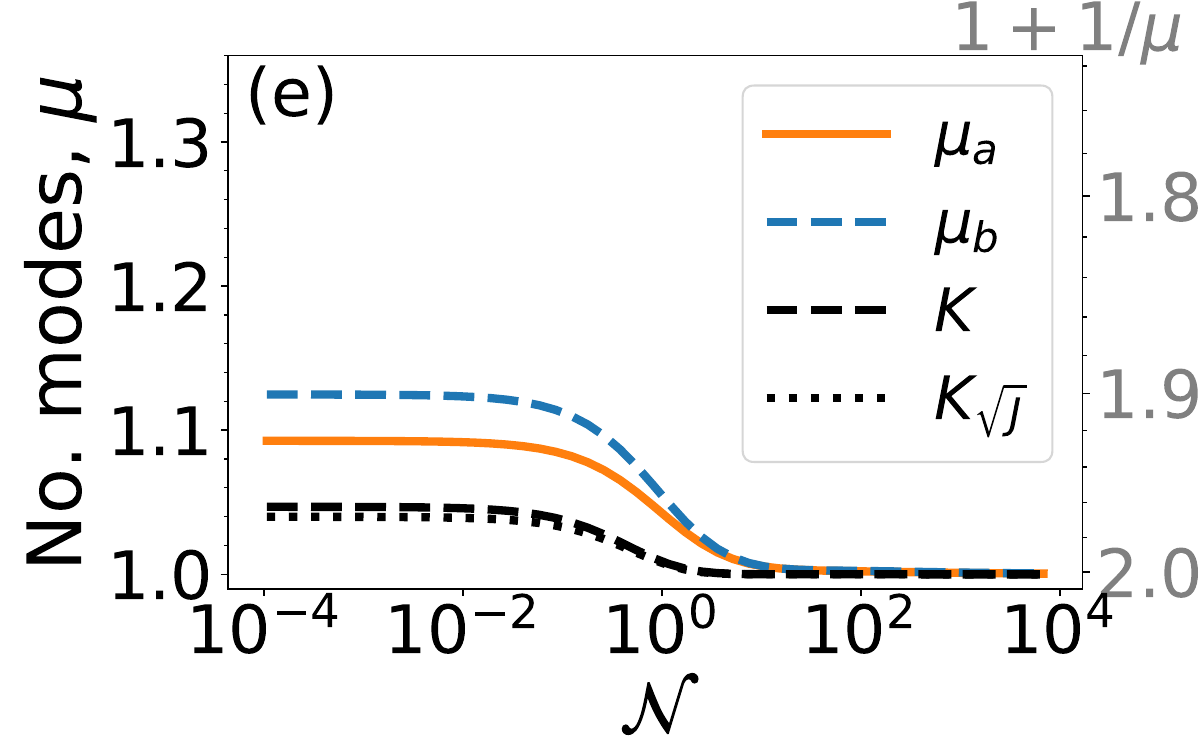}
    \includegraphics[width=0.26\linewidth]{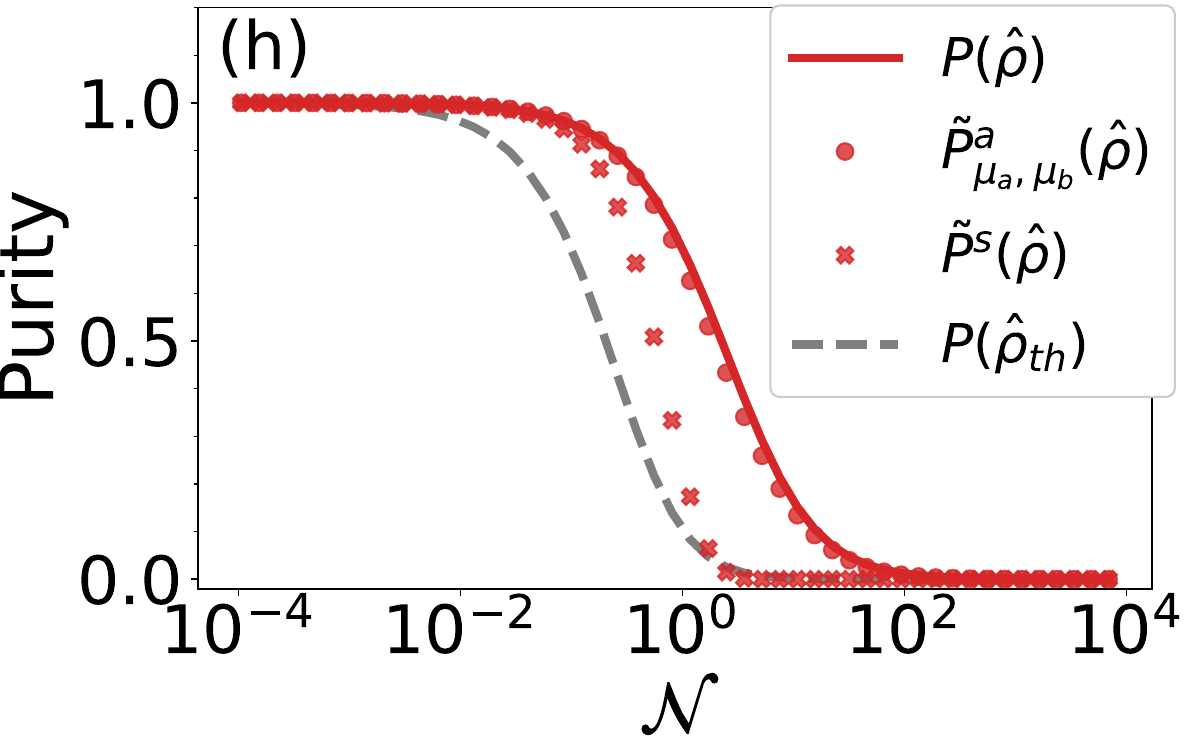}
    \includegraphics[width=0.26\linewidth]{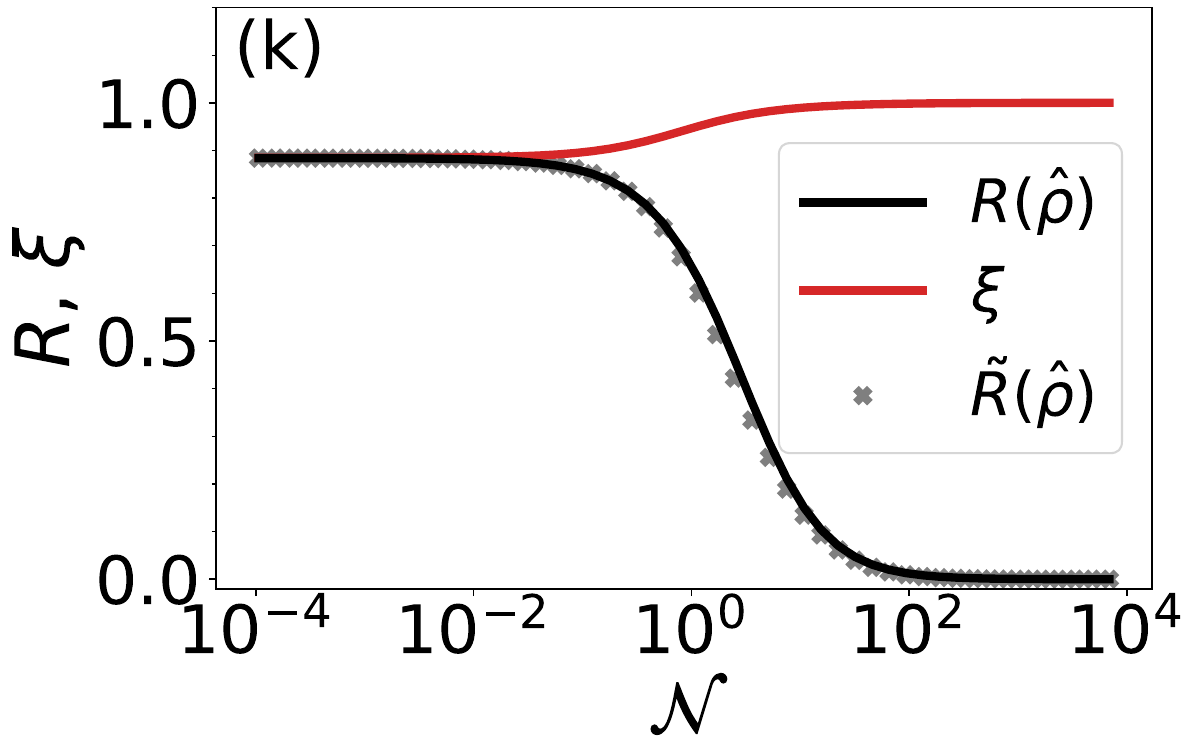}\\
    \includegraphics[width=0.2\linewidth]{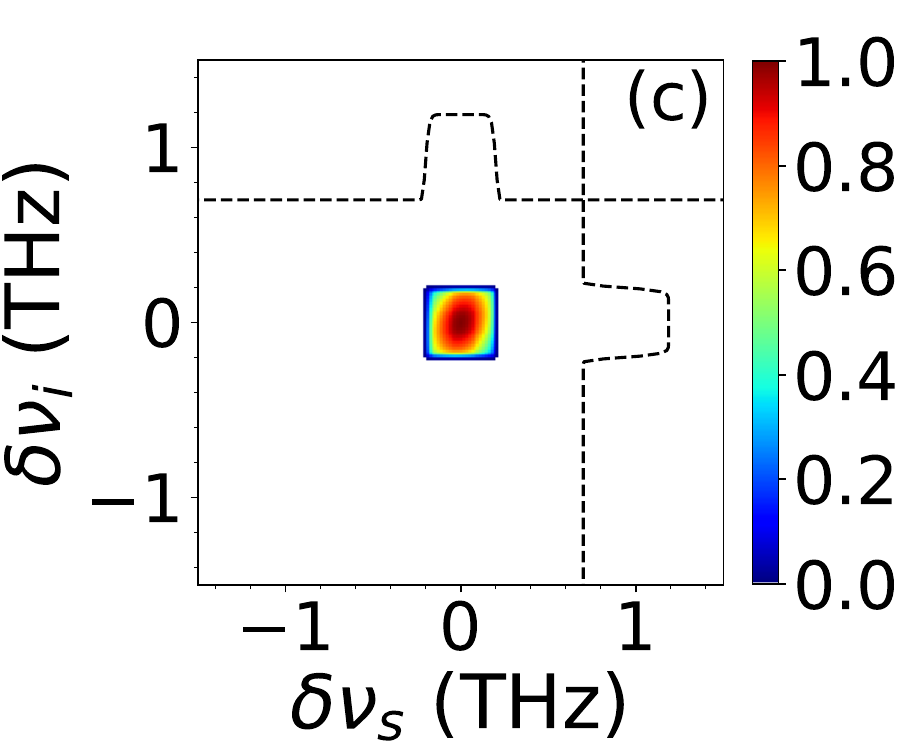}
    \includegraphics[width=0.26\linewidth]{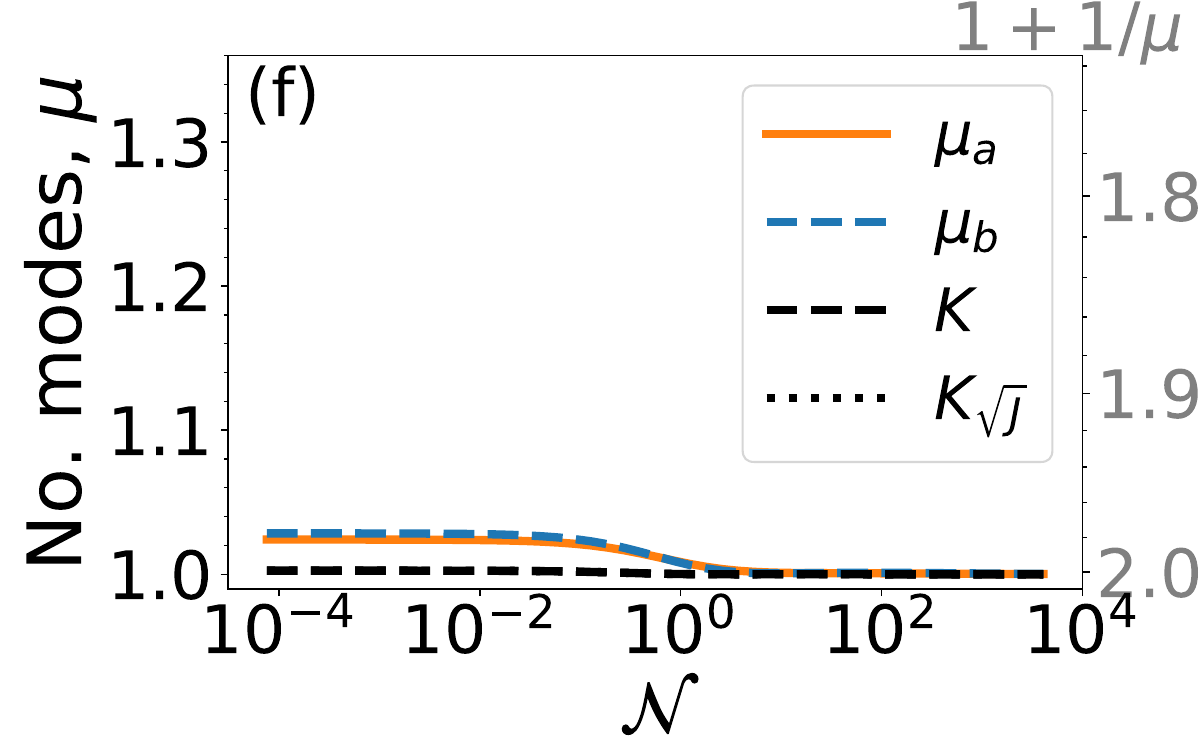}
    \includegraphics[width=0.26\linewidth]{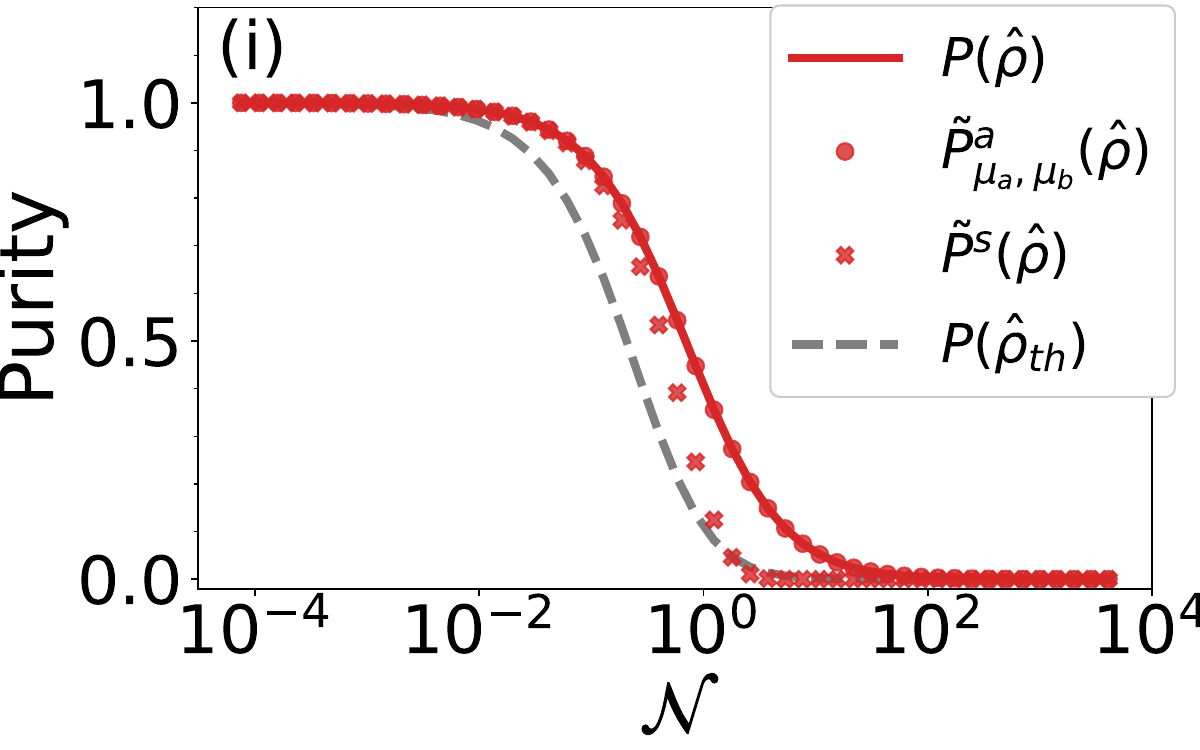}
    \includegraphics[width=0.26\linewidth]{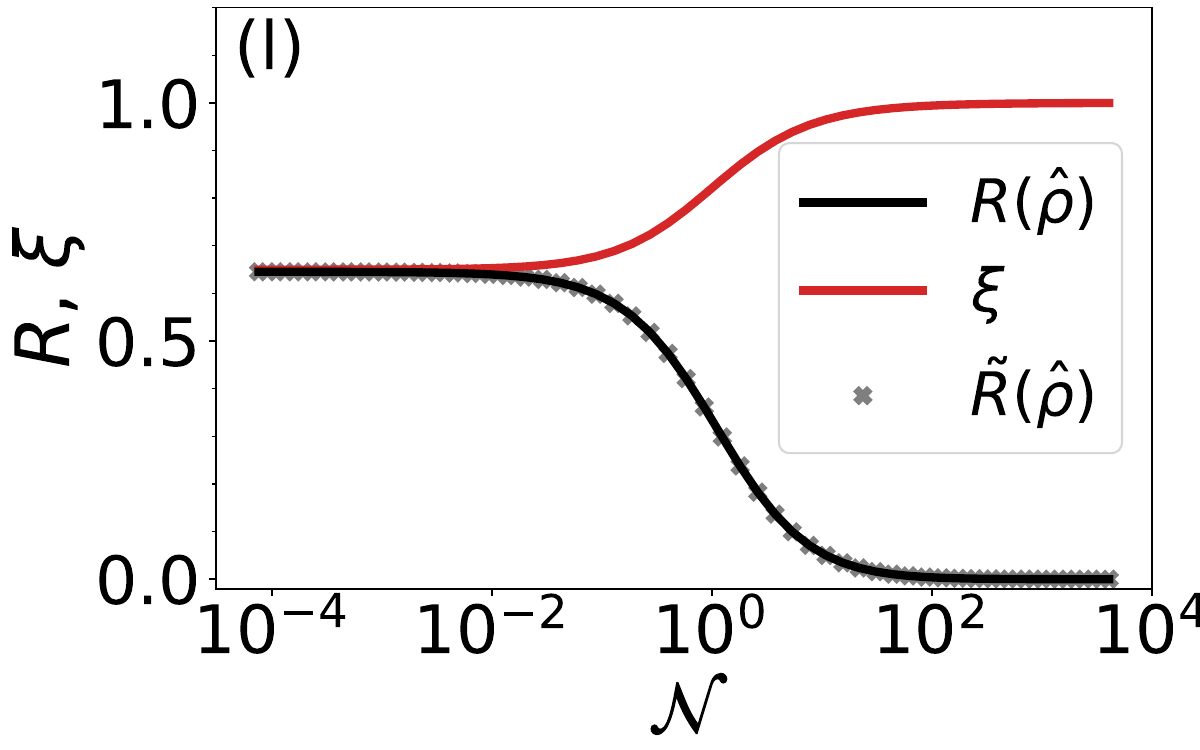}\\
    \caption{
        The numerical results for the GSID-sources `S1' (a,d,g,j), `S2' (b,e,h,k) and `S3'(c,f,i,l).  
        (a-c) The JSI for the sources `S1', `S2' and `S3'. 
        (d-f) The number of modes $\mu_{a,b}$, the number of correlated modes $K$, and the number $K_{\sqrt{J}}$.
        (g-i) The purity $P(\hat{\rho})$ (red solid), the approximate purity bound (red dots, Eq.~\eqref{eq_purity_bound_approx}) and strict purity bound (red crosses, Eq.~\eqref{eq_purity_bound_strict}) for the generated GSID-state $\hat{\rho}$ as a function of number of photons.
        The purity for the corresponding thermal state $P(\hat{\rho}_{th})$ is denoted by the dashed grey line.
        (j-l) The degree of correlation $\xi$ (red solid), the renormalized purity $R$ (black line) and the renormalized purity bound $\tilde{R}$ (grey crosses, Eq.~\eqref{eq_renormalized_purity_bound}).
    }
    \label{fig_4}
\end{figure*}

We now consider the impact of optical loss on the state exiting the generation stage, noting that this loss may be frequency dependent as is the case for optical filters.
Optical loss is included as components within the preparation stage, and we begin by assuming that all sources, `S1', `S2' and `S3', experience frequency-independent loss of 5\%.
To study the impact of frequency dependent loss, `S2' and `S3' have additional bandpass filters added to the preparation stage. 
In experiments, these bandpass filters are often chosen to reject the typical sinc squared profile side-lobes around the central peak.
For `S2' and `S3', the filter bandwidths are $\approx1.8\Delta$ and $\approx0.9\Delta$, respectively, where $\Delta$ is the FWHM of the JSI, depicted in Fig.~\ref{fig_3}(a).

In Fig~\ref{fig_4}(a-c), we present results for the JSIs, the number of modes, the purity and its bounds, and the renormalized purity for all sources.
Frequency-dependent loss results in an independent number of photons in the signal $\mathcal{N}_a$ and idler $\mathcal{N}_b$ subsystems, and as such, in these cases we define $\mathcal{N} = (\mathcal{N}_a + \mathcal{N}_b)/2$. 
According to the results in Fig.~\ref{fig_4}, we can define the three different regimes of generation according to the behavior of the parameters under investigation.
The first one is the Low Photon Number Regime (LPNR) corresponding to $\mathcal{N}\lesssim 10^{-1}$.
The second one is the High Photon Number Regime (HPNR), with $\mathcal{N}\gtrsim 10$, and the third one is the Intermediate Photon Number Regime (IPNR), where  $10^{-1}\lesssim \mathcal{N} \lesssim 10$. Let us analyze each of the regimes.

\subsubsection{Low Photon Number Regime}

Let us start with the number of modes depicted in Fig.~\ref{fig_4}(d-f).
For all the considered sources, the number of modes is approximately constant in the LPNR.  
For `S1', the number of modes $\mu_{a,b}$,  the number of correlated modes $K$, and the value $K_{\sqrt{J}}$ coincide with the initial lossless PDC (see Fig.~\ref{fig_3}(b)) because the frequency-uniform loss included as part of the preparation stage does not alter the mode structure.
Subsequently, the number of modes $\mu_{a,b}$ for `S2' and `S3' is decreased compared to `S1', as the bandpass filters present in these sources significantly modify the mode structure of the generated PDC light. 
According to Eq.~\eqref{eq_g2_2}, the different number of modes $\mu_a$ and $\mu_b$ lead to different values of second-order correlation functions $g^{(2)}_{a,b}$.

As was discussed in Sec.~\ref{sec_purity}, the purity for all the sources in LPNR is $P \approx 1$  (see Fig.~\ref{fig_4}(g-i)), and the renormalized purity~\eqref{eq_renormalized_purity} is a more informative state quality measure. 
For the sources considered here, the renormalized purities are shown in Fig.~\ref{fig_4}(j-l).
Sources `S2' and `S3' show a reduced renormalized purity when compared to source `S1'.
This occurs as the bandpass filters reduce the signal-idler correlations despite a decrease in the number of modes.

Unfortunately, the renormalized purity can not be directly measured with mode-unresolved detection; therefore, we need alternative experimentally accessible measures of the source quality.
In the LPNR, we can use both the degree of correlation and the renormalized purity bound $\tilde{R}$.
Both of them are plotted in Fig.~\ref{fig_4}(j-l), and for the studied sources, they coincide with the actual value of the renormalized purity.

Increasing the pump power of our sources, we notice that at around a mean photon number of $\mathcal{N} \approx 10^{-1}$, the source parameters begin to vary significantly, and we reach the IPNR.

\subsubsection{Intermediate Photon Number Regime}

GSID-sources operated in the IPNR are of special interest in state-of-the-art quantum applications. 
For pure single-mode squeezed light sources in the IPNR, one would expect a degree of squeezing ranging from $\approx0.1$~dB to $\approx20$~dB, which encompasses the typical operating range of squeezed light sources. 
The approach developed here allows us to perform a reliable and comprehensive characterization of the source quality in this regime -- something that is not possible using currently utilized techniques based on mode-unresolved detection. 
First, we observe that the number of modes decreases with increasing photon number $\mathcal{N}$ for all of the considered sources (see Fig.~\ref{fig_4}(d-f)) in the same way as for the initial PDC light (see Fig.~\ref{fig_3}(b)). 
Second, we also note that the purities decrease with increasing photon number: The brighter the quantum states become, the more sensitive they are to optical loss.
However, the degree of correlation $\xi$ continues to grow.
This is due to the presence of the quadratic photon number term in Eq.~\eqref{eq_degree_corr_def}, which becomes dominant for bright states.
This quadratic term can be interpreted as `classical' correlations.
Thus, in contrast to the LPNR, the degree of correlation $\xi$ in the IPNR, alone, does not characterise the quality of the GSID source.
Note that the renormalized purity bound $\tilde{R}$ (Eq.~\eqref{eq_renormalized_purity_bound}) also can not be used, because IPNR is beyond the approximation used in its derivation. 

In IPNR, both the strict~\eqref{eq_purity_bound_strict} and the approximate~\eqref{eq_purity_bound_approx} bounds become relevant.
Indeed, Fig.~\ref{fig_4}(g-i) shows that the approximate purity bound~\eqref{eq_purity_bound_approx} is higher compared to the strict bound.  
Moreover, for `S2', the approximate bound is very close to the actual state purity, and they coincide for `S3'.
To explain this, we notice that the number of modes is equal to 1, when the inequality~\eqref{eq_purity_bound_2__} turns to an equality.
This result shows that the approximate purity bound~\eqref{eq_purity_bound_approx} is extremely useful in experiments.
Note that the value $\xi$ is still required for the purity bounds.

\subsubsection{High Photon Number Regime}

In the HPNR, the number of modes for all the studied sources is close to 1.
Despite generating states that are very close to single-mode and with high degree of correlation $\xi\approx 1$, it is seen that state purity asymptotes to zero.
The purity bounds are valid in HPNR; however, they are not as informative as in IPNR.
The high sensitivity of purity to loss in HPNR requires the development of additional quality metrics and characterization methods that are beyond our framework.

\subsubsection{Purity bounds}

\begin{figure}
\centering
\includegraphics[width=\linewidth]{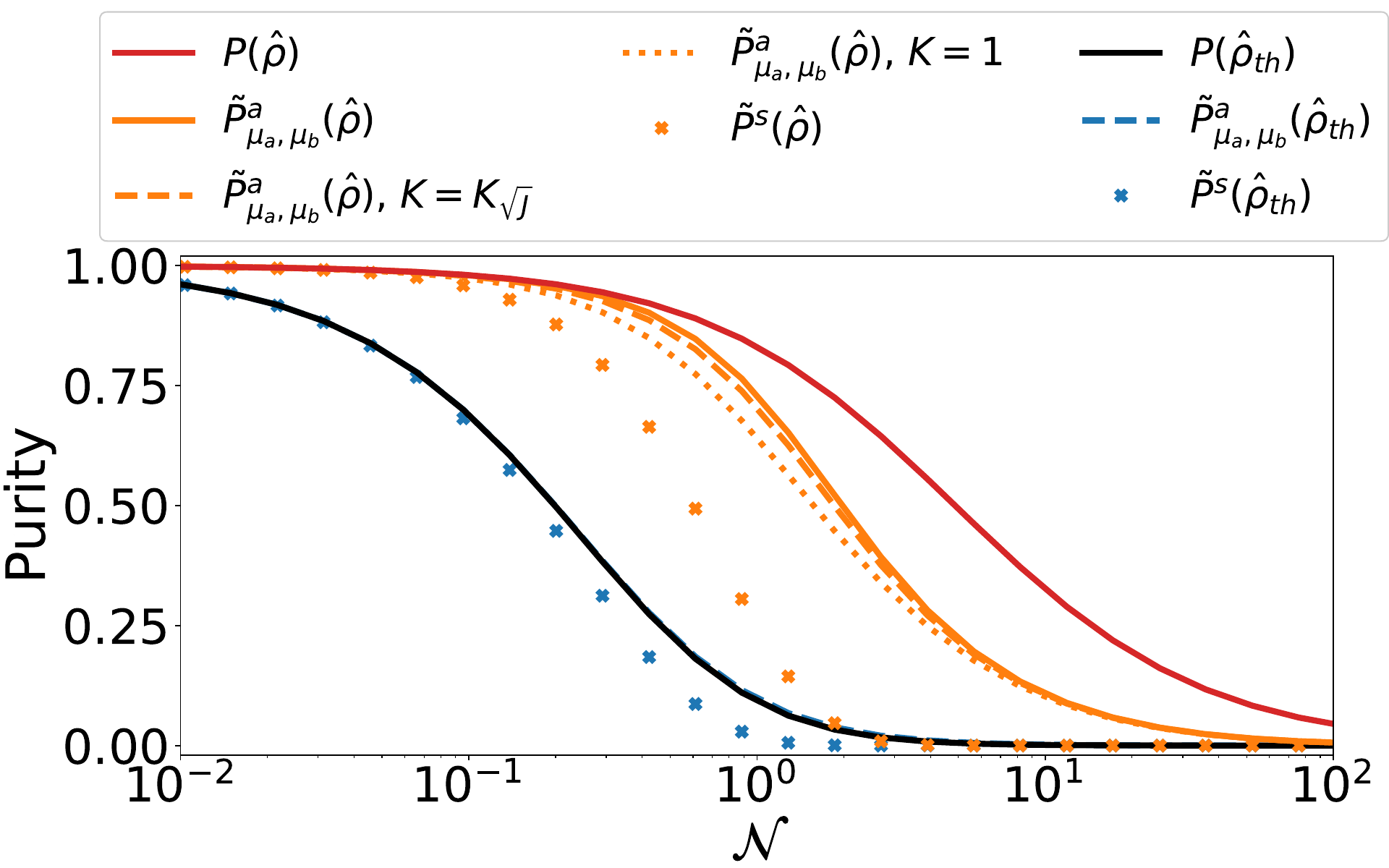}
\caption{ The purity (red line) and its bounds for the state $\hat{\rho}$ generated by $`S1'$ as a function of generated photon number. 
The approximate purity bounds~\eqref{eq_purity_bound_approx} are calculated using; the actual number of correlated modes $K$ (orange solid line), with $K$ replaced by $K_{\sqrt{J}}$ (orange dashed line), and with $K$ replaced by $1$ (dotted line). 
The strict purity bound~\eqref{eq_purity_bound_strict} is depicted with orange crosses.
The purity for the corresponding thermal state $\hat{\rho}_{th}$ (the state with the same matrices $\veccal{D}_{a,b}$ and $\veccal{C}_{a,b}=0$) is represented by the black solid line, and its approximate and strict bounds with blue dashed line and blue crosses, respectively.
   }
\label{fig_5}
\end{figure}

Let us discuss the purity bounds in detail. 
As we mentioned in Sec.~\ref{ssub_meas_corr_modes}, it is demanding to measure the parameter $K$ and, in many applications, it remains unknown.
However, in Sec.~\ref{sub_comments_bounds}, we showed that a lower bound can be found using instead the value $K_{\sqrt{J}}$ or, when $K_{\sqrt{J}}$ is also unattainable, even the value of $1$.
Here we consider such replacements for the source `S1,' and the numerical results are shown in Fig.~\ref{fig_5}. 
We see that the actual number of correlated modes $K$ gives the maximal approximate purity bound, while the values $K_{\sqrt{J}}$ and $1$ decrease it by 4\% and 14\%, respectively.
It means that even if the value $K$ is unknown, the lower purity bound can be determined, and JSI measurements can be used improve it.

Although the parameter $K$ is sufficient for purity bounds, determining the degree of correlation $\xi$ is necessary.
Indeed, the case of $\xi=0$ corresponds to thermal states and the absence of signal-idler correlation.
To show the difference between the actual state and a thermal one, we construct the corresponding thermal states $\hat{\rho}_{th}$ for the source `S1' by manually setting $\veccal{C}_{ab}=0$ (the correlation matrices $\veccal{D}_{a,b}$ remain the same). 
Their purity and purity bounds are shown in Fig.~\ref{fig_5}.

Two additional comments should be made regarding Fig.~\ref{fig_5}.
First, we see that the approximate purity bound for the thermal state is slightly violated. 
It demonstrates that the approximate bound~\eqref{eq_purity_bound_approx} is not a strict inequality.
Second, the strict purity bound~\eqref{eq_purity_bound_strict} is much more conservative compared to the approximate one.
These results illustrate the conclusions from Sec.~\ref{sub_purity_bounds} that the approximate bound is mathematically incorrect but very useful from an experimental perspective.

\subsubsection{Summary}
  
    In this section, we studied a realistic GSID-source and considered the properties of generated light in three different regimes.
    Although mode-unresolved detection is found to be insufficient for experimental characterisation of bright states in the HPNR, we derive the means by which it can provide a robust quality measure of GSID-sources operated in the LPNR and the IPNR.
    We show that the degree of correlation $\xi$ can be used as a quality measure of a GSID-source only in the LPNR; however, it is necessary but not sufficient in INPR.
    In the latter, the state purity is a good quality measure, and mode-unresolved detection allows one to determine its lower bound experimentally.
    This bound requires not only the degree of correlation $\xi$ but also the number of photons $\mathcal{N}_{a,b}$ and the number of modes $\mu_{a,b}$, and, ideally, the actual number of correlated modes $K$. 
    Reliable characterisation of the source therefore requires one to experimentally determine the value of all of these parameters.

\section{Discussion}
\label{sec_discussion}

In this section, we discuss additional aspects regarding our framework.

\subsection{Terminology}
\label{sub_terminology}

Our formalism is based on mixed Gaussian states, extending beyond the commonly used pure-state approaches.
In turn, the terminology developed and used in the simplified theories should also be revisited.

First, note that we do not use the term `two-photon amplitude' (TPA), which is meaningful only for pure two-photon states.
Indeed, for pure multimode squeezed states, we can always consider a two-photon Fock subspace and define the TPA.
The pure two-photon-state approximation holds for pure, weak squeezed states in which the two-photon Fock components dominate.
However, the approximation of an initially multimode mixed GSID-state via the pure two-photon state is incomplete and misleading, and therefore the TPA remains undefined.
The terms `spectral purity', `joint spectral amplitude', and `Schmidt modes' were initially introduced to describe two-photon states.
Thus, we should revise them for use in our holistic framework.

First, the term `spectral purity' appears from the description of bipartite two-photon states. It coincides with the inverse number of modes and indicates the quantum purity of the state in a first subsystem obtained after tracing out the second one~\cite{MeyerScott2017}.
This does not hold for GSID-states that have undefined photon number.
Therefore, the term `spectral purity' is misleading and should be avoided for general GSID-states.

Second, in the two-photon approximation, the joint spectral intensity is equal to the squared modulus of the two-photon amplitude.
This is why the two-photon amplitude is often referred to as the joint spectral amplitude.
In our approach, we define the JSI via frequency-dependent measurements and it is equal to the squared modulus of the correlation function $\veccal{C}_{ab}$.
Therefore, we use the 
term `JSA' for the correlation function $\veccal{C}_{ab}$. 
Note also that, in our framework, the JSA and the coupling function in the Hamiltonian Eq.~\eqref{eq_hamiltonian} do not coincide. 

Third, in the two-photon approximation, the signal photon number spectrum is given by integrating the JSI over the idler frequencies.
In our framework, the JSI is a property of the signal-idler correlations, while the photon number spectra of signal and idler fields are the properties of the individual subsystems.
They are defined via different correlation matrices and therefore are generally independent.
This means that the photon number spectra can not be expressed via the JSI for arbitrary GSID-states and should be measured independently from JSI.

Fourth, the common definition of `Schmidt modes' as a decomposition of TPA does not apply to our formalism.
For pure Gaussian states, the term `Schmidt modes' is applied to a unique basis obtained via the Bloch-Messiah reduction~\cite{Braunstein_2005,Wasilewski_2006,Cariolaro_2016,Horoshko_2019}.
However, the Bloch-Messiah reduction does not apply to mixed Gaussian states, which necessitates the use of alternative decompositions, suggested e.g. in Refs.~\cite{Simon_1994,Wolf_2008,Weedbrook_2012,Kopylov_2025,Hubenschmid_2026}.
In turn, different decompositions give different bases with different properties~\cite{Kopylov_2025,Kopylov_2025_APL}.
In this work, we have seen this by introducing the MW-basis and the C-basis in Sections~\ref{ssub_mw_basis} and \ref{ssub_c_basis}, respectively. 
Due to this variety in basis choice, for mixed states, we explicitly specify the bases and avoid calling them a `Schmidt'-basis.
We use the term `Schmidt-basis' only for pure states when the unique `good' basis exists.

    \subsection{Spatial degrees of freedom}
    \label{ssec_spatial_degrees}

    We initially required that the GSID-source should operate as a spatially single-mode source (see requirement R5).
    This requirement significantly simplified the description and presentation of results. Nevertheless, the developed formalism is also valid for fields that are multimode both in space and time.
    Indeed, in this case, we should consider the correlation matrices $\veccal{D}_{a,b}$ and $\veccal{C}_{ab}$ that depend both on temporal and spatial degrees of freedom: The matrix index runs both over the spatial and temporal modes.
    The number of modes $\mu_{a,b}$ and the number of correlated modes $K$ refer then to the joint spatio-temporal modes.
    However, the JSI is a more complicated function which will depend both on spatial and temporal degrees of freedom, making its measurement challenging.

    From a theoretical perspective, the presence of spatial modes does not change the analysis and derivations; thus, the developed formalism and state characterisation are also valid for spatially-temporal multimode GSID-states.

\subsection{Direct source optimization } 
\label{sec_source_optimization}
        
    One of the advantages of our holistic approach is that we treat the GSID-source, including both state generation and state preparation stages, as a complete system.
    This allows us to focus entirely on the properties of the generated light and neglect the impact of both imperfections and imprecise characterisation of individual elements in the complete state generation.
    In other words, we take a black-box approach to the GSID-source that does not require characterisation of individual components or additional modeling.
    Having control over some of these components, e.g., the pump pulse profile or the widths of the bandpass filters, allow one to experimentally optimize the quality of the sources, without needing to perform accurate numerical simulations that require precise characterisation of these components.

\subsection{Beyond basis-independent characteristics}

    In this work, we developed a GSID-source characterisation technique based on mode-unresolved measurements.     
    Unfortunately, such characterisation is incomplete: it does not provide detailed information regarding internal mode structure.    
    Therefore, we can ask ourselves: When is our approach helpful and sufficient?
    
    Quantum protocols operate with different type of detectors and any reliable source characterisation must be consistent with the detection scheme used in the protocol. 
    In particular, homodyne detection require knowledge of the internal mode structure.     
    This is because homodyne detection itself acts as a filter: the state is `projected' on the local-oscillator mode, and all other modes do not contribute to a measured signal.
    The optimal homodyne measurements should be based on the mode structure of the measured light~\cite{Kopylov_2025, Kopylov_2025_APL}.
    Therefore, our source characterisation technique is not sufficient for homodyne measurements: is sensitive to the mode structure but does not provide details regarding optical mode profiles and their distribution.

    Nevertheless, many quantum protocols combine squeezed light with photon-number-resolving detectors, for example in Gaussian boson sampling \cite{Hamilton2017,stefszky2026}. 
    Using the same type of detectors, our approach is fully consistent and compatible with them.
    This means that our source characterisation technique captures and reveals the light properties that affect the protocol's outcomes.

\section{Conclusions}

    In this work, we studied the properties of multimode mixed squeezed light, generated by a PDC-based source that produces fully distinguishable signal and idler fields.
    Starting from the Heisenberg uncertainty relation for the Gaussian covariance matrix, we presented an uncertainty relation for the field correlation matrices of GSID states.
    We found that the Mercer-Wolf basis does not necessarily diagonalize signal-idler correlations, and introduced the C-basis for their study.
    As a result, we show that the mode structure of mixed GSID states is characterized by three parameters: the number of modes in the signal subsystem $\mu_a$, the number of modes in the idler subsystem $\mu_b$, and the number of correlated modes $K$ between the signal and idler subsystems.
   This result shows that the two existing characterisation methods based on the second-order correlation functions and on the JSA are not complete: They yield different physical quantities of independent properties of the state.
   
Our analysis shows that JSI measuremetns, which are often used to determine the degree of correlation of a given source, provides only a lower bound on the actual number of correlated modes $K$.
Our work therefore demonstrates that JSI measurements are not sufficient for reliable source characterisation.
    
To address the shortcomings of the aforementioned characterisation methods, we develop a complete, holistic approach that unifies and extends upon these known methods. We reveal that the lower bound of the state purity is a suitable measure for GSID-source quality and show how it can be experimentally determined via mode-unresolved measurements.

The developed framework provides a robust and reliable characterisation of mixed multimode squeezed states, which goes beyond the typical assumptions and restrictions. This treatment reveals the limitations of current sources which is critical when assessing the impact that imperfections will have on the most demanding quantum photonics applications.

\section*{Acknowledgments}
This work has received funding from the German Federal Ministry of Research, Technology and Space within the PhoQuant project (Grant No. 13N16103).

\newpage

\appendix 

\section{Mathematical identities}
    \label{app_math}
    \subsection{Schur complement}

    Here we briefly present the main properties of the Schur complement.   
    For a detailed study, see, e.g., Ref.~\cite{Schur_complement_book_2005}. 

    Let us consider a block matrix of the form 
    \begin{equation}
        M = \begin{pmatrix}
                    A &  B      \\
                    C &  D   
                \end{pmatrix},
    \end{equation}
    where $A \in \mathbb{C}^{n\times n}$, $B \in \mathbb{C}^{m\times m}$, $C \in \mathbb{C}^{n\times m}$ and $D \in \mathbb{C}^{m\times n}$.
    If matrix $A$ is invertible, then matrix 
    \begin{equation}
         S_A = D - C A^{-1} B
    \end{equation}
    is defined and is called the Schur complement of $M$ relative to $A$.  
    Similarly, if matrix $D$ is invertible, then matrix
    \begin{equation}
        S_D = A - B D^{-1} C 
    \end{equation}
    is the Schur complement of $M$ relative to $D$.
    Further, when the Schur complements $S_A$ or $S_D$ are used, we always assume that they are well-defined. 
    Otherwise, the considered equations are not applicable.

    Schur's determinant formula (see Theorem 1.1 in Ref.~\cite{Schur_complement_book_2005}), applied to the matrix $M$, leads to two equivalent expressions for determinant
    \begin{align}
        \label{eq_schur_determinant_1}
        \det(M) & =\det(A)\det \left( S_A \right) ,\\
        \label{eq_schur_determinant_2}
        \det(M) & =\det(D)\det \left( S_D \right).
    \end{align}
    In this paper, we are interested in a special case of the matrix $M$, when it has a block form
    \begin{equation}
        M = 
        \left(
                \begin{array}{cc|cc}
                A_1 & 0      & 0   & B_1  \\
                0   & A_2    & B_2 & 0  \\
                \hline
                  0   & C_1  & D_1 & 0    \\
                  C_2 & 0    &  0 & D_2   \\
             \end{array} 
        \right),
        \label{eqap_matrix_form}
    \end{equation}
        where $A_1, D_1 \in \mathbb{C}^{n\times n}$; $A_2, D_2 \in \mathbb{C}^{m\times m}$;  $B_1, C_1 \in \mathbb{C}^{n\times m}$ and $B_2, C_2 \in \mathbb{C}^{m\times n}$.
        Then, the Schur complements of matrix $M$ relative to matrices $A=\left(
                    \begin{smallmatrix}
                     A_1 & 0      \\
                    0   & A_2  
                    \end{smallmatrix}
                        \right)$ 
                and 
                    $D =\left(
                    \begin{smallmatrix}
                     D_1 & 0      \\
                    0   & D_2  
                    \end{smallmatrix}
                        \right)$ are
        \begin{equation}
             S_A 
            = \begin{pmatrix}
                S_1^A & 0     \\
                0   & S_2^A   \\
            \end{pmatrix} ~ ~ \text{and} ~ ~
            S_D = \begin{pmatrix}
                        S_1^D & 0     \\
                        0   & S_2^D    \\
                    \end{pmatrix},
        \end{equation}
        respectively; and matrices 
        \begin{align}
            S_1^A &= D_1 - C_1 A_2^{-1} B_2 , \\
            S_2^A &= D_2 - C_2 A_1^{-1} B_1 , \\
            S_1^D &= A_1 - B_1 D_2^{-1} C_2 , \\
            S_2^D &= A_2 - B_2 D_1^{-1} C_1 .
        \end{align}

        Further, we use the well-known properties of the Schur complement and apply them to our matrix in the form of Eq.~\eqref{eqap_matrix_form}.
        The obtained equations are formulated as the following lemmas:                          
        \begin{lemma}
        \label{lemma_determinant}
                If matrices $A_{1,2}$ in Eq.~\eqref{eqap_matrix_form} are invertible, 
                \begin{equation}
                    \det{M} =  \det \begin{pmatrix}
                A_1 & B_1    \\
                C_2 & D_2   \\
            \end{pmatrix} \det \begin{pmatrix}
                 D_1 & C_1     \\
                 B_2 & A_2  \\
            \end{pmatrix}.
                \end{equation}
         \end{lemma}
            \begin{proof}

            Using Schur's determinant formula~\eqref{eq_schur_determinant_1}, we get
            \begin{multline}
                 \det(M) =\det \begin{pmatrix}
                    A_1 & 0     \\
                    0   & A_2   \\
                \end{pmatrix} 
                \det \begin{pmatrix}
                    S_1^A & 0     \\
                    0   & S_2^A   \\
                \end{pmatrix}  
                \\
                = \det(A_1)\det(A_2)\det(S_1^A)\det(S_2^A).
            \end{multline}
            Combining the first and the last factors, we can apply Eq.~\eqref{eq_schur_determinant_1} to them in the opposite direction.
            As a result, we get
            \begin{equation}
              \det(A_1)\det(S_2^A) 
              = \det 
               \begin{pmatrix}
                    A_1 & B_1    \\
                    C_2 & D_2   \\
                \end{pmatrix}.  
            \end{equation}
        Similarly, Eq.~\eqref{eq_schur_determinant_2} can be applied to the second and third factors, resulting in   
            \begin{equation}
                        \det(A_2)\det(S_1^A) 
                =
                \det \begin{pmatrix}
                     D_1 & C_1     \\
                     B_2 & A_2  \\
                \end{pmatrix}.
            \end{equation}
        Combining all these equations, we prove the desired identity.
            \end{proof}

        \begin{lemma}
        \label{lemma_schur_positivity}
                If matrix $M$ (Eq.~\eqref{eqap_matrix_form}) is Hermitian and matrices $A_{1,2}\succ0$ (are positive definite), then matrix
                \begin{equation}
                     M\succeq 0 ~ \text{ if and only if} ~
                            \begin{cases}
                                    S_1^A \succeq 0, \\
                                    S_2^A \succeq 0. 
                            \end{cases}
                \end{equation}
                Similarly, 
                \begin{equation}
                     M\succ 0 ~ \text{ if and only if} ~
                            \begin{cases}
                                    S_1^A \succ 0, \\
                                    S_2^A \succ 0. 
                            \end{cases}
                \end{equation}
        \end{lemma}
            \begin{proof}
                This lemma is a direct corollary of Theorem 1.12 in Ref.~\cite{Schur_complement_book_2005}, applied to the matrix~\eqref{eqap_matrix_form}.
            \end{proof}

        \begin{lemma}
        \label{lemma_schur_trace}
                If matrix $M$ (Eq.~\eqref{eqap_matrix_form}) is Hermitian and positive semi-definite ($M\succeq 0$), the following inequality holds:
                \begin{equation}
                    |\trace{ B^*_2 B_1  }|^2 \leq \trace{ A^*_2  D_2 } \trace{ D^*_1 A_1  }.
                \end{equation}
        \end{lemma}
            \begin{proof}
                The initial matrix $M$ is Hermitian; therefore, the matrices $C_{1,2}=[B_{2,1}]^H$, $A_{1,2}=A_{1,2}^H$ and $D_{1,2}=[D_{1,2}]^H$.
                Applying Lemma~\ref{lemma_schur_positivity} to the positive semi-definite matrix $M$, we get that  matrices $S_{1,2}^A\succeq 0 $.
                In turn, both matrices $S_{1}^A$ and $S_{2}^A$ can be considered as the Schur complements of matrices
                \begin{equation}
                    M_1 = \begin{pmatrix}
                            A_2    & B_2        \\
                            B_2^H  & D_1       \\
                        \end{pmatrix} 
                        ~ \text{and} ~  
                    M_2 = \begin{pmatrix}
                            D_2  & B_1^H     \\
                            B_1  &  A_1      \\
                        \end{pmatrix},
                \end{equation}
                respectively.
                Matrices $M_{1,2}$ are Hermitian, their Schur complements are positive semi-definite; therefore, $M_{1,2} \succeq 0$.
                Note that the complex-conjugated matrices $M^*_{1,2}$ and hermitian-conjugated matrices $M^H_{1,2}$ are also positive semi-definite.

                It is known that, having two positive semi-definite matrices, we can construct a trace inequality for their submatrices. 
                Taking the matrices $M_1^*=\left(
                    \begin{smallmatrix}
                                A^*_2    & B^*_2        \\
                                B_2^T  & D^*_1       \\                    \end{smallmatrix}
                        \right)$ and $ M_2 $, and applying Eq.~(3.1.13) from Ref.~\cite{Schur_complement_book_2005})), we get the desired inequality.

            \end{proof}

\subsection{Trace inequalities for positive (semi-) definite matrices}

For matrices $A,B \in \mathbb{C}^{n\times m}$ and any positive definite matrix $X \succ 0, \in \mathbb{C}^{n\times n}$, the following inequality holds (see Chapter 3 in Ref.~\cite{Schur_complement_book_2005})
    \begin{equation}
        |\trace{A^H B}|^2 \leq \trace{A^H X A} \trace{B^H X^{-1} B}.
        \label{eqap_trace_ineq_1}
    \end{equation}
In the special case of $X=\vec{1}_{n}$,
    \begin{equation}
        |\trace{A^H B}|^2 \leq \trace{A^H A} \trace{B^H B}.
        \label{eqap_trace_ineq_2}
    \end{equation}

The following lemmas are necessary for the further study: 
    \begin{lemma}
    \label{lemma_trace_1}
        For positive semi-definite Hermitian matrices $A, B \in \mathbb{C}^{n\times n}$ 
        the following inequalities holds: 
            \begin{equation}
                    \trace{A B} \leq \dfrac{\trace{A} \trace{B}}{\sqrt{\mu_A \mu_B}}  \leq  \trace{A} \trace{B},
            \end{equation}
            where $ \mu_A = \dfrac{[\trace{A}]^2}{\trace{A^2}} \geq 1$ and $ \mu_B = \dfrac{[\trace{B}]^2}{\trace{B^2}} \geq 1$.
     \end{lemma}
            \begin{proof}
            Let us start with the inequality Eq.~\eqref{eqap_trace_ineq_2}.
            Using $A^H = A$ and $B^H = B$, and taking the square root of the obtained expression, we get
            \begin{equation}
                 \trace{A B} \leq  \sqrt{ \trace{A A} } \sqrt{ \trace{B B} }.
            \end{equation}
            Here we also used the fact that the trace of the product of two positive semi-definite Hermitian matrices is real and positive; therefore, $|\trace{A B}| = \trace{A B} \geq 0$.

            Writing the traces as $\trace{A}=\trace{\Lambda_A}$ and $\trace{A^2}=\trace{\Lambda_A^2}$, where $\Lambda_A$ is the diagonal matrix built from eigenvalues $\set{\lambda_i}$ of matrix $A$, we define
            \begin{equation*}
                 \dfrac{1}{\mu_A} \equiv \dfrac{\trace{\Lambda_A^2}}{[\trace{\Lambda_A}]^2} = \dfrac{\sum_i^n [\lambda_i]^2 }{ [\sum_i^n  \lambda_i  ]^2} = \sum_i^n \bigg[  \dfrac{\lambda_i}{ \sum_j^n  \lambda_j}   \bigg]^2  \leq 1.
            \end{equation*} 
            Thus, $\mu_{A}\geq 1$.
            The value $\mu_{B}$ is defined similarly for matrix $B$.
            Therefore, the first inequality can be written as
            \begin{align*}
                 \trace{A B} & \leq \sqrt{ \trace{A^2} } \sqrt{ \trace{B^2} } = \dfrac{\trace{A} \trace{B}}{\sqrt{\mu_A \mu_B}}  \\  & \leq \trace{A} \trace{B}.
             \end{align*} 
            \end{proof}

    \begin{lemma}
    \label{lemma_trace_2}
            For an arbitrary matrix $A \in \mathbb{C}^{n\times m}$ and a  Hermitian positive semi-definite matrix $X \succeq 0, \in \mathbb{C}^{n\times n}$, the following inequality holds: 
            \begin{multline}
               \trace{A^H (X + \vec{1}_n)^{-1} A} \geq   \dfrac{  \sqrt{K\mu} \: \trace{A^H A} }{  \trace{X}  +  \sqrt{K\mu}   } \\ \geq  \dfrac{ \trace{A^H A} }{  \trace{X} + 1  },
            \end{multline}
            where 
    \begin{equation}
        K = \dfrac{[\trace{A A^H  }]^2}{\trace{A A^H  A A^H}} \geq 1 ~ ~ \text{ and } ~ \mu = \dfrac{[\trace{X}]^2}{\trace{X^2}} \geq 1.
        \label{eq_app_Kmu_defin}
    \end{equation}

     \end{lemma}
            \begin{proof}
            The matrix $X + \vec{1}_n$ is positive definite; therefore, we can start with Eq.~\eqref{eqap_trace_ineq_1}
                \begin{equation}
                    |\trace{A^H A}|^2 \leq \trace{A^H (X + \vec{1}_n) A} \trace{A^H (X + \vec{1}_n)^{-1} A}.
                    \label{eq_app_tin_der_1}
                \end{equation}
            Its first factor in the right-hand side is
            \begin{equation}
                \trace{A^H (X + \vec{1}_n) A} = \trace{A A^H X } + \trace{A^H A  }.
            \end{equation}
            Applying Lemma~\ref{lemma_trace_1} to the first term, we get
            \begin{equation}
                \trace{A A^H X } \leq \dfrac{1}{\sqrt{K \mu}} \trace{A A^H  } \trace{X},
            \end{equation}         
             where $K, \mu \geq 1$ and their  explicit form is given in Eq.~\eqref{eq_app_Kmu_defin}. 
            Using the last inequality, Eq.~\eqref{eq_app_tin_der_1} reads as
            \begin{widetext}
            \begin{align}
                 |\trace{A^H A}|^2 & \leq (\trace{A A^H X } + \trace{A^H A  }) \trace{A^H (X + \vec{1}_n)^{-1} A} \\
                 &  \leq \left(\dfrac{1}{\sqrt{K \mu}} \trace{A A^H  } \trace{X} + \trace{A^H A  }\right) \trace{A^H (X + \vec{1}_n)^{-1} A},
            \end{align}
            \end{widetext}
            Taking into account that $|\trace{A^H A}| = \trace{A^H A}\geq 0$, we prove the first desired inequality
            \begin{equation}
               \trace{A^H (X + \vec{1}_n)^{-1} A} \geq   \dfrac{ \trace{A^H A} }{   \dfrac{ \trace{X} }{ \sqrt{K\mu} }  + 1  }.
            \end{equation}
            The second result follows straightforwardly by taking $K=\mu = 1$.

            \end{proof}

\subsection{Trace-determinant inequalities}
    In this section, we derive two inequalities that will be used for derivations of the purity bounds.
    Let us consider a Hermitian matrix $\vec{M} \in \mathbb{C}^{N\times N}$, such that $\vec{M}+\vec{1}_N\succ0$.
    Then,  
    \begin{multline}
          \det(\vec{M}+\vec{1}_N) =  \exp{(\trace{\ln(\vec{M}+\vec{1}_N)})}  \\ =\exp\Big[\sum_i^N \ln(\lambda_i + 1)\Big],
          \label{eq_app_det_trace_init}
    \end{multline}
     where we used the well-known equation $\det(e^{\vec{A}}) = e^{\trace{\vec{A}}}$ for Hermitian matrix $\vec{A}$ and the eigenvalues $\set{\lambda_i}$ of the matrix $\vec{M}$.

     \subsubsection{First inequality}
        Using $\ln(x+1) \leq x$, we get that 
    \begin{equation}
        \sum_i^N \ln(\lambda_i + 1) \leq \sum_i^N \lambda_i =  \trace{\vec{M}}
    \end{equation}
    and therefore Eq.~\eqref{eq_app_det_trace_init} reads as
    \begin{equation}
        \det(\vec{M}+\vec{1}_N) \leq \exp(\trace{\vec{M}}).
        \label{eqap_det_first_est}
    \end{equation}

   \subsubsection{Second inequality}
        Here, both sides of Eq.~\eqref{eq_app_det_trace_init} are raised to the power of $\frac{1}{N}$
            \begin{equation}
                \left[\det(\vec{M}+\vec{1}_N)\right]^{\frac{1}{N}} = \exp \left({\frac {\sum_{i}^{N}\ln  ( \lambda_i + 1 )}{N}}\right) 
            \end{equation}
   Using Jensen's inequality  
    \begin{equation}
       \exp \left({\frac {\sum_{i}^{N}\ln  ( \lambda_i + 1 )}{N}}\right) \leq {\frac {\sum_{i}^{N}(\lambda_i + 1)}{N}}, 
       \label{eq_app_jensen}
    \end{equation}
    we get 
    \begin{equation}
        \det(\vec{M}+\vec{1}_N) \leq  \bigg[ \frac{ \trace{\vec{M}} + N }{N} \bigg]^N.
    \end{equation}

 If some of the eigenvalues equal zero, the sum over all of them results in an overestimated value.    Therefore, we can exclude these zero-values from the sum Eq.~\eqref{eq_app_jensen}.
Then, the resulting inequality takes the form
\begin{equation}
        \det(\vec{M}+\vec{1}_N) \leq  \bigg[ \frac{ \trace{\vec{M}} + n }{n} \bigg]^n,
        \label{eqap_det_second_est}
    \end{equation}
    where $n\leq N$ is the number of nonzero eigenvalues of matrix $\vec{M}$.

\subsection{ Derivatives of multivariate Gaussian function }
    \label{app_derivatives_Gaussian}

        In what follows, we need to compute the partial derivatives $
            \prod_{ij} \partial^{l_i}_{\beta_i}  \partial^{k_j}_{\beta^*_j} \: \Phi(\vec{\beta}, \vec{\beta^*}) \big|_{\vec{\beta}=\vec{\beta^*}=0} $,
        where the variables are the column-vectors $\vec{\beta} = (\beta_1, ..., \beta_M)^T$, $\vec{\beta^*} = (\beta^*_1, ..., \beta^*_M)^T$, $\partial^{n}_{x}\equiv \frac{\partial^{n}}{(\partial x)^{n}}$, the function 
        $
            \Phi(\vec{\beta}, \vec{\beta^*}) =  e^{P(\vec{\beta}, \vec{\beta^*})},
        $ 
        and $P(\vec{\beta}, \vec{\beta^*})$ is a second-order polynomial in the form
        \begin{equation}
            P(\vec{\beta}, \vec{\beta^*}) = \vec{\beta}^T \vec{A} \vec{\beta} + \vec{\beta}^T \vec{B} \vec{\beta^*} + (\vec{\beta^*})^T \vec{C} \vec{\beta^*} .
            \end{equation}
        Here the matrices $\vec{A}$, $\vec{B}$ and $\vec{C}$ are the complex $M\times M$ matrices with matrix elements ${A}_{ij}$, ${B}_{ij}$ and ${C}_{ij}$, respectively.
        The first-, second- and third- order derivatives are
        \begin{align}
            \partial_{\beta_{i}} \Phi(\vec{\beta}, \vec{\beta^*}) \big|_{\vec{\beta}=\vec{\beta^*}=0} & =  0, \\
            \partial_{\beta_{i}} \partial_{\beta_{j}} \Phi(\vec{\beta}, \vec{\beta^*}) \big|_{\vec{\beta}=\vec{\beta^*}=0}  & = A_{ij} + A_{ji} , \\
               \partial_{\beta^*_i}\partial_{\beta^*_j} \Phi(\vec{\beta}, \vec{\beta^*}) \big|_{\vec{\beta}=\vec{\beta^*}=0}  & = C_{ij} + C_{ji} ,  \\
               \partial_{\beta_i}\partial_{\beta^*_j} \Phi(\vec{\beta}, \vec{\beta^*}) \big|_{\vec{\beta}=\vec{\beta^*}=0}
             & = B_{ij},
             \\
               \partial_{\beta_i}\partial_{\beta_j}\partial_{\beta^*_k} \Phi(\vec{\beta}, \vec{\beta^*}) \big|_{\vec{\beta}=\vec{\beta^*}=0}
             & = 0 .
        \end{align}
        The fourth-order derivative is
        \begin{multline}
            \partial_{b_z} \partial_{b_x}\partial_{b^*_w}\partial_{b^*_y} \Phi \big|_{\vec{b}=\vec{b^*}=0}
             = (C_{wy} + C_{yw})(A_{zx}+A_{xz}) \\ +  B_{zw}  B_{xy} + B_{xw} B_{zy} .
         \end{multline}

\section{Formalism of Gaussian states }
\label{app_Gaussian_states}
    
\subsection{ General description }
\label{app_gaussian_general}

    Multimode Gaussian states are typically described in terms of the covariance matrix, built from the second moments of quadrature operators~\cite{Simon_1994,Weedbrook_2012}.
    For the $M$-mode bosonic system, the quadrature operators can be written in a compact vector form as $\vec{x}  = (\hat{q}_1, \hat{q}_2, \dots, \hat{q}_M; \hat{p}_1, \hat{p}_2, \dots, \hat{p}_M)^T$ with the bosonic commutation relation $[\hat{x}_i, \hat{x}_j ] = i \hbar [\vec{\Omega}]_{ij}$, where
    $
        \vec{\Omega} = \left(
                    \begin{smallmatrix}
                      \vec{0}   & \vec{1}_M  \\
                     -\vec{1}_M & \vec{0}  
                \end{smallmatrix} 
            \right)
    $.
    For a quantum state $\hat{\rho}$, the matrix of quadrature moments is $\vec{X}  = \braket{ \vec{\hat{x}} \cdot \vec{\hat{x}}^T } $.
    Here and further, the notation $\braket{ \vec{\hat{x}} \cdot \vec{\hat{x}}^T } $ indicates the matrix with matrix elements $\braket{ \hat{x}_i \hat{x}_j  } \equiv \trace{\hat{x}_i \hat{x}_j \hat{\rho}} $; $\braket{\vec{x}}$ denotes the column-vector with the elements $\braket{ \hat{x}_i  } \equiv \trace{\hat{x}_i \hat{\rho}} $.
    The symmetric-ordered covariance matrix (known also as the Wigner-covariance matrix) is defined as
    \begin{equation}
        \vec{\sigma} = \dfrac{\vec{X} + \vec{X}^T}{2} -  \braket{ \vec{\hat{x}} }  \braket{\vec{\hat{x}}^T } .
    \end{equation}
    For a valid quantum state, the covariance matrix obey the following uncertainty condition~\cite{Simon_1994}
    \begin{equation}
          \vec{\sigma} + \frac{i\hbar}{2}\vec{\Omega} \succeq 0.
          \label{eqap_uncert_covar}
    \end{equation} 
    The uncertainty relation can also be written for the matrix $\vec{X}$.
    Indeed, using the commutation relation, $\vec{X}^T = \vec{X} - i\hbar \vec{\Omega}$, we get that $\vec{\sigma} = \vec{X} - \frac{i\hbar}{2} \vec{\Omega} -  \braket{ \vec{\hat{x}} }  \braket{\vec{\hat{x}}^T }$.
    Therefore,
    \begin{equation}
        \vec{X} -  \braket{ \vec{\hat{x}} }  \braket{\vec{\hat{x}}^T }  \succeq 0.
    \end{equation}

   In addition to the quadrature operators, the annihilation $\hat{\xi}_i = \tfrac{1}{\sqrt{2\hbar}}(\hat{q}_i + i \hat{p}_i)$ and creation $\hat{\xi}^\dagger_i = \tfrac{1}{\sqrt{2\hbar}}(\hat{q}_i - i \hat{p}_i)$ operators can be introduced.
   For a compact notation, we define two vectors $\vec{\xi} \equiv  (\hat{\xi}_1, \dots, \hat{\xi}_M; \hat{\xi}^\dagger_1, \dots, \hat{\xi}^\dagger_M)^T$ and $\vec{\xi}^H \equiv  (\hat{\xi}^\dagger_1, \dots, \hat{\xi}^\dagger_M; \hat{\xi}_1, \dots, \hat{\xi}_M) $.
   Then, the bosonic commutation relation has a form  $[\hat{\xi}_i, \hat{\xi}^\dagger_j] = K_{ij}$, where 
    $
        \vec{K} = \left(
                    \begin{smallmatrix}
                     \vec{1}_M & \vec{0} \\
                     \vec{0} & -\vec{1}_M 
                    \end{smallmatrix}
                \right)   
    $. 
    The transformation between quadratures and field operators is $\vec{x} = \sqrt{\hbar}  W \vec{\xi}$, where the unitary matrix
    $ W =   \frac{1}{\sqrt{2}} \left(
                    \begin{smallmatrix}
                      \vec{1}_M   &  \vec{1}_M      \\
                    -i\vec{1}_M   & i\vec{1}_M  
                    \end{smallmatrix}
                        \right)$.
    The uncertainty relation~\eqref{eqap_uncert_covar} can be therefore reformulated in terms of field operators as
    \begin{equation}
        \veccal{M} - \braket{\vec{\xi} } \braket{\vec{\xi}^H} \succeq 0,
    \end{equation} 
    where $\veccal{M} = \braket{\vec{\xi} \cdot \vec{\xi}^H}$ and $\vec{X} = \hbar W \veccal{M} W^H $.

 In contrast to formalisms in which symmetric-ordered moments of the creation and annihilation operators are studied~\cite {Simon_1994}, here, we are interested only in the normal-ordered moments.     For this, we define the normal-ordered $M\times M$-matrix $\veccal{D}$ with elements $[\veccal{D}]_{ij} = \braket{{\hat{\xi}_i^\dagger \hat{\xi}_j} }$ and matrix $\veccal{C}$ with elements $ [\veccal{C}]_{ij} =  \braket{{\hat{\xi}_i \hat{\xi}_j} }$. 
    By definition, the matrix $\veccal{D}$ is Hermitian and the matrix $\veccal{C}$ is symmetric. 
    In terms of $\veccal{D}$ and $\veccal{C}$, matrix $\veccal{M}$ has the form
    \begin{equation}
         \veccal{M} =
        \begin{pmatrix}
            \veccal{D}^T + \vec{1}_M &  \veccal{C}      \\
            \veccal{C}^H     & \veccal{D}  
        \end{pmatrix}.
    \end{equation}

    \subsection{Undisplaced states}
    
    For undisplaced states, i.e., when $\braket{\vec{x}}=0$ and $\braket{\vec{\xi}}=0$, the uncertainty condition takes the form 
    \begin{equation}
        \veccal{M} \succeq 0.
    \end{equation} 
    Taking its Schur complement relative to $\veccal{D}^T + \vec{1}_M$, and using Theorem 1.12 from Ref.~\cite{Schur_complement_book_2005}, the uncertainty condition has an alternative form 
    \begin{equation}
         \veccal{D} - \veccal{C}^H (\veccal{D}^T + 1 )^{-1}  \veccal{C} \succeq 0.
    \end{equation}

\subsection{Multimode bipartite states}
\label{app_bipartitie_Gaussian}

    In this work, we are interested in a multimode bipartite system, consisting of two parts: part $a$ consisting of $N$- modes and part $b$ -- $M$-bosonic modes. 
    Denoting the annihilation operators for subsystem $a$ and subsystem $b$ as vectors $\hat{\vec{a}} \equiv  (\hat{a}_1, \dots, \hat{a}_N)^T$ and $\hat{\vec{b}} \equiv  (\hat{b}_1, \dots, \hat{b}_M)^T$,
    the vector $\vec{\xi}$ from the previous subsection is
    \begin{equation}
        \vec{\xi} \equiv  (\hat{a}_1, ..., \hat{a}_N, \hat{b}_1, ..., \hat{b}_M; \hat{a}^\dagger_1, ..., \hat{a}^\dagger_N, \hat{b}^\dagger_1, ..., \hat{b}^\dagger_M)^T.
    \end{equation}
    Then the matrix $ \veccal{M} $ takes the form
    \begin{equation}
         \veccal{M} \equiv 
        \begin{pmatrix}
            \veccal{D}_{a}^T + \vec{1}_N & \veccal{D}_{ba}^T   & \veccal{C}_{a} &  \veccal{C}_{ab}  \\
            \veccal{D}_{ab}^T  & \veccal{D}_{b}^T + \vec{1}_N & \veccal{C}_{ba} &  \veccal{C}_{b}  \\
              \veccal{C}_{a}^H  &  \veccal{C}_{ba}^H &   \veccal{D}_{a}  & \veccal{D}_{ab}     \\
              \veccal{C}_{ab}^H  & \veccal{C}_{b}^H &  \veccal{D}_{ba}  & \veccal{D}_{b}   \\
         \end{pmatrix},
    \end{equation}
    where the elements of matrices $\veccal{D}_{a}$, $\veccal{D}_{b}$ and $\veccal{D}_{ab}$ are $ \braket{{\hat{a}_i^\dagger \hat{a}_j} }$, $ \braket{{\hat{b}_i^\dagger \hat{b}_j} }$ and $ \braket{{\hat{a}_i^\dagger \hat{b}_j} }$, respectively.
    The elements of matrices $\veccal{C}_{a}$, $\veccal{C}_{b}$ and $\veccal{C}_{ab}$ are $ \braket{{\hat{a}_i \hat{a}_j} }$, $ \braket{{\hat{b}_i \hat{b}_j} }$ and $ \braket{{\hat{a}_i \hat{b}_j} }$, respectively.

    \subsection{GSID-states}
    \label{app_GSID_Gaussian}
        
        GSID-states are the undisplaced bipartite Gaussian states with zero correlation matrices $\veccal{D}_{ab}=0$, $\veccal{C}_{a}=0$, $\veccal{C}_{b}=0$.
        Therefore,
    \begin{equation}
         \veccal{M} \equiv 
        \begin{pmatrix}
            \veccal{D}_{a}^T + \vec{1}_N & 0  & 0 &  \veccal{C}_{ab}  \\
            0 & \veccal{D}_{b}^T + \vec{1}_M & \veccal{C}_{ab}^T &  0 \\
              0 &  \veccal{C}_{ab}^* &   \veccal{D}_{a}  & 0    \\
              \veccal{C}_{ab}^H  &0 & 0  & \veccal{D}_{b}   \\
         \end{pmatrix} \succeq 0 .
         \label{eqap_matrix_M_general}
    \end{equation}
    Matrix $\veccal{M}$ is Hermitian and has the form of Eq.~\eqref{eqap_matrix_form}; therefore, Lemma~\ref{lemma_schur_positivity} gives us the uncertainty condition in the form of Eq.~\eqref{eq_uncertainty_main}.

    Applying Lemma~\ref{lemma_schur_trace}, we get 
    \begin{equation}
        |\trace{ \veccal{C}_{ab}^H \veccal{C}_{ab}  }|^2 \leq \trace{ (\veccal{D}_{b}^H + \vec{1}_M) \veccal{D}_{b} } \trace{ (\veccal{D}_{a}^H + \vec{1}_N) \veccal{D}_{a}  }.
    \end{equation}
    Matrices $\veccal{D}_{a,b}^H = \veccal{D}_{a,b}$, and we immediately get Eq.~\eqref{eq_uncet_inequality}.

\section{Purity of GSID-state}
    \label{app_purity}

    The purity $P = \trace{\hat{\rho}^2}$ for an $(N+M)$-mode Gaussian state $\hat{\rho}$ is expressed in terms of a covariance matrix $\veccal{\sigma}$ as~\cite{Fabre_Treps_2020_review}
    \begin{equation}
        P= \dfrac{(\frac{\hbar}{2})^{N+M}}{\sqrt{\det(\veccal{\sigma})}}.
        \label{eqap_purity_deriv_1}
    \end{equation}
    To write down the purity for GSID-states in terms of matrices $\veccal{D}_a$, $\veccal{D}_b$ and $\veccal{C}_{ab}$, we first connect matrices $\veccal{\sigma}$ and $\veccal{M}$.  
    The structure matrices are related as $ W^H \vec{\Omega} W = -i \vec{K}  $ (or 
         $  \vec{\Omega}  = -i W \vec{K} W^H $).
    Therefore,
    \begin{multline}
         \vec{\sigma} = \vec{X} - \frac{i\hbar}{2} \vec{\Omega}    = \hbar W \veccal{M} W^H - \frac{i\hbar}{2} ( -i W \vec{K} W^H) \\ 
         = \hbar W\left( \veccal{M} - \tfrac12 \vec{K} \right) W^H .
         \label{eqap_connection_diff_matrix}
    \end{multline}
    Taking the determinant of both sides, we get 
    \begin{multline}
         \det(\vec{\sigma}) = \det\left( \hbar W \left( \veccal{M} - \tfrac12 \vec{K} \right) W^H \right) \\  = \hbar^{2(N+M)} \det(W) \det(W^H) \det\left(  \veccal{M} - \tfrac12 \vec{K} \right) \\  = \hbar^{2(N+M)} \det\left(  \veccal{M} - \tfrac12 \vec{K} \right),
            \label{eqap_purity_deriv_2}  
    \end{multline}
    where $\vec{\tfrac12} = \tfrac12\vec{1}_{2(N+M)}$.

    Substituting the matrix $\veccal{M}$ for the GSID-state~\eqref{eqap_matrix_M_general} and using Lemma~\ref{lemma_determinant}, we get
    \begin{multline}
             \det\left(  \veccal{M} - \tfrac12 \vec{K} \right) =
            \det \begin{pmatrix}
                \tilde{\veccal{D}}_a^T & \veccal{C}_{ab}    \\
                \veccal{C}_{ab}^H & \tilde{\veccal{D}_b}   \\
            \end{pmatrix} 
            \det \begin{pmatrix}
                 \tilde{\veccal{D}}_a & \veccal{C}_{ab}^*     \\
                 \veccal{C}_{ab}^T & \tilde{\veccal{D}}_b^T  \\
            \end{pmatrix} 
            \\ = \left[  \det \begin{pmatrix}
                \tilde{\veccal{D}}_a^T & \veccal{C}_{ab}    \\
                \veccal{C}_{ab}^H & \tilde{\veccal{D}_b}   \\
            \end{pmatrix} \right]^2,
            \label{eqap_purity_deriv_3}
    \end{multline}
    where $\tilde{\veccal{D}}_a = \veccal{D}_a + \tfrac12\vec{1}_{N}$ and $\tilde{\veccal{D}_b} = \veccal{D}_b + \tfrac12\vec{1}_{M}$.
    For the last equality, we used $\det(\vec{A}^*) = \det(\vec{A})$, which is valid for any Hermitian matrix $\vec{A}$.
    Combining Eqs.~\eqref{eqap_purity_deriv_1},~\eqref{eqap_purity_deriv_2} and \eqref{eqap_purity_deriv_3}, we get the final expression Eq.~\eqref{eq_purity}.

    Note that, by definition, the covariance matrix  is positive definite ($\vec{\sigma} \succ 0$) ~\cite{Simon_1994,Weedbrook_2012}; therefore, $ \veccal{M} - \tfrac12 \vec{K} \succ 0$.
    As a result, its Schur complements are also positive definite.
    This means that the matrix $\left( \begin{smallmatrix}
                    \tilde{\veccal{D}}_a^T & \veccal{C}_{ab}    \\
                    \veccal{C}_{ab}^H & \tilde{\veccal{D}_b}   \\
                \end{smallmatrix} \right) \succ 0$
    , which will be further used in Appendix~\ref{app_purity_estimation}.

\section{Degree of correlation bound }
\label{app_ineq_degree_corr}
    Taking the trace of the first inequality from the uncertainty relation Eq.~\eqref{eq_uncertainty_main}, we get 
    \begin{equation}
                 0 \leq \mathcal{N}_a - \trace{ \veccal{C}_{ab}(\veccal{D}_b  + \vec{1}_M)^{-1}\veccal{C}^H_{ab} }.
        \label{eq_uncert_type_II_v2}
    \end{equation}
    Using Lemma~\ref{lemma_trace_2}, we have  
    \begin{equation}
       \trace{\veccal{C}_{ab}^H(\veccal{D}_b  + \vec{1}_M)^{-1} \veccal{C}_{ab}} \geq   \dfrac{ \mathcal{J}  }{   \frac{ \mathcal{N}_b   }{ \sqrt{K\mu} }  + 1  }, 
    \end{equation}
    and consequently
    \begin{equation}
        0  \leq \mathcal{N}_a - \dfrac{ \mathcal{J} }{   \frac{ \mathcal{N}_b }  { \sqrt{K\mu_b} }  + 1  }.
    \end{equation}
    Writing the total amount of correlations $\mathcal{J}$ in terms of degree of correlation $\xi$, we get
       \begin{equation}
        0  \leq \mathcal{N}_a - \dfrac{ \xi \sqrt{\mathcal{N}_a + \frac{\mathcal{N}_a^2}{\mu_a} }   \sqrt{\mathcal{N}_b + \frac{\mathcal{N}_b^2}{\mu_b} } }{   \frac{ \mathcal{N}_b }  { \sqrt{K\mu_b} }  + 1  }.
    \end{equation}
    and
    \begin{equation}
        \xi \leq  \dfrac{ \mathcal{N}_a \left(  \frac{ \mathcal{N}_b }  { \sqrt{K\mu_b} }  + 1   \right) }{\sqrt{\mathcal{N}_a + \frac{\mathcal{N}_a^2}{\mu_a} }   \sqrt{\mathcal{N}_b + \frac{\mathcal{N}_b^2}{\mu_b} }}.
    \end{equation}

    Applying the same operations to the second inequality in Eq.~\eqref{eq_uncertainty_main}, and taking into account that both conditions hold, we obtain the resulting inequalities~\eqref{eq_inequality_xi}.

\section{Uncertainty of MW-modes}
\label{app_uncertainty_of_mw_modes}

    In this section, we briefly discuss the properties of the MW-modes for the GSID-state.
    First, let us build matrices
   \begin{equation}
        \vec{\Lambda}_a = \mathrm{diag}(\lambda^a_1,  ...  \lambda^a_N) ~ ~ \text{and} ~ ~ \vec{\Lambda}_b= \mathrm{diag}(\lambda^b_1, ...  \lambda^b_M),
    \end{equation}
    where $\set{\lambda^{a}_n}$ and $\set{\lambda^{b}_m}$ are the eigenvalues of matrices $\veccal{D}_a$ and $\veccal{D}_b$, respectively.
    The eigenvalues are always taken in decreasing order ($\lambda^{a,b}_i \geq \lambda^{a,b}_j\geq 0, \forall j>i$), therefore, the matrices $\vec{\Lambda}_{a,b}$ are  uniquely-defined.
    The eigendecompoisitions of the matrices $\veccal{D}_{a,b}$ are then
    $
         \veccal{D}_a = W_A \vec{\Lambda}_a W_A^H, ~ ~      \veccal{D}_b = W_B \vec{\Lambda}_b W_B^H,
    $ 
    where $W_A$ and $W_B$ are unitary matrices.
    These eigendecompositions define the MW-basis and therefore it has the following properties:
    \begin{enumerate}
        \item The MW-basis is well-defined but not unique. 
        In the case of degenerate eigenvalues, the unitary matrices and consequently the MW-modes are not unique. 
        Indeed, let us consider the matrix $\veccal{D}_a$ that has $M$ different eigenvalues $\lambda_i$ with the degree of degeneracy $k_i$.
        Its eigendecomposition is $\veccal{D}_a = W_A \vec{\Lambda}_a W_A^H$, where
        \begin{equation}
            \vec{\Lambda}_a  = \lambda_1 \vec{1}_{k_1} \oplus \lambda_2 \vec{1}_{k_2} \oplus ... \oplus \lambda_{k_M} \vec{1}_{k_M}.
            \label{eq_lambda_a}
        \end{equation}
        However, any matrix $\tilde{W}_A = W_A V$ is a valid unitary transformation that diagonalizes matrix $\veccal{D}_a$, where $V=\vec{V}_{1} \oplus ... \oplus \vec{V}_{M}$, and $\vec{V}_{i}$ are unitary matrices of the size $k_i\times k_i$.
        Therefore, both unitarz matrices $W_A$ and $\tilde{W}_A $ can be considered as a valid MW-basis for the subsystem $a$.
        
        \item In the case of a non-degenerate spectrum, the matrix $V=\textrm{diag}(e^{i\phi_1}, ..., e^{i\phi_N})$; i.e, every mode is defined up to a constant phase.

        \item Depending on the explicit form of the matrix $\veccal{C}_{AB}$, some of the MW-bases may be more preferable then others.
              To choose the `best' unitary matrices $W_A$ and $W_B$, we can introduce additional criteria requiring the matrix $\veccal{C}_{AB}$ to be close to some particular form.
                As a first example, such an additional criteria exists for the pure states, and we can introduce the unique basis of Schmidt modes (see Sec.~\ref{sec_pure_gtts}).
                As a second example, for the non-degenerate matrices $\veccal{D}_a$ and $\veccal{D}_b$, matrices $W_A$ and $W_B$ can always be chosen in such a way that the diagonal elements of $\veccal{C}_{AB}$ are real (see Refs.~\cite{Kopylov_2025,Kopylov_2025_APL}).
                However, for arbitrary GSID-states, such reasonable criteria can depend on application and their detailed studies are  out of the scope of this paper.
    \end{enumerate}

    As soon as the basis-independent quantities depend only on the eigenvalues, the uncertainty of matrices $W_A$ and $W_B$ does not play any role for them.

\section{Purity bounds}
    \label{app_purity_estimation}

        In this section, we derive lower purity bounds that depend only on the parameters $\mathcal{N}_a$, $\mathcal{N}_b$, $\mu_a$, $\mu_b$, $K$ and $\xi$.
        We start with Eq.~\eqref{eq_purity} and from the Schur's determinant formulas Eq.~\eqref{eq_schur_determinant_1} and~\eqref{eq_schur_determinant_2}, we obtain
        \begin{equation}
             P(\hat{\rho}) =   [\det(\veccal{A}) \det\left( \veccal{B} - 4 \veccal{G}_a  \right)]^{-1} 
             \label{eqap_pur_der_1}
        \end{equation} 
        and
        \begin{equation}
             P(\hat{\rho}) =   [\det(\veccal{B})  \det\left( \veccal{A} - 4 \veccal{G}_b \right)]^{-1}.
                \label{eqap_pur_der_3}
        \end{equation} 
        where, for the compactness, we introduced $\veccal{A} = 2\veccal{D}_a^T + \vec{1}_{N}$, $\veccal{B} = 2\veccal{D}_b   + \vec{1}_{M}$, $\veccal{G}_a = \veccal{C}_{ab}^H \veccal{A}^{-1} \veccal{C}_{ab}$ and 
         $\veccal{G}_b = \veccal{C}_{ab} \veccal{B}^{-1} \veccal{C}_{ab}^H$.
        From Lemma~\ref{lemma_trace_2}, we obtain
        \begin{equation}
            \trace{\veccal{G}_a}  \geq \dfrac{  \sqrt{K\mu_a} \: \mathcal{J} }{  2\mathcal{N}_a   +  \sqrt{K\mu_a}   } \equiv F.
            \label{eqap_pur_der_2}
        \end{equation}
        
        In the following subsections, we derive two purity bounds based on inequalities~\eqref{eqap_det_first_est} and \eqref{eqap_det_second_est}.

\subsection{Strict purity bound}
    \label{app_purity_bound_strict}

        Let us apply inequality~\eqref{eqap_det_first_est} to the determinants in Eq.~\eqref{eqap_pur_der_1}
        \begin{equation}
            \det(\veccal{A}) = \det(2\veccal{D}_a^T + \vec{1}_{N}) \leq \exp( 2\mathcal{N}_a ), \\
         \end{equation} 
        and
        \begin{multline}
            \det\left( \veccal{B} - 4\veccal{G}_a  \right)  \leq \exp\left( 2\mathcal{N}_b - 4\trace{\veccal{G}_a} \right)
            \\ 
            \leq \exp\left( 2\mathcal{N}_b - 4 F \right).
        \end{multline}
        As a result, we get the first expression for the strict purity bound
        \begin{equation}
             P(\hat{\rho}) \leq P^s_1 \equiv \exp( 2\mathcal{N}_a ) \exp\left( 2\mathcal{N}_b - 4 F \right).
        \end{equation}
        Applying inequality~\eqref{eqap_det_first_est} to the determinants in Eq.~\eqref{eqap_pur_der_3}, we get the second expression with the value $P^s_2$.
        Taking the maximal value of $P^s_1$ and $P^s_2$, we get the strict purity bound Eq.~\eqref{eq_purity_bound_strict}.

\subsection{Approximate purity bound}
    \label{app_purity_bound_approx}

        Let us first apply inequality~\eqref{eqap_det_second_est} to the first determinant in Eqs.~\eqref{eqap_pur_der_1}
        \begin{equation}
             \det(\veccal{A}) = \det(2\veccal{D}_a^T + \vec{1}_{N}) \leq \bigg[ \frac{2\mathcal{N}_a}{n_a} +1  \bigg]^{n_a}  , 
        \end{equation}   
        where $n_a$ is the number of nonzero eigenvalues of matrix $\veccal{D}_a$.
        Applying it to the second one, we get
         \begin{multline}
            \det\left( \veccal{B} - 4\veccal{C}_{ab}^H \det(\veccal{A})^{-1} \veccal{C}_{ab}  \right) \\
             \leq  \Bigg[ 1 + \frac{2\mathcal{N}_b}{n_b} - 4 F   \Bigg]^{n_b}, 
        \end{multline}
        where $n_b$ is the number of nonzero eigenvalues of matrix $\veccal{D}_b$.

        Similarly to the derivation of the strict bound, we have to apply inequality~\eqref{eqap_det_second_est} to the determinants in Eq.~\eqref{eqap_pur_der_3}.
        Taking the maximal value from both of them, we get the purity bound in the form of Eq.~\eqref{eq_purity_bound_approx}.

    \subsection{Upper bound for the renormalized purity}
    \label{app_renormalized_purity_bound}

    Here we derive the upper bound for the renormalized purity.
    To get the renormalized purity $R(\hat{\rho})$, we start with Eq.~\eqref{eqap_pur_der_1} 
    \begin{equation}
         P(\hat{\rho})  =   \big[\det(\veccal{A}) \det\left( \veccal{B} - 4\veccal{G}_a  \right)\big]^{-1}
    \end{equation}
    and the purity of the corresponding thermal state 
    \begin{equation}
         P(\hat{\rho}_{th}) = \big[\det(\veccal{A}) \det(\veccal{B})\big]^{-1}.
    \end{equation}

    In the case of low number of photons,
    \begin{align}
        \det(\veccal{A}) & \approx 1 + 2\mathcal{N}_a, \\
        \det(\veccal{B}) & \approx 1 + 2\mathcal{N}_b, \\
        \det\left( \veccal{B} - 4\veccal{G}_a  \right) & \approx 1 + 2\mathcal{N}_b - 4\trace{\veccal{G}_a}.
    \end{align}

    Using these expressions for the renormalized purities, we get 
    \begin{multline}
        R(\hat{\rho}) = \dfrac{P(\hat{\rho})-P(\hat{\rho}_{th})}{1-P(\hat{\rho}_{th})} 
        = \dfrac{P(\hat{\rho})[P(\hat{\rho}_{th})]^{-1} - 1} {[P(\hat{\rho}_{th})]^{-1} - 1}  \\
        = \dfrac{ \frac{1 + 2\mathcal{N}_b}{(1 + 2\mathcal{N}_b - 4\trace{\veccal{G}_a})} - 1} {(1 + 2\mathcal{N}_a)(1 + 2\mathcal{N}_b)  - 1}. 
    \end{multline}
    Using Eq.~\eqref{eqap_pur_der_2}, the desired upper bound is obtained  
    \begin{equation}
        R(\hat{\rho}) \leq \dfrac{ \frac{1 + 2\mathcal{N}_b}{(1 + 2\mathcal{N}_b - 4 F)} - 1} {(1 + 2\mathcal{N}_a)(1 + 2\mathcal{N}_b)  - 1} .
    \end{equation}
   Applying the same procedure for Eq.~\eqref{eqap_pur_der_3}, and taking the lowest value, we get the final expression Eq.~\eqref{eq_renormalized_purity_bound}.

\section{ Normal-ordered characteristic functions for undisplaced states}
 \label{app_norm_ordered_undisplcaed}

\subsection{ General undisplaced states }

        Let us consider an $M$-mode Gaussian state $\hat{\rho}$, described by the vector of field operators  $\vec{\xi} \equiv  (\hat{\xi}_1, \dots, \hat{\xi}_M; \hat{\xi}^\dagger_1, \dots, \hat{\xi}^\dagger_M)^T$ (see Appendix~\ref{app_gaussian_general}).
        For the undisplaced Gaussian state, the normal-ordered characteristic function $\chi(\vec{\beta}, \vec{\beta^*}) $ is Gaussian~\cite{Vogel_Welsch_book} and
        \begin{equation}
            \chi(\vec{\beta}, \vec{\beta^*}) = \exp \left( \vec{\beta}^T \vec{A} \vec{\beta} + \vec{\beta}^T \vec{B} \vec{\beta^*} + (\vec{\beta^*})^T \vec{C} \vec{\beta^*}    \right).
        \end{equation}
        Its derivatives relate to the normal-ordered moments as
        \begin{multline}
           \braket{\prod_{ij} (\hat{\xi}_i^\dagger)^{l_i} \hat{\xi}_j^{k_j} } = (-1)^{\sum_i k_i} \\ \times \prod_{ij} \partial^{l_i}_{\beta_i} \partial^{k_j}_{\beta^*_j} \: \Phi(\vec{\beta}, \vec{\beta^*}) \big|_{\vec{\beta}=\vec{\beta^*}=0} .
        \end{multline}

        Using results from Appendix~\ref{app_derivatives_Gaussian} for second-order moments, we obtain
        that  $A_{ij} = \frac12 \braket{ \hat{\xi}^\dagger_{i} \hat{\xi}^\dagger_{j} }$, $C_{ij}  = \frac12 \braket{  \hat{\xi}_{i} \hat{\xi}_{j} }$ and $B_{ij}  =  - \braket{   \hat{\xi}^\dagger_{i} \hat{\xi}_{j} } $.
        The fourth-order normal-ordered moment is then
        \begin{multline}
             \langle \hat{\xi}^\dagger_{i} \hat{\xi}^\dagger_{j} \hat{\xi}_{k} \hat{\xi}_{l} \rangle  
                =  
                \langle \hat{\xi}_{k} \hat{\xi}_{l} \rangle \langle \hat{\xi}^\dagger_{j} \hat{\xi}^\dagger_{i} \rangle + 
                \langle \hat{\xi}^\dagger_{i} \hat{\xi}_{k} \rangle \langle \hat{\xi}^\dagger_{j} \hat{\xi}_{l} \rangle  \\ +
                \langle \hat{\xi}^\dagger_{j} \hat{\xi}_{k} \rangle \langle \hat{\xi}^\dagger_{i} \hat{\xi}_{l} \rangle .
                \label{eqap_norm_ord_correlator}
        \end{multline}
        Using the photon number operator $\hat{n}_i = \hat{\xi}^\dagger_i \hat{\xi}_i$, its second order correlation function reads as
        \begin{multline}
            \braket{\hat{n}_i \hat{n}_j} =  
            \braket{\hat{\xi}^\dagger_{i} \hat{\xi}^\dagger_{j}\hat{\xi}_{i} \hat{\xi}_{j}} +
            \braket{\hat{\xi}^\dagger_{i} \hat{\xi}_{j}} \delta_{ij} =   \\ 
             \langle \hat{\xi}_{i} \hat{\xi}_{j} \rangle \langle  \hat{\xi}^\dagger_{i} \hat{\xi}^\dagger_{j} \rangle + 
                \langle \hat{\xi}^\dagger_{i} \hat{\xi}_{i} \rangle \langle \hat{\xi}^\dagger_{j} \hat{\xi}_{j} \rangle \\ + 
                \langle \hat{\xi}^\dagger_{j} \hat{\xi}_{i} \rangle \langle \hat{\xi}^\dagger_{i} \hat{\xi}_{j} \rangle +
            \braket{\hat{\xi}^\dagger_{i} \hat{\xi}_{j}} \delta_{ij} .
            \label{eqap_fourth_order_general}
        \end{multline}

\subsection{ Bipartite states }

        In the same way as in Appendix~\ref{app_bipartitie_Gaussian}, let us consider a bipartite system, consisting of two parts: $c$ and $d$.    
        Denoting the annihilation operators for subsystems $c$ and $d$ as vectors $\hat{\vec{c}} \equiv  (\hat{c}_1, \dots, \hat{c}_N)^T$ and $\hat{\vec{d}} \equiv  (\hat{d}_1, \dots, \hat{d}_M)^T$, respectively,  the vector 
            $ 
                \vec{\xi} \equiv  (\hat{c}_1, ..., \hat{c}_N, \hat{d}_1, ..., \hat{d}_N; \hat{c}^\dagger_1, ..., \hat{c}^\dagger_N, \hat{d}^\dagger_1, ..., \hat{d}^\dagger_N)^T.
            $
        Defining the photon number operators for each of the modes, we get $\hat{n}^{c}_i = \hat{c}^\dagger_i \hat{c}_i$ and $\hat{n}^{d}_i = \hat{d}^\dagger_i \hat{d}_i$, and 
        the correlator~\eqref{eqap_fourth_order_general} is then
        \begin{equation}
            \braket{\hat{n}^c_i \hat{n}^d_j} =  
             |\langle \hat{c}_{i} \hat{d}_{j} \rangle |^2 + |\langle \hat{c}^\dagger_{i} \hat{d}_{j} \rangle |^2   +  \langle \hat{n}^{c}_i\rangle \langle \hat{n}^{d}_i \rangle . 
             \label{eqap_bipart_correlator}
        \end{equation}

\section{ The normalized second-order-correlation function with frequency-unresolved detectors }
\label{app_sec_ord_corr_func}

    Let us consider two detectors that measure modes $c$ and $d$ of the bipartite system presented in previous section.
    Let us define a measured normalized second-order correlation function between the detectors $c$ and $d$ 
    \begin{equation}
          G_{(c,d)} = \dfrac{\braket{\hat{N}_c \hat{N}_d}}{\braket{\hat{N}_c} \braket{\hat{N}_d}}.
    \end{equation}
    The operators $\hat{N}_c = \sum_i \hat{c}^\dagger_i \hat{c}_i$ and $\hat{N}_d=\sum_j \hat{d}^\dagger_j \hat{d}_j$ are the total number of photons operators for each subsystem.
    The correlator $\braket{\hat{N}_c \hat{N}_d} = \sum_{ij}\braket{\hat{c}^\dagger_i \hat{c}_i \hat{d}^\dagger_j \hat{d}_j}$, we get 
    \begin{equation}
       G_{(c,d)} = 1 + \dfrac{ \sum_{ij} |\braket{\hat{c}_{i} \hat{d}_{j}}  |^2 }{N_c N_d} + \dfrac{\sum_{ij} |\braket{\hat{c}^\dagger_{i} \hat{d}_{j} } |^2  }{N_c N_d},
       \label{eqap_sec_ord_measuremet_geenral}
    \end{equation}
    where we used Eq.~\eqref{eqap_bipart_correlator}.
    This general expression is valid for any undisplaced state, and below we apply it to the GSID-state.

\subsection{ Signal-idler second-order correlation function for GSID-state }

    Let us take the GSID-state defined by the matrices $\veccal{D}_a$,  $\veccal{D}_b$ and $\veccal{C}_{ab}$, and send the signal and the idler fields to the detectors $c$ and $d$, respectively.
    Then, $\braket{\hat{c}_{i} \hat{d}_{j}} = \braket{\hat{a}_{i} \hat{b}_{j}}$ and $\braket{\hat{c}^\dagger_{i} \hat{d}_{j} } = \braket{\hat{a}^\dagger_{i} \hat{b}_{j} } = 0$, and we obtain  
    \begin{equation}
       G_{(c,d)} = 1 + \dfrac{ \sum_{ij} | \braket{\hat{a}_{i} \hat{b}_{j}}   |^2 }{\sum_i\braket{\hat{a}_{i} \hat{a}_{i}} \sum_j \braket{\hat{b}_{j} \hat{b}_{j}}}.
    \end{equation}

    The normal ordered signal-idler second-order correlation function is defined as 
    \begin{equation}
        g^{(1,1)} = \dfrac{\braket{:\hat{N}_a \hat{N}_b:}}{\braket{\hat{N}_a} \braket{\hat{N}_b}}.
    \end{equation} 
    From Eq,~\eqref{eqap_norm_ord_correlator}, we see that $g^{(1,1)} = G_{(c,d)}$.

\subsection{ Signal-signal second-order correlation function for GSID-state }

    Let us send a signal field of a GSID-state to the first port of a 50:50 beamsplitter, and the resulting field from the first and second output ports to detectors $c$ and $d$, respectively.
    The measured modes are then $\hat{c}_i = \sqrt{\frac12}(\hat{a}_i + \hat{f}_i)$,  $\hat{d}_i = \sqrt{\frac12}(\hat{a}_i - \hat{f}_i)$, where $\hat{f}_i$ is the annihilation operator from the input of the second beamsplitter port.
    If the initial state of this input is vacuum, $\braket{\hat{c}_{i} \hat{d}_{j} } = \braket{\hat{a}_{i} \hat{a}_{j} } = 0$ and $\braket{\hat{c}^\dagger_{i} \hat{d}_{j}} =  \braket{\hat{a}^\dagger_{i} \hat{a}_{j}} $.
    Then, the measured second-order correlation function~\eqref{eqap_sec_ord_measuremet_geenral} is
    \begin{equation}
       G_{(c,d)} = 1 + \dfrac{ \sum_{ij} | \braket{\hat{a}^\dagger_{i} \hat{a}_{j}}   |^2 }{\big[  \sum_i\braket{\hat{a}_{i} \hat{a}_{i}}  \big] } .
    \end{equation}

    Note that in our derivations, we have started with the two detectors in channels $c$ and $d$ and have considered their fourth-order correlation function.
    In order to express it in terms of the normal-ordered second-order correlation function of the initial signal field
    \begin{equation}
          g^{(2)}_{a} = \dfrac{\braket{:\hat{N}^2_a:}}{\braket{\hat{N}_a}^2} = 
          \dfrac{\sum_{ij}\braket{\hat{a}^\dagger_{i} \hat{a}^\dagger_{j} \hat{a}_{i} \hat{a}_{j}} }{\big[\sum_{i}\braket{\hat{a}^\dagger_{i} \hat{a}_{i} }\big]^2},
    \end{equation}
    we should use Eq.~\eqref{eqap_norm_ord_correlator}.
    Indeed,      
    \begin{equation}
         \langle \hat{a}^\dagger_{i} \hat{a}^\dagger_{j} \hat{a}_{i} \hat{a}_{j} \rangle  
            =  
            |\langle \hat{a}^\dagger_{i} \hat{a}_{j} \rangle|^2 + 
            \langle \hat{a}^\dagger_{i} \hat{a}_{i} \rangle \langle \hat{a}^\dagger_{j} \hat{a}_{j} \rangle  ,
    \end{equation}
    and therefore $g^{(2)}_{a} = G_{(c,d)}$.

    Similarly, the expression of the normal second-order correlation function for the idler field is obtained.

\section{Proof of Theorem 1: bipartite Schmidt modes}
    \label{app_theorem_proof}

A proof of Theorem~\ref{theorem_pure_state} consist of four steps: 
    (1) We start with a bipartite system whose subsystems have the same number of modes. 
    Considering a GSID-state $\hat{\rho}$ that satisfies the minimal uncertainty relation, we find a connection between its correlation matrices;
    (2) We define the unique basis of Schmidt modes;
    (3) We show that the minimal uncertainty states are pure.  
    (4) We extend the obtained results to bipartite systems with arbitrary dimensions.

    \textbf{Step 1.}
    Let us assume that the most minimal form of the uncertainty condition~\eqref{eq_uncertainty_main} is
        \begin{equation}
            \begin{cases}
             \veccal{D}^T_a - \veccal{C}_{ab}(\veccal{D}_b  + 1)^{-1}\veccal{C}^H_{ab} & = 0, \\
                     \veccal{D}_b - \veccal{C}_{ab}^H (\veccal{D}_a^T  + 1)^{-1}\veccal{C}_{ab} & = 0.
            \end{cases}
            \label{eq_app_uncertainty_general}
        \end{equation}
    Here we manually replaced the symbol `$\succeq$' with `$=$', meaning that the matrices on the left-hand side have only zero eigenvalues. 
    In other words, these matrices are zero.
    This replacement does not guarantee that the solution of Eq.~\eqref{eq_app_uncertainty_general} exists. 
    However, if it does, the uncertainty condition~\eqref{eq_uncertainty_main} is the minimal uncertainty condition, and it defines the minimal uncertainty states.         

    To show that Eq.~\eqref{eq_app_uncertainty_general} is resolvable, let us start with an $(N+N)$-mode system and a GSID-state $\hat{\rho}$ that is given by the correlation matrices $\veccal{D}_a, \veccal{D}_b, \veccal{C}_{ab} \in \mathbb{C}^{N\times N}$. 
    Here, we assume that all $N$ eigenvalues of matrix $\veccal{D}_A$ are nonzero.  
    In the MW-basis (see Sec.~\ref{ssub_mw_basis}), matrices $\veccal{D}_{A,B}$ are diagonal: $\veccal{D}_A = \vec{\Lambda}_a$ and $\veccal{D}_B = \vec{\Lambda}_b$.
    It follows from Eq.~\eqref{eq_app_uncertainty_general} that the minimal uncertainty relation should take the form
        \begin{align}
            \begin{cases}
                \vec{\Lambda}_a - \veccal{C}_{AB}(\vec{\Lambda}_b  + 1)^{-1}\veccal{C}^H_{AB} & = 0, \\
                \vec{\Lambda}_b - \veccal{C}_{AB}^H (\vec{\Lambda}_a  + 1)^{-1}\veccal{C}_{AB} &  = 0.
            \end{cases}
            \label{eq_app_minuncert_diag}
        \end{align}
    The resolvability of this system is implied by the following lemma.
    Therefore, the  \textit{minimal uncertainty states} obeying the system~\eqref{eq_app_minuncert_diag} exist. 
        
    \begin{lemma}
        \label{lemma_uncert_resolv}
        \begin{equation}
            \vec{\Lambda}_a  = \lambda_1 \vec{1}_{k_1} \oplus \lambda_2 \vec{1}_{k_2} \oplus ... \oplus \lambda_{k_M} \vec{1}_{k_M},
            \label{eq_app_lambda_a}
        \end{equation}
        where $\lambda_i > \lambda_j > 0,  \forall i<j$, and $k_i \in \mathbb{N}$,
         the system~\eqref{eq_app_minuncert_diag} is resolvable if and only if
        \begin{equation}
            \vec{\Lambda}_b = \vec{\Lambda}_a \equiv \vec{\Lambda} 
        \end{equation}
        and 
        \begin{equation}
            \veccal{C}_{AB} = (\vec{\Lambda} + \vec{\Lambda}^2)^{\frac12} V ,
            \label{eq_app_res_c_theorem}
        \end{equation}
        where $V = {V}_1 \oplus {V}_2 \oplus ... \oplus {V}_M$ and ${V}_i$ are unitary matrices of the size $k_i \times k_i$.
            
        \end{lemma}
                
                \begin{proof}

                The first equation in the system~\eqref{eq_app_minuncert_diag} has a solution if $\mathrm{rank}(\vec{\Lambda}_a) = \mathrm{rank}(\veccal{C}_{AB}(\vec{\Lambda}_b  + 1)^{-1} \veccal{C}^H_{AB})$. 
                In turn, $ \mathrm{rank}(\veccal{C}_{AB}(\vec{\Lambda}_b  + 1)^{-1} \veccal{C}^H_{AB})\leq \mathrm{min}[\mathrm{rank}(\veccal{C}_{AB}),N,\mathrm{rank}(\veccal{C}^H_{AB}) ]$, and an equality holds if $\mathrm{rank}(\veccal{C}^H_{AB}) = N$ and therefore $\veccal{C}_{AB}$ has to be invertible.
                Then, we can extract the matrices $\veccal{C}_{AB} $ and $\veccal{C}_{AB}^H $ from the first and second equations \eqref{eq_app_minuncert_diag}
                \begin{subequations}
                    \begin{align}
                            \label{eq_app_diag_step_1a}                
                            \veccal{C}_{AB} = \vec{\Lambda}_a(\veccal{C}^H_{AB})^{-1}(\vec{\Lambda}_b  + 1),  \\
                            \label{eq_app_diag_step_1b}                    
                            \veccal{C}_{AB}^H = \vec{\Lambda}_b (\veccal{C}_{AB})^{-1} (\vec{\Lambda}_a  + 1).
                    \end{align}
                \end{subequations}
                Substituting Eqs.~\eqref{eq_app_diag_step_1a} and \eqref{eq_app_diag_step_1b} into the second and the first equations, respectively, we obtain
                \begin{subequations}
                    \begin{align}
                            \label{eq_app_diag_step_2a}                 
                            \vec{\Lambda}_b(\vec{\Lambda}_b  + 1)^{-1}  & =  \veccal{C}_{AB}^H (\vec{\Lambda}_a  + 1)^{-1} \vec{\Lambda}_a(\veccal{C}^H_{AB})^{-1},  \\ 
                            \label{eq_app_diag_step_2b}    
                            \vec{\Lambda}_a (\vec{\Lambda}_a  + 1)^{-1}  & =  \veccal{C}_{AB}(\vec{\Lambda}_b  + 1)^{-1}\vec{\Lambda}_b (\veccal{C}_{AB})^{-1}.
                    \end{align}
                \end{subequations}
                Matrices $\vec{\Lambda}_{a,b}$ are diagonal; therefore,  $\vec{\Lambda}_{a,b}(\vec{\Lambda}_{a,b}  + 1)^{-1} = (\vec{\Lambda}_{a,b}  + 1)^{-1} \vec{\Lambda}_{a,b}$ and we can rewrite the obtained relation as
                \begin{subequations}
                    \begin{align}
                            \label{eq_app_diag_step_3a}                 
                            M_b  & =  \veccal{C}_{AB}^H M_a (\veccal{C}^H_{AB})^{-1},  \\
                            \label{eq_app_diag_step_3b}    
                            M_a  & =  \veccal{C}_{AB} M_b (\veccal{C}_{AB})^{-1}.
                    \end{align}
                \end{subequations}
                where $M_{a,b} \equiv \vec{\Lambda}_{a,b}(\vec{\Lambda}_{a,b}  + 1)^{-1}$.
                Substituting Eq.~\eqref{eq_app_diag_step_3a} into Eq.~\eqref{eq_app_diag_step_3b} and Eq.~\eqref{eq_app_diag_step_3b} into Eq.~\eqref{eq_app_diag_step_3a}, we get
                \begin{subequations}
                     \begin{align}
                            \label{eq_app_diag_step_4a}
                            M_a  & =  \veccal{C}_{AB} \veccal{C}_{AB}^H M_a (\veccal{C}_{AB} \veccal{C}^H_{AB})^{-1},  \\ 
                            \label{eq_app_diag_step_4b}
                            M_b  & =  \veccal{C}_{AB}^H \veccal{C}_{AB} M_b  (\veccal{C}^H_{AB}\veccal{C}_{AB})^{-1}.  
                    \end{align}
                \end{subequations}
                Using the explicit form of the matrix $\vec{\Lambda}_a = \lambda_1 \vec{1}_{k_1} \oplus \lambda_2 \vec{1}_{k_2} \oplus ... \oplus \lambda_{k_M} \vec{1}_{k_M} $ in Eq.~\eqref{eq_app_diag_step_4a}, we get that the matrix $\veccal{C}_{AB}^H \veccal{C}_{AB}$ has a block-form
                \begin{equation}
                    \veccal{C}_{AB}^H \veccal{C}_{AB} =  \tilde{C}_1 \oplus  \tilde{C}_2 \oplus ... \oplus \tilde{C}_M , 
                \end{equation} 
                where $\tilde{C}_i$ are invertible matrices of the size $k_i\times k_i$.
                In turn, condition \eqref{eq_app_diag_step_4b}, implies that the matrix $\vec{\Lambda}_b$ should have the same structure of the degenerate values.
                Together with Eq.~\eqref{eq_app_diag_step_3a} and \eqref{eq_app_diag_step_3b}, it means that the matrix $\veccal{C}_{AB}$ itself has the same block structure as $ \veccal{C}_{AB}^H \veccal{C}_{AB}$, namely
                \begin{equation}
                    \veccal{C}_{AB} =  C_1 \oplus C_2 \oplus ... \oplus  C_M , 
                    \label{eq_app_cmatrix_general_form}
                \end{equation} 
                where the matrix $C_i$ is an invertible matrix such as $C^H_i C_i = \tilde{C}_i$.
                Thus, Eqs.~\eqref{eq_app_diag_step_3a} and \eqref{eq_app_diag_step_3b} imply that $M_a = M_b$ and therefore 
                \begin{equation}
                    \vec{\Lambda}_b=\vec{\Lambda}_a \equiv \vec{\Lambda} .
                    \label{eq_app_dmatrix_general_form}
                \end{equation} 

                In order to find the matrices $C_i$, we substitute Eqs.~\eqref{eq_app_cmatrix_general_form} and \eqref{eq_app_dmatrix_general_form} into Eq.~\eqref{eq_app_minuncert_diag}. 
                The resulting equations have a block structure, and each block corresponds to a different eigenvalue. Therefore, we can solve each block independently.                    
                \begin{align}
                        \begin{cases}
                            \lambda_i \vec{1}_{k_i}  =  C_i(\lambda_i \vec{1}_{k_i}  + 1)^{-1} C_i^H , \\
                            \lambda_i \vec{1}_{k_i}  = C_i^H (\lambda_i \vec{1}_{k_i}  + 1)^{-1} C_i .
                        \end{cases}
                    \end{align}    
                The matrices commute, and their order can be changed, resulting in
                \begin{align}
                    \begin{cases}
                        \lambda_i(\lambda_i + 1)\vec{1}_{k_i}  = C_i C_i^H , \\
                        \lambda_i(\lambda_i + 1)\vec{1}_{k_i}  = C_i^H C_i .
                    \end{cases}
                \end{align}    
                As a result, matrix $C_i$ can be written as 
                \begin{equation}
                     C_i = \sqrt{\lambda_i(\lambda_i + 1)} V_i, 
                \end{equation} 
                where $V_i$ is a unitary matrix of dimension $k_i \times k_i$, completing the proof of Eq.~\ref{eq_app_res_c_theorem}.

                Note, that due to the block-form, the matrices $\vec{\Lambda}$ and $V$ commute: $\vec{\Lambda} V = V \vec{\Lambda}$.
            \end{proof}

    \textbf{Step 2.}

    Let us consider a minimal uncertainty state $\hat{\rho}$ that is given in a monochromatic basis by matrices $\veccal{D}_a$, $\veccal{D}_b$ and $\veccal{C}_{ab}$.
    According to Lemma~\ref{lemma_uncert_resolv}, the eigendecompositions are
    $ 
         \veccal{D}_a = W_A \vec{\Lambda} W_A^H, ~  \veccal{D}_b = \tilde{W}_B \vec{\Lambda} \tilde{W}_B^H,
    $ 
    and, in the MW-basis, we have
    \begin{equation}
        \veccal{D}_A = \veccal{D}_B = \vec{\Lambda}.
    \end{equation}
    and 
    \begin{equation}
        \tilde{\veccal{C}}_{AB} = W_A^T \veccal{C}_{ab} \tilde{W}_B \\ = (\vec{\Lambda} + \vec{\Lambda}^2)^{\frac12} V ,
    \end{equation}
    where $V$ is unitary.

    As was mentioned in Appendix~\ref{app_uncertainty_of_mw_modes}, the MW-basis is not unique and instead of matrix $\tilde{W}_B$ we can take the matrix $W_B = \tilde{W}_B V^H$. 
    Indeed,
    \begin{multline}
        \veccal{D}_b = W_B  \vec{\Lambda} W_B^H = (\tilde{W}_B V) \vec{\Lambda}_b (\tilde{W}_B V)^H \\ =  \tilde{W}_B   \vec{\Lambda}  V V^H \tilde{W}_B^H =  \tilde{W}_B \vec{\Lambda} \tilde{W}_B^H.
    \end{multline}
    Using the matrix $W_B$ instead of  $\tilde{W}_B^H$ for the MW-basis construction, we get
    \begin{multline}
        \veccal{C}_{AB} = W_A^T \veccal{C}_{ab} W_B = W_A^T \veccal{C}_{ab} \tilde{W}_B V^H = \\ \tilde{\veccal{C}}_{AB} V^H   = (\vec{\Lambda} + \vec{\Lambda}^2)^{\frac12}.
    \end{multline}

    Therefore, for the minimal uncertainty states, we can always find unitary matrices $W_A$ and $W_B$ that transform the correlation matrices into the diagonal form
    \begin{equation}
        \veccal{D}_A = \veccal{D}_B = \vec{\Lambda} ~ \text{and} ~ \veccal{C}_{AB} = (\vec{\Lambda} + \vec{\Lambda}^2)^{\frac12}.
        \label{eq_app_intermed_bipartmodes}
    \end{equation}
    The basis constructed with the matrices $W_A$ and $W_B$ is called the basis of bipartite Schmidt modes.
 
    \textbf{Step 3.}
    
    To show that the minimal uncertainty states are pure, let us write purity Eq.~\eqref{eq_purity} with Schur's determinant formula Eq.~\eqref{eq_schur_determinant_1} as 
    \begin{multline}
         P^{-1} =   \det(2\veccal{D}_a^T + \vec{1}_{N}) \\ \times \det\left( 2\veccal{D}_b + \vec{1}_{M} - 4\veccal{C}_{ab}^H (2\veccal{D}_a^T + \vec{1}_{N})^{-1} \veccal{C}_{ab}  \right).
    \end{multline} 
    The purity is basis independent and we can use any desired basis.
    As soon as we consider the minimal uncertainty state, we can take its correlation matrices in the bipartite Schmidt basis~\eqref{eq_app_intermed_bipartmodes}.
    Then,
    \begin{multline}
        P^{-1} =   \det(2\vec{\Lambda} + \vec{1}_{N}) \\ \times \det\left( 2\vec{\Lambda} + \vec{1}_{N} - 4(\vec{\Lambda} + \vec{\Lambda}^2)^{\frac12} (2\vec{\Lambda} + \vec{1}_{N})^{-1} (\vec{\Lambda} + \vec{\Lambda}^2)^{\frac12}  \right).
    \end{multline} 
    All the matrices in a second determinant are diagonal; therefore, they commute.
    After some algebra, we get
    \begin{equation}
         P^{-1} =  \det(2\vec{\Lambda} + \vec{1}_{N}) \det\left((2\vec{\Lambda} + \vec{1}_{N})^{-1}\right) \equiv 1.
    \end{equation}
    This result shows that the states with the minimal uncertainty are equivalent to pure states.

    \textbf{Step 4.}

    Here we show that all the aforementioned results are valid for bipartite systems with an arbitrary size.
    Let us consider an $(N+M)$-mode system and a GSID-state $\hat{\rho}$ with matrices $\veccal{D}_a \in \mathbb{C}^{N\times N}$, $\veccal{D}_b \in \mathbb{C}^{M\times M}$ and $\veccal{C}_{ab} \in \mathbb{C}^{N\times M}$.
    Let us assume that the matrices $\veccal{D}_a$ and $\veccal{D}_b$ have $n \leq N$ and $m \leq M$ nonzero eigenvalues, respectively. 
    Building diagonal matrices $\vec{\Lambda}_a$ and $\vec{\Lambda}_b$ that contain only these nonzero eigenvalues, we get that, in the MW-basis, the correlation matrices can be written as
    $\veccal{D}_A = \vec{\Lambda}_a \oplus \vec{0}_{N-n}$ and $\veccal{D}_B = \vec{\Lambda}_b\oplus \vec{0}_{M-m}$.
    Due to Eq.~\eqref{eq_uncet_inequality_elements}, the matrix $\veccal{C}_{AB}$ has the following block-form in the MW-basis
    \begin{equation}
        \veccal{C}_{AB} = \begin{pmatrix}
                            \tilde{\veccal{C}}_{AB}     &  \vec{0}_{M-m}       \\
                            \vec{0}_{N-n}  & \vec{0}_{N-n, M-m}       \\
                        \end{pmatrix}.
    \end{equation}

    This implies that $\rank{\veccal{D}_A}= n$, $\rank{\veccal{D}_B} = m$ and $\rank{\veccal{C}_{AB}} \leq \mathrm{min}[n, m]$.
    In the case of $n\neq m$, $\rank{\veccal{D}_A} \neq \rank{\veccal{D}_B}$ and the system~\eqref{eq_app_uncertainty_general} is not resolvable.
    Therefore, it can be solved only for $m=n$ 
    and we can consider only the case of the squared $n\times n$ sub-matrices of the full matrices $\veccal{D}_A$, $\veccal{D}_B$ and $\veccal{C}_{AB}$ and apply to them the proof from Step 1.

\section{ Bound for the number of correlated modes }
    \label{app_JSI_inequality}

    Let us start with the SVD of a complex matrix $\vec{F} = U_F\vec{\Sigma}_FV_F$.
    Introducing $Q = \trace{\vec{\Sigma}_F^4}$ and applying Eq.~\eqref{eq_trace_indent}, we can write it as
    \begin{multline}
        Q = \trace{\vec{\Sigma}_F^4}  = \trace{(\vec{F}^H \vec{F})^H \vec{F}^H \vec{F}} \\ = \sum_{ij}  ([(\vec{F}^H \vec{F})]_{ij})^* [(\vec{F}^H \vec{F})]_{ij}     \\  = 
        \sum_{ij} \left(\sum_k  f_{ki} f^*_{kj}\right) \left(\sum_m f^*_{mi} f_{mj}\right),
        \label{eq_fake_intermed}
    \end{multline}
    where we use $f_{ij} \equiv [\vec{F}]_{ij}$.
     In turn, $Q$ is real and $Q>0$. Therefore, $Q = |Q|$ and
    \begin{multline}
        \sum_{ijkm}  f_{ki} f^*_{kj} f^*_{mi} f_{mj} = \Big| \sum_{ijkm}  f_{ki} f^*_{kj} f^*_{mi} f_{mj} \Big| \\ \leq  \sum_{ijkm}   \big|f_{ki} f^*_{kj} f^*_{mi} f_{mj} \big | \leq \sum_{ijkm}    |f_{ki}| |f_{kj}| |f_{mi}| |f_{mj}| .
    \end{multline}
    
    Let us construct a real matrix $\vec{A}$ such that $ a_{ij} = |f_{ij}| $.
    Its SVD is $\vec{A} = U_A\vec{\Sigma}_A V_A$.
    From Eq.~\eqref{eq_trace_indent}, we get that 
    \begin{multline}
        \trace{\vec{\Sigma}_A^2} = \trace{\vec{A}^T \vec{A}} = \sum_{ij} a_{ij}^2 = \sum_{ij} |f_{ij}|^2   \\ =  \trace{\vec{F}^H \vec{F}} =  \trace{\vec{\Sigma}_F^2}.
    \end{multline}
    In turn, applying same steps as in Eq.~\eqref{eq_fake_intermed}, we get 
    \begin{multline}
        \trace{\vec{\Sigma}_A^4} = \sum_{ijkm} a_{ki} a_{kj} a_{mi} a_{mj} \\ = \sum_{ijkm}    |f_{ki}| |f_{kj}| |f_{mi}| |f_{mj}| .
    \end{multline}
    Therefore,
    \begin{equation}
        \trace{\vec{\Sigma}_F^4} \leq \trace{\vec{\Sigma}_A^4},
    \end{equation}
    and
    \begin{equation}
        \dfrac{\trace{\vec{\Sigma}_F^4}}{\trace{\vec{\Sigma}_F^2}} \leq \dfrac{\trace{\vec{\Sigma}_A^4}}{\trace{\vec{\Sigma}_A^2}}.
    \end{equation}

\sloppy

\bibliography{literature.bib}

\end{document}